\documentclass[10pt]{article}
\usepackage[a4paper,margin=0.8in]{geometry}
\usepackage{fancyhdr}
\usepackage{microtype}
\usepackage[small]{titlesec}
\usepackage[T1]{fontenc}
\usepackage[utf8]{inputenc}
\usepackage{tabularx}
\usepackage{duckuments}
\usepackage{amssymb,amsmath,amsthm,amsfonts}
\usepackage{bm}             %
\usepackage[normalem]{ulem}  %
\usepackage{hyperref}
\usepackage{pifont}
\usepackage{mathtools}
\usepackage{dsfont}
\newtheorem{theorem}{Theorem}
\usepackage[table]{xcolor}
\usepackage{tikz}
\usepackage[version=4]{mhchem}
\usepackage{multirow}
\usepackage[exponent-product=\cdot]{siunitx}
\usepackage{booktabs}
\usepackage{xspace}
\usepackage{subcaption}
\usepackage{cleveref}
\usepackage{float}
\usepackage{enumitem}
\usepackage[all]{nowidow}
\usepackage{lineno}
\usepackage{jabbrv}
\usepackage{pdflscape}

\DefineSpuriousJournalWord{The}

\usepackage[natbib, autocite=superscript, style=nature, backend=biber, defernumbers=true
]{biblatex}
\let\cite\autocite%
\let\citep\autocite%
\let\citenum\parencite%
\renewcommand{\citet}[1]{\citeauthor{#1}\,\supercite{#1}}

\usepackage{tocloft}
\usepackage{etoc}

\usepackage{parskip}

\usepackage[
    justification=justified,
    singlelinecheck=true,
]{caption}
\usetikzlibrary{arrows.meta}
\usetikzlibrary{calc}
\graphicspath{{figures/intro/}}

\DeclareUnicodeCharacter{2212}{--}
\DeclareUnicodeCharacter{03C9}{\(\omega\)}
\DeclareUnicodeCharacter{0394}{\(\Delta\)}
\DeclareUnicodeCharacter{03C3}{\(\sigma\)}

\usepackage{graphicx}

\definecolor{msft_teal}{rgb}{0.286, 0.773, 0.694}
\definecolor{msft_light_teal}{rgb}{0.725, 0.863, 0.824}
\definecolor{msft_dark_teal}{rgb}{0.133, 0.357, 0.384}
\definecolor{msft_yellow}{rgb}{1.0,   0.725, 0.0}
\definecolor{msft_blue}{rgb}{0.0,   0.471, 0.831}
\definecolor{msft_purple}{rgb}{0.525, 0.380, 0.773}
\definecolor{msft_orange}{rgb}{1.0,   0.639, 0.545}
\definecolor{msft_white}{rgb}{1.0,   1.0,   1.0}
\definecolor{msft_black}{rgb}{0.0,   0.0,   0.0}
\definecolor{msft_grey}{rgb}{0.851, 0.851, 0.839}

\definecolor{ibm_yellow}{RGB}{255, 176, 0}
\definecolor{ibm_orange}{RGB}{254, 97, 0}
\definecolor{ibm_pink}{RGB}{220, 38, 127}
\definecolor{ibm_purple}{RGB}{120, 94, 240}
\definecolor{ibm_blue}{RGB}{100, 143, 255}

\definecolor{arrowgray}{rgb}{0.665, 0.665, 0.665}
\definecolor{circlebg}{rgb}{0.965, 0.965, 0.965}

\definecolor{tab10_blue}{rgb}{0.121, 0.466, 0.705}
\definecolor{tab10_orange}{rgb}{1.0,   0.498, 0.054}
\definecolor{tab10_green}{rgb}{0.172, 0.627, 0.172}
\definecolor{tab10_red}{rgb}{0.839, 0.152, 0.156}
\definecolor{tab10_purple}{rgb}{0.580, 0.403, 0.741}
\definecolor{tab10_brown}{rgb}{0.549, 0.337, 0.294}
\definecolor{tab10_pink}{rgb}{0.890, 0.466, 0.760}
\definecolor{tab10_gray}{rgb}{0.498, 0.498, 0.498}
\definecolor{tab10_olive}{rgb}{0.737, 0.741, 0.133}
\definecolor{tab10_cyan}{rgb}{0.090, 0.745, 0.811}

\def\commentType{0}
\usepackage{xspace}

\ifnum\commentType=0
  \usepackage[disable]{todonotes} %

\fi
\ifnum\commentType=1

\fi
\ifnum\commentType=2
  \usepackage{pdfcomment}

\fi
\ifnum\commentType=3
  \usepackage{todonotes} %
  \usepackage{silence}
\fi

\newcommand*{\clippy}{Skala\xspace}
\newcommand*{\clippyOnePointZero}{Skala\kern.10em‑1.0\xspace}
\newcommand*{\clippyOnePointOne}{Skala\kern.10em-1.1\xspace}

\newcommand*{\dens}{\rho}               %
\newcommand*{\oneRDM}{\Gamma}           %
\newcommand{\params}{\mathbf{\theta}}   %

\newcommand*{\cor}{\text{c}}            %
\newcommand*{\xc}{\text{xc}}            %
\newcommand{\Exc}{E_{\xc}}              %

\DeclareMathOperator{\grad}{\boldsymbol{\nabla}}

\DeclarePairedDelimiter{\abs}{\lvert}{\rvert}
\newcommand{\E}{\mathbb{E}}     

\DeclarePairedDelimiter{\norm}{\lVert}{\rVert}  %

\newcommand{\R}{\mathbb{R}}        %

\newcommand*{\ud}{\mathrm{d}}           %
\newcommand*{\isDefinedAs}{\coloneqq}   %

\DeclareSIUnit\hartree{\text {\ensuremath {E}}_{\mathrm {h}}}
\DeclareSIUnit\planckbar{\text {\ensuremath {\hbar }}}
\DeclareSIUnit\bohr{\text {\ensuremath {a}}_{0}}

\makeatletter
\newcommand\shouldcite[1][]{\def\@shouldcite@arg{#1}\textcolor{tab10_red}{[needs citation\if\relax\detokenize{#1}\relax\else: #1\fi]}\xspace}
\makeatother

\newcommand*{\tablefontsize}{\scriptsize}

\providecolor{gold_medal}{HTML}{FFD700}
\providecolor{gold_medal_dark}{HTML}{B8960A}
\providecolor{silver_medal}{HTML}{C0C0C0}
\providecolor{silver_medal_dark}{HTML}{808080}
\providecolor{bronze_medal}{HTML}{CD7F32}
\providecolor{bronze_medal_dark}{HTML}{8B5521}
\DeclareRobustCommand{\goldmedal}{\tikz[baseline=-0.6ex]{\fill[gold_medal] circle(2.8pt);}}
\DeclareRobustCommand{\silvermedal}{\tikz[baseline=-0.6ex]{\fill[silver_medal] circle(2.8pt);}}
\DeclareRobustCommand{\bronzemedal}{\tikz[baseline=-0.6ex]{\fill[bronze_medal] circle(2.8pt);}}

\makeatletter
\newif\if@firstpageoftable
\newcommand{\multipagetable}[3]%
{\bgroup
  \def\@captype{table}%
  \setbox2=\vbox{\caption{#1}}%
  \setbox3=\vbox{\addtocounter{table}{-1}\caption{#2}}%
  \setbox0=\vbox{#3}%
  \@firstpageoftabletrue
  \loop\ifvoid0\else    
    \if@firstpageoftable
      \@firstpageoftablefalse
      \dimen0=\dimexpr \pagegoal-\pagetotal-\textfloatsep\relax
      \ifdim\dimen0>\bottomfraction\textheight
        \dimen0=\bottomfraction\textheight
      \fi
      \setbox1=\copy0%
      \setbox4=\vsplit1 to 0pt
      \setbox1=\vbox{\unvbox4}%
      \ifdim\dimen0<\dimexpr \ht1+\dp1+\ht2+\dp2\relax
        \dimen0=\textheight
      \fi
      \setbox1=\vsplit0 to \dimexpr \dimen0-\ht2-\dp2\relax
      \begin{table}[bp]
        \unvbox2
        \unvbox1
      \end{table}
    \else
      \setbox1=\vsplit0 to \dimexpr \textheight-\ht3-\dp3\relax
      \begin{table}[tp]
        \usebox3
        \unvbox1
      \end{table}
    \fi
  \repeat
\egroup}
\makeatother

\makeatletter
\def\blfootnote{\xdef\@thefnmark{}\@footnotetext}
\makeatother

\begin{document}

\suppressfloats

\begin{center}
    ~\\[0.5em]
    {\LARGE\bfseries
        Accurate and scalable exchange-correlation\\ with deep learning
    }
    \vspace{1.5em}
    
    \newcommand{\maybecomma}{,}
    \renewcommand{\author}[2][]{#2${}^{\text{#1}}$\maybecomma}
    \author[1,\(\dagger\)]{Giulia Luise}
    \author[1,\(\dagger\)]{Chin-Wei Huang}
    \author[1,\(\dagger\)]{Thijs Vogels}
    \author[1,\(\dagger\)]{Derk P.\ Kooi}
    \author[1,\(\dagger\)]{Sebastian Ehlert}
    \author[1]{Stephanie Lanius} 
    \author[1]{Klaas J.\ H.\ Giesbertz} 
    \author[1,2]{Amir Karton}
    \author[1]{Deniz Gunceler} 
    \author[1]{Stefano Battaglia}
    \author[1]{Gregor N.\ C.\ Simm}
    \author[1]{P.\ Bernát Szabó}
    \author[1]{Megan Stanley}
    \author[1]{Wessel P.\ Bruinsma} 
    \author[1]{Lin Huang}
    \author[1]{Xinran Wei}
    \author[1]{Jos\'e Garrido Torres} \\
    \author[1]{Abylay Katbashev}
    \author[1]{Rodrigo Chavez Zavaleta}
    \author[1,a]{B\'alint M\'at\'e}
    \author[1,b]{S\'ekou-Oumar Kaba}
    \author[1]{Roberto Sordillo}
    \author[3]{Yingrong Chen}
    \author[3]{David B.\ Williams-Young}
    \author[1]{Christopher M.\ Bishop} \\
    \author[1,*]{Jan Hermann}
    \author[1,*]{Rianne van den Berg}
    \renewcommand{\maybecomma}{}
    \author[1,*]{Paola Gori-Giorgi} 
    \vspace{1em}
    
    $^{\dagger}$These authors contributed equally and are ordered randomly.
    
    {
    \itshape
    $^1$Microsoft Research, AI for Science\\
    $^2$School of Science and Technology, University of New England, Australia\\
    $^3$Microsoft Quantum
    }

    $^*$\texttt{\small\{pgorigiorgi,rvandenberg,jan.hermann\}@microsoft.com}

    \vspace{1.5em}

  \begin{minipage}{0.85\linewidth}
    \small
    \paragraph{Abstract}
Density Functional Theory (DFT) underpins much of modern computational chemistry and materials science. Yet, the reliability of DFT-derived predictions of experimentally measurable properties remains fundamentally limited by the need to approximate the unknown exchange-correlation (XC) functional.
The traditional paradigm for improving accuracy has relied on increasingly elaborate hand-crafted functional forms. 
This approach has led to a longstanding trade-off between computational efficiency and accuracy, which remains insufficient for reliable predictive modelling of laboratory experiments.
Here we introduce \clippy, a deep learning–based XC functional that surpasses state-of-the-art hybrid functionals in accuracy across the main-group chemistry benchmark set GMTKN55 with an error of $2.8$~kcal/mol, while retaining the lower computational cost characteristic of semi-local DFT.
This demonstrated departure from the historical trade-off between accuracy and efficiency is enabled by learning non-local representations of electronic structure directly from data, bypassing the need for increasingly costly hand-engineered features. Leveraging an unprecedented volume of high-accuracy reference data from wavefunction-based methods, we establish that modern deep learning enables systematically improvable neural exchange-correlation models as training datasets expand, positioning first-principles simulations to become progressively more predictive.
\end{minipage}
  \end{center}
\vspace{1em}

\blfootnote{
\hspace{-1em}
$^\mathrm{ab}$Work done during an internship at Microsoft Research AI for Science\hspace{1em}
$^\mathrm{a}$Current address: \textit{University of Geneva, Department of Computer Science}\hspace{1em}
$^\mathrm{b}$Current address: \textit{Mila -- Quebec AI Institute}
}

\begin{refsection}

\section{Introduction}
The energy of the electrons in molecules and materials serves as a glue between their atoms, determining the stability and properties of the chemical structure. 
Accurately computing the many-electron energy is therefore essential for predictive modeling across a broad spectrum of applications, including assessing whether a chemical reaction will proceed, whether a candidate drug molecule will bind to its target protein, whether a material is suitable for carbon capture, or whether a flow battery can be optimized for renewable energy storage. 
Unfortunately, computing this energy amounts to solving the Schr\"odinger equation, whose cost scales exponentially with the number of electrons $N$. Density functional theory (DFT)\cite{kohn_nobel_1999} provides an exact reformulation that replaces the many-electron wavefunction with the much simpler electron density. 
Although exact in principle, one component of the total energy --- the exchange-correlation (XC) functional --- remains unknown and must be approximated in practical implementations. The role of the XC functional is to capture intricate quantum many-body interactions of electrons using only the electron density, making this a universal functional that has the same form for all molecules and materials.\cite{hohenberg_inhomogeneous_1964,lieb_density_1983} 
Equipped with a formalism\cite{kohn_selfconsistent_1965} whose cost scales asymptotically as $O(N^3)$, and supported by practical functional approximations pioneered over several decades, \cite{perdew_accurate_1986,becke_densityfunctional_1988,perdew_generalized_1996,becke_densityfunctional_1993,grimme_semiempirical_2006,sun_strongly_2015,mardirossian_thirty_2017,lebeda_metagga_2025} DFT has become the computational workhorse in disciplines ranging from (bio)chemistry to catalysis to materials science.\cite{m.teale_dft_2022} 
However, DFT users must still choose from among hundreds of XC functional approximations,\cite{mardirossian_thirty_2017,goerigk_look_2017,m.teale_dft_2022} often relying on dedicated benchmark studies or experimental results to guide the choice for the application at hand. 
Crucially, current XC approximations still fall short of the accuracy required to reliably {\em predict} experimental outcomes across a wide range of chemical systems and properties.\cite{goerigk_look_2017,mardirossian_thirty_2017,m.teale_dft_2022} Achieving this level of accuracy --- commonly known as {\em chemical accuracy} --- typically demands errors below $\sim$1~kcal/mol for processes involving making and breaking covalent chemical bonds.\cite{m.teale_dft_2022}
This means, for example, that in silico screening pipelines for molecule and material discovery often pass too many candidates to the lab, with a large fraction failing experimental verification.  
In addition, lower-cost methods such as force fields and property-guided generative models trained on DFT data inherit these same limitations. 
The search for a general-purpose XC functional that meets chemical accuracy has persisted for over 60 years and is sometimes referred to as ``the pursuit of the divine functional''\cite{mattsson_pursuit_2002} --- a challenge with profound implications for accelerating scientific discovery.

\begin{figure}
\centering
\pgfmathsetmacro{\figscale}{\the\textwidth / 1170}%
\pgfmathsetlengthmacro{\bwidth}{688 * \figscale pt}%
\pgfmathsetlengthmacro{\bheight}{238 * \figscale pt}%
\pgfmathsetlengthmacro{\gap}{1pt}%
\scalebox{\figscale}{\input{figures/intro/top_figure}}%

\smallskip
\centerline{\small\textbf{(a)} Anatomy of an exchange-correlation functional}

\vspace{30pt}
\hspace{-18pt}
\noindent
\begin{minipage}[t]{\bwidth}%
\scalebox{\figscale}{\input{figures/intro/model_overview}}%

\smallskip
\centering{\small\textbf{(b)} \textcolor{ibm_pink}{\textbf{\clippy}}'s learned grid point processing}
\end{minipage}%
\hspace{\gap}%
\begin{minipage}[t]{\dimexpr\textwidth-\bwidth-\gap\relax}%
\raisebox{10pt}[\bheight][0pt]{%
  \vbox to \bheight{%
    \vfill
    \hbox to \linewidth{\hfill\begin{tikzpicture}
\pgfmathsetmacro{\pw}{5.40}
\pgfmathsetmacro{\ph}{2.90}

\draw[black!20, thin] (0, 0.400) -- (\pw, 0.400);
\draw[thin, black] (-0.06, 0.400) -- (0, 0.400);
\node[anchor=east, font=\scriptsize, text=black, inner sep=1pt] at (-0.10, 0.400) {GGA};
\draw[black!20, thin] (0, 1.500) -- (\pw, 1.500);
\draw[thin, black] (-0.06, 1.500) -- (0, 1.500);
\node[anchor=east, font=\scriptsize, text=black, inner sep=1pt] at (-0.10, 1.500) {meta-GGA};
\draw[black!20, thin] (0, 2.600) -- (\pw, 2.600);
\draw[thin, black] (-0.06, 2.600) -- (0, 2.600);
\node[anchor=east, font=\scriptsize, text=black, inner sep=1pt] at (-0.10, 2.600) {Hybrid};

\draw[black, thin] (0, 0) -- (\pw, 0);

\draw[thin, black!70] (0.000, -0.06) -- (0.000, 0);
\node[anchor=north, font=\scriptsize, inner sep=1pt] at (0.000, -0.08) {2};
\draw[thin, black!70] (0.732, -0.06) -- (0.732, 0);
\node[anchor=north, font=\scriptsize, inner sep=1pt] at (0.732, -0.08) {3};
\draw[thin, black!70] (1.464, -0.06) -- (1.464, 0);
\node[anchor=north, font=\scriptsize, inner sep=1pt] at (1.464, -0.08) {4};
\draw[thin, black!70] (2.196, -0.06) -- (2.196, 0);
\node[anchor=north, font=\scriptsize, inner sep=1pt] at (2.196, -0.08) {5};
\draw[thin, black!70] (2.928, -0.06) -- (2.928, 0);
\node[anchor=north, font=\scriptsize, inner sep=1pt] at (2.928, -0.08) {6};
\draw[thin, black!70] (3.660, -0.06) -- (3.660, 0);
\node[anchor=north, font=\scriptsize, inner sep=1pt] at (3.660, -0.08) {7};
\draw[thin, black!70] (4.392, -0.06) -- (4.392, 0);
\node[anchor=north, font=\scriptsize, inner sep=1pt] at (4.392, -0.08) {8};
\draw[thin, black!70] (5.124, -0.06) -- (5.124, 0);
\node[anchor=north, font=\scriptsize, inner sep=1pt] at (5.124, -0.08) {9};
\node[anchor=north, font=\scriptsize] at (2.562, -0.28) {WTMAD-2 on GMTKN55 (kcal/mol)};

\fill[ibm_orange] (4.645, 0.400) circle[radius=0.065];
\fill[ibm_pink] (0.589, 1.500) circle[radius=0.065];
\fill[ibm_purple] (2.607, 1.500) circle[radius=0.065];
\fill[ibm_purple] (3.842, 1.500) circle[radius=0.065];
\fill[ibm_blue] (0.899, 2.600) circle[radius=0.065];
\fill[ibm_blue] (1.432, 2.600) circle[radius=0.065];
\fill[ibm_blue] (2.074, 2.600) circle[radius=0.065];
\fill[ibm_blue] (3.209, 2.600) circle[radius=0.065];

\node[anchor=west, rotate=25, font=\scriptsize, text=ibm_orange, inner sep=1pt] at (4.685, 0.460) {revPBE};
\node[anchor=west, rotate=25, font=\scriptsize, text=ibm_pink, inner sep=1pt] at (0.629, 1.560) {\clippy-1.1};
\node[anchor=west, rotate=25, font=\scriptsize, text=ibm_purple, inner sep=1pt] at (2.647, 1.560) {B97M-V};
\node[anchor=west, rotate=25, font=\scriptsize, text=ibm_purple, inner sep=1pt] at (3.882, 1.560) {r\textsuperscript{2}SCAN};
\node[anchor=west, rotate=25, font=\scriptsize, text=ibm_blue, inner sep=1pt] at (0.939, 2.660) {$\omega$B97M-V};
\node[anchor=west, rotate=25, font=\scriptsize, text=ibm_blue, inner sep=1pt] at (1.472, 2.660) {$\omega$B97X-V};
\node[anchor=west, rotate=25, font=\scriptsize, text=ibm_blue, inner sep=1pt] at (2.114, 2.660) {M06-2X};
\node[anchor=west, rotate=25, font=\scriptsize, text=ibm_blue, inner sep=1pt] at (3.249, 2.660) {B3LYP};
\end{tikzpicture}\hfill}%
    \vfill
  }%
}%

\smallskip
\centering{\small\textbf{(c)} Benchmark: errors for main-group chemistry}
\end{minipage}
\caption{ \textbf{\clippy is a scalable deep learned exchange-correlation functional.}
    (a) Jacob's ladder of density functional approximations\cite{perdew_jacobs_2001} defines the rungs \textcolor{ibm_yellow}{\textbf{LDA}}, \textcolor{ibm_orange}{\textbf{GGA}} and 
    \textcolor{ibm_purple}{\textbf{meta-GGA}} by expanding the set of semi-local features they extract from an electronic density matrix into a grid representation.
    The next rungs, \textcolor{ibm_blue}{\textbf{hybrid}} and \textcolor{ibm_blue}{\textbf{double hybrid}}, extract increasingly expensive wavefunction-based information directly from the density matrix. 
    \textcolor{ibm_pink}{\textbf{\clippy}} departs from this ladder by extracting relatively cheap meta-GGA features, and instead gaining expressivity by learning non-local interactions between grid points at a manageable and controllable cost.
    (b) High-level overview of the neural network architecture for the \clippy functional.
    (c) Weighted total mean absolute deviation (WTMAD-2) on the GMTKN55\cite{goerigk_look_2017} benchmark for general main-group thermochemistry, kinetics and non-covalent interactions. Results are shown for popular functionals, including the best-performing ones in each computational cost category. 
}
\label{fig:jacobs-ladder}
\end{figure}
The prevailing approach has been to handcraft functional forms based on a limited set of ingredients defined by the so-called Jacob's ladder of DFT;\cite{perdew_jacobs_2001} see Fig.~\ref{fig:jacobs-ladder}. Like its biblical namesake, it is intended to guide users toward the ``heaven'' of chemical accuracy.
The ingredients at the lower rungs allow retaining the asymptotic $O(N^3)$
scaling of DFT, but amount to XC functionals that only use (semi-)local information such as the density, its gradient, the Laplacian and the Kohn-Sham kinetic energy density. 
However, it is well established that the exact XC functional exhibits non-local dependence on the density,\cite{lieb_density_1983} and in practice lower-rung approximations yield only limited accuracy.
To improve accuracy, researchers began introducing non-locality through wavefunction-like ingredients.\cite{becke_densityfunctional_1993,grimme_semiempirical_2006}  While this approach enhances accuracy in many cases, it does not do so for all chemical problems, and it increases the computational complexity\footnotemark 
\footnotetext{We express computational cost using the standard asymptotic scaling with system size, $O(f(N))$. In practice, however, actual performance depends on algorithmic speedups, hardware optimizations, and prefactors. Therefore, in Sec.~\ref{sec:cost}, we empirically compare the cost of \clippy with that of other functionals.} to $O(N^4)$, $O(N^5)$ or higher, thereby defining the higher rungs of the ladder.  
All XC functionals in common use are built from this hierarchy of Jacob's ladder ingredients, with a clear accuracy-cost trade-off. They differ primarily in how these ingredients are combined and the number of parameters involved. 
The focus on these ingredients is driven by their availability in standard software packages and compatibility with exact constraints, offering a practical foundation for building functional approximations. 

As in many other areas of science, machine learning (ML) has been explored as a promising approach for developing accurate XC functionals, revealing the challenges and subtleties of this complex learning problem.\cite{akashi_can_2025} Yet, to date this has not led to a meaningful shift in the established accuracy-cost trade-off, and no ML-based functional has seen widespread adoption. There are two interlinked reasons for this. First, high-accuracy data for this complex learning problem are very scarce, as they must be generated using computationally intensive wavefunction methods that require specialized expertise to be used at scale. Second, confined to this low-data regime, the vast majority of efforts have been limited to feeding handcrafted features into machine learning models, whether based on Jacob's ladder ingredients\cite{tozer_exchangecorrelation_1996,dick_machine_2020,kasim_learning_2021,cuierrier_constructing_2021,kirkpatrick_pushing_2021,kanungo_learning_2024,liu_supervised_2023} or newly designed descriptors.\cite{nagai_completing_2020,bystrom_cider_2022,bystrom_nonlocal_2023,polak_realspace_2024} 
This approach mirrors classic machine learning strategies prior to the rise of deep learning (DL), which may partly account for the limited impact observed so far. In the absence of sufficient data to power a more modern deep representation learning approach, the handful of efforts to move beyond handcrafted features, though promising, have remained focused on model systems or
 narrowly defined problems.\cite{schmidt_machine_2019,li_kohnsham_2021,margraf_pure_2021,kalita_how_2022,gao_learning_2024}

In this work, we introduce \clippy, a modern deep-learning solution to this long-standing problem, addressing both the scarcity of chemically accurate training data and several fundamental challenges in applying deep learning to exchange-correlation functionals. Using an efficient wavefunction-based protocol,\cite{karton_explicitly_2012} we generated a training set of approximately $400,000$ high-accuracy energy differences spanning diverse chemistry, augmented by the limited public datasets available at comparable accuracy.
We further design a neural network architecture capable of learning non-local electronic representations directly from data, while relying only on simple semi-local input features.
\clippy achieves higher accuracy than leading hybrid functionals across a broad main-group chemical space. Notably, this is accomplished using a scalable neural-network formulation that retains the asymptotic complexity of semi-local DFT and naturally enables GPU acceleration, effectively decoupling accuracy from computational cost along Jacob’s ladder. \clippy also matches or exceeds hybrid functionals in predicting equilibrium structures and dipole moments.

As additional high-accuracy data covering increasingly diverse regions of chemical space are incorporated in training, we demonstrate that \clippy’s accuracy improves systematically without an increase in computational complexity at inference time. This empirical scaling behaviour motivates the design of \clippy as a continuously evolving exchange-correlation modelling framework rather than a static functional.
Unlike the traditional DFT ``zoo'',\cite{goerigk_look_2017} in which new approximations accumulate without retirement of older ones, \clippy follows an explicit, versioned development process, with successive releases superseding their predecessors through significant demonstrable gains in accuracy and robustness.
In this manuscript, \clippy refers to the current best model, \clippyOnePointOne\footnotemark[2], with version labels provided in figures and tables where appropriate to ensure transparency and reproducibility.
More broadly, by establishing a framework for systematically improvable DFT, this work outlines a concrete route towards genuinely predictive accuracy. Achieving such predictive accuracy would remove a longstanding bottleneck that has constrained the field’s ability to shift the centre of gravity from laboratory-based experimentation towards in silico discovery.

\footnotetext[2]{A previous version of this manuscript reported results for \clippyOnePointZero trained on roughly $150,000$ data points. The present, more accurate version \clippyOnePointOne is trained on $\sim 400,000$ data points and features a refined deep-learning architecture and training protocol.}

\section{Deep learning the XC functional: basic challenges and solutions}
\label{sec:learning-xc-functional}
The success of DFT is based on the Kohn-Sham (KS) formalism,\cite{kohn_selfconsistent_1965} which decomposes the total energy functional into components that capture large effects such as the Pauli exclusion principle and long-range classical electrostatics, as well as the remaining unknown term that we aim to learn --- the XC functional --- which accounts for a smaller but crucial energy contribution due to quantum many-body effects. The XC functional $\Exc[\dens]$ maps the electron density $\dens(r)$, a non-negative function over three-dimensional space, to a scalar value representing the XC energy. 
In practical implementations of KS DFT, all terms except for the XC functional are evaluated using the density represented in a basis set via the density matrix, with atom-centered Gaussian functions being the most commonly used basis functions in chemistry. Focusing on the semi-local functional rungs, the XC energy is instead evaluated using a representation of the electron density on a large integration grid. For molecules containing up to several hundred atoms, the integration grids typically consist of $\sim$\!$10^4$--$10^6$ points. Since any such grid is necessarily a finite representation of the continuous density, the XC functional should have a well-defined limit when the grid becomes infinitely dense and show good convergence as the grid is refined.

\subsection{The two-stage training protocol}
The learning problem we address is to obtain an XC functional $\Exc^\params[\dens]$ parameterized by a neural network with parameters $\theta$ that, combined with the other known energy terms, yields accurate total energy predictions. 
An important factor in designing a training protocol is that, at inference time, $\Exc^\params[\dens]$ is not just evaluated once in isolation.  Instead, it is used inside a self-consistent field (SCF) loop that minimizes the total energy with respect to the density, yielding both the energy minimum and the minimizing self-consistent density $\dens_{\mathrm{SCF}}$.
This makes the learning problem different from a standard regression problem, such as one encounters for instance when training machine-learned force fields.\cite{behler_generalized_2007} 

While previous ML attempts at learning the XC functional have proposed and analyzed several training protocols to address this challenge,\cite{nagai_completing_2020,dick_machine_2020,kirkpatrick_pushing_2021,li_kohnsham_2021,gao_learning_2024,kanungo_learning_2024} many of them are too computationally demanding for the much larger-scale training considered in this work. Performing SCF calculations for every training example is more expensive than doing a single inference call of the model, and this cost becomes prohibitive when considering our large-scale training dataset. Moreover, close to model initialization, the parameterized $\Exc^\params$ is a very poor estimate of the true XC functional: the SCF procedure would fail to converge for many systems, and for those that do converge, the resulting self-consistent densities and energies would be of poor quality. Starting from such poor predictions risks the functional compensating functional-driven errors with density-driven errors in order to drive down the error between the predicted total energy and the reference energy provided by the high-accuracy wavefunction method.

The alternative strategy we adopt is to introduce a two-stage training protocol, consisting of a pre-training phase and a fine-tuning phase. In the pre-training phase the functional is evaluated for a fixed input density, bypassing the SCF procedure. This brings the problem closer to standard regression: a single forward pass of the model, combined with the other known energy terms, yields a total energy prediction for fixed input density features. 
However, this approach requires a reasonable approximation of the true minimizing ground-state density. Producing highly accurate ground-state densities beyond the quality of DFT densities is even more challenging than producing accurate energy labels,\cite{assaraf_improved_2007,chen_initio_2021,cheng_highly_2025} and does not scale to our dataset size. Instead, following \citet{kirkpatrick_pushing_2021}, we use approximate B3LYP\cite{becke_densityfunctional_1993,stephens_initio_1994} densities (not B3LYP energies) during training, which can be generated at scale for our entire dataset.\cite{ehlert_accurate_2025} 
This enables us to expose the model to a broad range of approximate densities, combined with very accurate wavefunction-level reference energy labels, at reasonable training cost. To predict the total energy, the remaining KS energy components are computed from the B3LYP density and KS orbitals, as detailed in Sec.~\ref{sec:training-objective} of the Supplementary Information. 
Since reference energies from wavefunction methods are generally more accurate for reaction energy differences than for total energies, the model is trained with a reaction energy regression loss rather than a total energy loss.

After this pre-training phase, the model produces reasonable energy predictions given B3LYP densities, but it has not yet been incentivized to produce accurate self-consistent densities and total energies when deployed inside an SCF procedure. This motivates a second, much shorter, fine-tuning phase.
In the fine-tuning phase, the model is trained using its own SCF densities, generated on the fly during training. This aims to close the gap between the accuracy achieved when evaluating the functional on the fixed approximate input densities from the pre-training stage, and the accuracy obtained when evaluating the functional on its own SCF densities. Like during pre-training, the energy references we train against originate from wavefunction-level theory. Crucially, this procedure does not require backpropagating through the SCF loop, as described in more detail in Sec.~\ref{subsupp:SCF-finetuning} of the Supplementary Information.
During the SCF fine-tuning phase we monitor the aforementioned accuracy gap on a holdout validation set, as well as the accuracy of our SCF densities by comparing dipole moments with accurate labels available in the literature.\cite{hait_how_2018a} As we show in Sec.~\ref{sec:densities}, a notable empirical observation of our XC learning framework is that, during fine-tuning, both the SCF energies and the SCF electron densities improve simultaneously, suggesting that the increased accuracy in energies is not the result of an error compensation between a functional-driven error and a density-driven error.

\subsection{The XC functional architecture of \clippy}
Several mathematical properties of the XC functional are known, usually referred to as exact constraints.\cite{levy_hellmannfeynman_1985,kaplan_predictive_2022,lieb_density_1983,sun_strongly_2015} 
Following a well-established practice in DFT, we build in some of the most energetically relevant constraints (such as the high-density uniform coordinate scaling, size-consistency, and the Lieb-Oxford lower bound\cite{lieb_lower_1979}) by constructing \clippy as
\begin{equation}
\label{eq:Exc-basic}
 \Exc^\theta[\dens] = - \frac{3}{4} \left(\frac{6}{\pi}\right)^{\frac{1}{3}}\int \left(\dens^{(\uparrow)}(r)^{4/3} + \dens^{(\downarrow)}(r)^{4/3}\right) f_\params[\mathbf{x}[\dens]](r) \,dr,   
\end{equation}
where $\dens^{(\uparrow)}$ and $\dens^{(\downarrow)}$ are the densities of the two spin channels and $f_\params$ is a bounded enhancement factor.
The vast majority of previous ML attempts only learned the enhancement factor with a {\em local function} $f_\params(\mathbf{x}[\dens](r))$ of the given hand-designed input features $\mathbf{x}[\dens](r)$.\cite{dick_machine_2020,kasim_learning_2021,cuierrier_constructing_2021,kirkpatrick_pushing_2021,kanungo_learning_2024,nagai_completing_2020,bystrom_cider_2022,bystrom_nonlocal_2023,polak_realspace_2024} In contrast, our DL approach takes inspiration from neural operators,\cite{kovachki_neural_2023} and models the enhancement factor as a neural functional that learns non-local representations from the input features, hence the distinct notation $f_\params[\mathbf{x}[\dens]](r)$.
The architecture for the enhancement factor, outlined in Fig.~\ref{fig:jacobs-ladder}(b), takes as input the set of density-dependent semi-local features $\mathbf{x}[\dens]$ of the standard meta-generalized-gradient approximation (meta-GGA) $O(N^3)$ rung at each integration point $r$. The first module learns spin-order invariant features, followed by further learned local processing at each position $r$. Subsequently, in the non-local interactions module, information is exchanged between each grid point and its atom's coarse point (a representative point on a coarser, atom-centred grid; see Sec.~\ref{sec:architecture}), resulting in a scalable architecture that preserves the asymptotic computational complexity of semi-local DFT. Finally, the spin-order invariant features and the non-local representations are combined and further processed to predict the enhancement factor $f_\params[\mathbf{x}[\dens]](r)$, followed by a numerical integration using Eq.~\eqref{eq:Exc-basic}.
For a more detailed overview of the architecture, the reader is referred to Extended Data Fig.~\ref{fig:architecture-diagram-nonlocal} and Sec.~\ref{sec:architecture}. Importantly, \clippy's architecture is designed to have a well-defined limit when the grid becomes infinitely dense: once trained, it can be evaluated on different grid types and sizes, and displays convergence behavior similar to traditional functionals, as shown in Fig.~\ref{fig:grid-convergence} of the Supplementary Information.

Non-locality on the DFT grid has also been used in prior hand-crafted approaches to model dispersion\cite{dion_van_2004,langreth_van_2005,vydrov_nonlocal_2010} --- a long-range, subtle yet crucial component of the XC energy, essential for capturing interactions that do not involve the making or breaking of covalent chemical bonds.\cite{hermann_firstprinciples_2017} Our focus here is distinct: we target thermochemistry --- the energetics of forming and breaking covalent bonds --- and kinetics, namely the energy barriers separating reactants from products. Accordingly, we do not yet attempt to model dispersion explicitly, and instead train our functional using a fixed D3 dispersion correction.\cite{grimme_consistent_2010,grimme_effect_2011} Learning dispersion interactions within this framework is deferred to future work.

With learned non-locality, \clippy eliminates the need for computationally expensive ingredients higher up Jacob's ladder, achieving high accuracy for main-group chemistry given sufficient training data without increasing computational cost. In particular, accuracy in main-group thermochemistry and kinetics has long been dominated by the accuracy–cost trade-off that defines Jacob’s ladder. By disrupting this trade-off, we establish a deep-learning approach that is systematically improvable with more data, charting a route towards learning the universal XC functional in a scalable manner.

\subsection{The training data}
More broadly, our approach mirrors the shift from classic machine learning with hand-designed features to deep representation learning with neural networks.\cite{ lecun_deep_2015, goodfellow_deep_2016, bishop_deep_2024} Just as large-scale datasets such as ImageNet\cite{deng_imagenet_2009,russakovsky_imagenet_2015} acted as catalysts for deep learning in computer vision\cite{krizhevsky_imagenet_2012}, realizing this approach for DFT demands a dataset far exceeding what has been available to date. 
Given the lack of public datasets at the required scale and accuracy, we have initiated\cite{ehlert_accurate_2025} a large-scale effort to generate and assemble a chemically diverse collection of high-accuracy quantum chemistry data, which we term the Microsoft Research Accurate Chemistry Collection (MSR-ACC). All reference energies are obtained with wavefunction methods achieving at least 1~kcal/mol accuracy with respect to experiment, and the collection continues to expand.

The current version of MSR‑ACC comprises approximately 395,000 energy differences labeled with the high‑accuracy W1‑F12 and W1w thermochemical protocols,\cite{karton_w4_2006,karton_explicitly_2012} achieving CCSD(T)/CBS accuracy. Its largest component, in terms of both scale and coverage, consists of roughly 120,000 total atomization energies (TAEs). The remaining datasets span a broad range of chemical space, including proton and electron affinities, ionization potentials, reaction pathways, non‑covalent clusters, water dimer structures, as well as slightly distorted geometries and conformers derived from selected molecules in the TAE set, in addition to atomic properties.
This collection is further augmented with approximately 80,000 energy differences from publicly available datasets of comparable or higher accuracy, covering diverse chemistry such as non‑covalent interactions, atomization energies of larger molecules, decomposition energies of mindless structures, and reaction barriers. After removing overlap with test sets based on molecular graphs and reserving dedicated holdout sets, the resulting training set for the current version of \clippy\ (\clippyOnePointOne) comprises about 400,000 energy differences, representing an unprecedented combination of size, accuracy, and chemical diversity.

\begin{table}[h!]
    \centering
    \caption{GMTKN55 unweighted MAE breakdown (kcal/mol). All functionals without VV10 corrections use D3(BJ), except M06-2X which uses D3(0). For each subset, we call out the functionals that perform best (\goldmedal{}), second place (\silvermedal{}) and third place (\bronzemedal{}). \label{tab:gmtkn55-breakdown}}
    
\providecolor{skala_highlight}{HTML}{EAEEF3}
\providecolor{gold_medal}{HTML}{FFD700}
\providecolor{gold_medal_dark}{HTML}{B8960A}
\providecolor{silver_medal}{HTML}{C0C0C0}
\providecolor{silver_medal_dark}{HTML}{808080}
\providecolor{bronze_medal}{HTML}{CD7F32}
\providecolor{bronze_medal_dark}{HTML}{8B5521}
\providecolor{rung_gga}{HTML}{FE6100}
\providecolor{rung_meta_gga}{HTML}{785EF0}
\providecolor{rung_hybrid}{HTML}{648FFF}
\providecommand{\goldmedal}{\tikz[baseline=-0.8ex]{\fill[gold_medal] circle(2.8pt);}}
\providecommand{\silvermedal}{\tikz[baseline=-0.8ex]{\fill[silver_medal] circle(2.8pt);}}
\providecommand{\bronzemedal}{\tikz[baseline=-0.8ex]{\fill[bronze_medal] circle(2.8pt);}}
\tablefontsize
\setlength{\tabcolsep}{2.5pt}
\ifdefined\rbs\else\newlength{\rbs}\fi\setlength{\rbs}{4pt}
\ifdefined\ras\else\newlength{\ras}\fi\setlength{\ras}{4pt}
\ifdefined\rhs\else\newlength{\rhs}\fi\setlength{\rhs}{6pt}
\setlength{\aboverulesep}{0pt}
\setlength{\belowrulesep}{0pt}
\newlength{\gmtknnumwd}
\newlength{\gmtknmedalwd}
\newlength{\gmtknpairwd}
\settowidth{\gmtknnumwd}{88.88}
\settowidth{\gmtknmedalwd}{\goldmedal{}}
\setlength{\gmtknpairwd}{\dimexpr\gmtknnumwd + \gmtknmedalwd + 2pt\relax}
\newcommand{\gmtkcell}[2]{\makebox[\gmtknnumwd][r]{#1}\,\makebox[\gmtknmedalwd][c]{#2}}
\noindent
\resizebox{\textwidth}{!}{%
\begin{tabular}{@{} l @{\hspace{4pt}} l *{8}{c} @{}}
    \toprule
    \rule{0pt}{\rbs+\ht\strutbox} &
    & \multicolumn{1}{c}{\textbf{GGA}}
    & \multicolumn{3}{c}{\textbf{meta-GGA}}
    & \multicolumn{4}{c}{\textbf{Hybrid}}
    \\
    &
    & revPBE
    & r\textsuperscript{2}SCAN
    & B97M-V
    & \cellcolor{skala_highlight}%
\clippyOnePointOne
    & B3LYP
    & M06-2X
    & \(\omega\)B97X-V
    & \(\omega\)B97M-V\rule[-0.5\rhs]{0pt}{\ht\strutbox+\rhs}
    \\
    \cmidrule(l{2pt}r{2pt}){3-3}
    \cmidrule(l{2pt}r{2pt}){4-6}
    \cmidrule(l{2pt}r{2pt}){7-10}
    \multicolumn{5}{@{}l}{\textbf{Basic properties}\rule[-0.3ex]{0pt}{0.3ex}} & \cellcolor{skala_highlight} & \multicolumn{4}{l}{} \\
    AL2x6\rule{0pt}{\rbs+\ht\strutbox} & {\scriptsize Dimerization energies of AlX${}_3$ compounds}
    & \gmtkcell{2.07}{}
    & \gmtkcell{1.71}{}
    & \gmtkcell{1.53}{}
    & \cellcolor{skala_highlight}%
\gmtkcell{0.32}{\goldmedal{}%
}
    & \gmtkcell{2.72}{}
    & \gmtkcell{0.90}{\silvermedal{}%
}
    & \gmtkcell{1.20}{\bronzemedal{}%
}
    & \gmtkcell{1.29}{}
    \\
    ALK8 & {\scriptsize Dissociation and other reactions of alkaline compounds}
    & \gmtkcell{3.62}{}
    & \gmtkcell{3.48}{}
    & \gmtkcell{2.63}{}
    & \cellcolor{skala_highlight}%
\gmtkcell{3.07}{}
    & \gmtkcell{2.46}{\bronzemedal{}%
}
    & \gmtkcell{2.34}{\silvermedal{}%
}
    & \gmtkcell{0.92}{\goldmedal{}%
}
    & \gmtkcell{2.48}{}
    \\
    ALKBDE10 & {\scriptsize Dissociation energies in group-1 and -2 diatomics}
    & \gmtkcell{5.16}{}
    & \gmtkcell{5.07}{}
    & \gmtkcell{4.12}{}
    & \cellcolor{skala_highlight}%
\gmtkcell{3.79}{\silvermedal{}%
}
    & \gmtkcell{4.39}{}
    & \gmtkcell{4.80}{}
    & \gmtkcell{4.05}{\bronzemedal{}%
}
    & \gmtkcell{3.74}{\goldmedal{}%
}
    \\
    BH76RC & {\scriptsize Reaction energies of the BH76 set}
    & \gmtkcell{2.80}{}
    & \gmtkcell{2.97}{}
    & \gmtkcell{2.02}{}
    & \cellcolor{skala_highlight}%
\gmtkcell{0.60}{\goldmedal{}%
}
    & \gmtkcell{1.98}{}
    & \gmtkcell{0.94}{\bronzemedal{}%
}
    & \gmtkcell{1.51}{}
    & \gmtkcell{0.87}{\silvermedal{}%
}
    \\
    DC13 & {\scriptsize 13 difficult cases for DFT methods}
    & \gmtkcell{8.46}{}
    & \gmtkcell{7.88}{}
    & \gmtkcell{5.20}{\silvermedal{}%
}
    & \cellcolor{skala_highlight}%
\gmtkcell{2.20}{\goldmedal{}%
}
    & \gmtkcell{10.58}{}
    & \gmtkcell{7.49}{}
    & \gmtkcell{6.52}{}
    & \gmtkcell{5.40}{\bronzemedal{}%
}
    \\
    DIPCS10 & {\scriptsize Double-ionization potentials of closed-shell systems}
    & \gmtkcell{4.77}{}
    & \gmtkcell{5.12}{}
    & \gmtkcell{4.52}{}
    & \cellcolor{skala_highlight}%
\gmtkcell{2.07}{\goldmedal{}%
}
    & \gmtkcell{4.68}{}
    & \gmtkcell{3.11}{\silvermedal{}%
}
    & \gmtkcell{4.05}{\bronzemedal{}%
}
    & \gmtkcell{5.20}{}
    \\
    FH51 & {\scriptsize Reaction energies in various (in-)organic systems}
    & \gmtkcell{3.33}{}
    & \gmtkcell{2.19}{}
    & \gmtkcell{2.27}{}
    & \cellcolor{skala_highlight}%
\gmtkcell{0.48}{\goldmedal{}%
}
    & \gmtkcell{2.60}{}
    & \gmtkcell{1.19}{\bronzemedal{}%
}
    & \gmtkcell{2.32}{}
    & \gmtkcell{1.06}{\silvermedal{}%
}
    \\
    G21EA & {\scriptsize Adiabatic electron affinities}
    & \gmtkcell{2.84}{}
    & \gmtkcell{3.59}{}
    & \gmtkcell{3.09}{}
    & \cellcolor{skala_highlight}%
\gmtkcell{1.11}{\goldmedal{}%
}
    & \gmtkcell{1.99}{}
    & \gmtkcell{1.86}{\bronzemedal{}%
}
    & \gmtkcell{1.93}{}
    & \gmtkcell{1.44}{\silvermedal{}%
}
    \\
    G21IP & {\scriptsize Adiabatic ionization potentials}
    & \gmtkcell{4.19}{}
    & \gmtkcell{4.69}{}
    & \gmtkcell{3.02}{}
    & \cellcolor{skala_highlight}%
\gmtkcell{2.00}{\goldmedal{}%
}
    & \gmtkcell{3.55}{}
    & \gmtkcell{2.63}{\silvermedal{}%
}
    & \gmtkcell{3.04}{}
    & \gmtkcell{2.88}{\bronzemedal{}%
}
    \\
    G2RC & {\scriptsize Reaction energies of selected G2/97 systems}
    & \gmtkcell{6.15}{}
    & \gmtkcell{5.69}{}
    & \gmtkcell{4.60}{}
    & \cellcolor{skala_highlight}%
\gmtkcell{0.85}{\goldmedal{}%
}
    & \gmtkcell{2.73}{}
    & \gmtkcell{1.94}{\silvermedal{}%
}
    & \gmtkcell{3.93}{}
    & \gmtkcell{1.99}{\bronzemedal{}%
}
    \\
    HEAVYSB11 & {\scriptsize Dissociation energies in heavy-element compounds}
    & \gmtkcell{2.79}{}
    & \gmtkcell{3.07}{}
    & \gmtkcell{2.17}{\bronzemedal{}%
}
    & \cellcolor{skala_highlight}%
\gmtkcell{0.89}{\goldmedal{}%
}
    & \gmtkcell{3.37}{}
    & \gmtkcell{8.21}{}
    & \gmtkcell{1.46}{\silvermedal{}%
}
    & \gmtkcell{2.70}{}
    \\
    NBPRC & {\scriptsize Reactions involving trihydride oligomers}
    & \gmtkcell{1.98}{}
    & \gmtkcell{1.61}{}
    & \gmtkcell{1.76}{}
    & \cellcolor{skala_highlight}%
\gmtkcell{0.86}{\goldmedal{}%
}
    & \gmtkcell{2.01}{}
    & \gmtkcell{0.94}{\bronzemedal{}%
}
    & \gmtkcell{1.45}{}
    & \gmtkcell{0.91}{\silvermedal{}%
}
    \\
    PA26 & {\scriptsize Adiabatic proton affinities (incl. amino acids)}
    & \gmtkcell{4.70}{}
    & \gmtkcell{2.63}{}
    & \gmtkcell{3.32}{}
    & \cellcolor{skala_highlight}%
\gmtkcell{0.92}{\goldmedal{}%
}
    & \gmtkcell{2.86}{}
    & \gmtkcell{1.22}{\silvermedal{}%
}
    & \gmtkcell{2.62}{}
    & \gmtkcell{1.40}{\bronzemedal{}%
}
    \\
    RC21 & {\scriptsize Fragmentations and rearrangements in radical cations}
    & \gmtkcell{4.85}{}
    & \gmtkcell{5.11}{}
    & \gmtkcell{3.51}{}
    & \cellcolor{skala_highlight}%
\gmtkcell{2.13}{\bronzemedal{}%
}
    & \gmtkcell{2.53}{}
    & \gmtkcell{1.63}{\goldmedal{}%
}
    & \gmtkcell{3.53}{}
    & \gmtkcell{1.86}{\silvermedal{}%
}
    \\
    SIE4x4 & {\scriptsize Self-interaction-error related problems}
    & \gmtkcell{23.42}{}
    & \gmtkcell{18.17}{}
    & \gmtkcell{16.27}{}
    & \cellcolor{skala_highlight}%
\gmtkcell{12.68}{}
    & \gmtkcell{18.05}{}
    & \gmtkcell{8.67}{\goldmedal{}%
}
    & \gmtkcell{11.48}{\bronzemedal{}%
}
    & \gmtkcell{10.75}{\silvermedal{}%
}
    \\
    TAUT15 & {\scriptsize Relative energies in tautomers}
    & \gmtkcell{1.55}{}
    & \gmtkcell{1.57}{}
    & \gmtkcell{0.89}{}
    & \cellcolor{skala_highlight}%
\gmtkcell{0.32}{\goldmedal{}%
}
    & \gmtkcell{1.16}{}
    & \gmtkcell{0.76}{}
    & \gmtkcell{0.72}{\bronzemedal{}%
}
    & \gmtkcell{0.33}{\silvermedal{}%
}
    \\
    W4-11 & {\scriptsize Total atomization energies}
    & \gmtkcell{7.56}{}
    & \gmtkcell{3.88}{}
    & \gmtkcell{2.87}{}
    & \cellcolor{skala_highlight}%
\gmtkcell{1.50}{\goldmedal{}%
}
    & \gmtkcell{3.21}{}
    & \gmtkcell{3.02}{}
    & \gmtkcell{2.82}{\bronzemedal{}%
}
    & \gmtkcell{2.08}{\silvermedal{}%
}
    \\
    YBDE18\rule[-0.5em]{0pt}{0.5em} & {\scriptsize Bond-dissociation energies of AlX${}_3$ compounds}
    & \gmtkcell{4.42}{}
    & \gmtkcell{3.26}{}
    & \gmtkcell{4.37}{}
    & \cellcolor{skala_highlight}%
\gmtkcell{0.85}{\goldmedal{}%
}
    & \gmtkcell{4.72}{}
    & \gmtkcell{2.42}{\bronzemedal{}%
}
    & \gmtkcell{2.01}{\silvermedal{}%
}
    & \gmtkcell{2.90}{}
    \\
    \multicolumn{5}{@{}l}{\textbf{Reaction energies}\rule[-0.3ex]{0pt}{0.3ex}} & \cellcolor{skala_highlight} & \multicolumn{4}{l}{} \\
    BSR36 & {\scriptsize Bond-separation reactions of saturated hydrocarbons}
    & \gmtkcell{1.79}{}
    & \gmtkcell{0.45}{\silvermedal{}%
}
    & \gmtkcell{0.22}{\goldmedal{}%
}
    & \cellcolor{skala_highlight}%
\gmtkcell{0.99}{}
    & \gmtkcell{3.38}{}
    & \gmtkcell{2.48}{}
    & \gmtkcell{2.01}{}
    & \gmtkcell{0.46}{\bronzemedal{}%
}
    \\
    C60ISO & {\scriptsize Relative energies between C${}_{60}$ isomers}
    & \gmtkcell{9.85}{}
    & \gmtkcell{5.32}{\bronzemedal{}%
}
    & \gmtkcell{4.85}{\silvermedal{}%
}
    & \cellcolor{skala_highlight}%
\gmtkcell{8.16}{}
    & \gmtkcell{2.40}{\goldmedal{}%
}
    & \gmtkcell{7.86}{}
    & \gmtkcell{13.70}{}
    & \gmtkcell{11.85}{}
    \\
    CDIE20 & {\scriptsize Double-bond isomerization energies in cyclic systems}
    & \gmtkcell{1.49}{}
    & \gmtkcell{1.59}{}
    & \gmtkcell{1.62}{}
    & \cellcolor{skala_highlight}%
\gmtkcell{0.19}{\goldmedal{}%
}
    & \gmtkcell{0.99}{}
    & \gmtkcell{0.54}{\silvermedal{}%
}
    & \gmtkcell{0.62}{}
    & \gmtkcell{0.59}{\bronzemedal{}%
}
    \\
    DARC & {\scriptsize Reaction energies of Diels-Alder reactions}
    & \gmtkcell{3.69}{}
    & \gmtkcell{2.65}{}
    & \gmtkcell{3.50}{}
    & \cellcolor{skala_highlight}%
\gmtkcell{0.40}{\goldmedal{}%
}
    & \gmtkcell{8.01}{}
    & \gmtkcell{2.16}{\bronzemedal{}%
}
    & \gmtkcell{4.33}{}
    & \gmtkcell{0.75}{\silvermedal{}%
}
    \\
    ISO34 & {\scriptsize Isomerization energies of medium-sized organic molecules}
    & \gmtkcell{1.50}{}
    & \gmtkcell{1.29}{}
    & \gmtkcell{1.46}{}
    & \cellcolor{skala_highlight}%
\gmtkcell{0.22}{\goldmedal{}%
}
    & \gmtkcell{1.78}{}
    & \gmtkcell{1.23}{}
    & \gmtkcell{1.17}{\bronzemedal{}%
}
    & \gmtkcell{0.62}{\silvermedal{}%
}
    \\
    ISOL24 & {\scriptsize Isomerization energies of large organic molecules}
    & \gmtkcell{4.58}{}
    & \gmtkcell{4.02}{}
    & \gmtkcell{4.09}{}
    & \cellcolor{skala_highlight}%
\gmtkcell{1.18}{\goldmedal{}%
}
    & \gmtkcell{5.80}{}
    & \gmtkcell{2.73}{\bronzemedal{}%
}
    & \gmtkcell{2.98}{}
    & \gmtkcell{1.62}{\silvermedal{}%
}
    \\
    MB16-43 & {\scriptsize Decomposition energies of artificial molecules}
    & \gmtkcell{27.03}{}
    & \gmtkcell{14.69}{\bronzemedal{}%
}
    & \gmtkcell{35.97}{}
    & \cellcolor{skala_highlight}%
\gmtkcell{2.57}{\goldmedal{}%
}
    & \gmtkcell{25.28}{}
    & \gmtkcell{15.58}{}
    & \gmtkcell{33.06}{}
    & \gmtkcell{14.52}{\silvermedal{}%
}
    \\
    PArel & {\scriptsize Relative energies in protonated isomers}
    & \gmtkcell{1.53}{}
    & \gmtkcell{1.55}{}
    & \gmtkcell{1.36}{}
    & \cellcolor{skala_highlight}%
\gmtkcell{0.51}{\goldmedal{}%
}
    & \gmtkcell{1.17}{}
    & \gmtkcell{0.97}{}
    & \gmtkcell{0.63}{\bronzemedal{}%
}
    & \gmtkcell{0.59}{\silvermedal{}%
}
    \\
    RSE43\rule[-0.5em]{0pt}{0.5em} & {\scriptsize Radical-stabilization energies}
    & \gmtkcell{2.32}{}
    & \gmtkcell{1.52}{}
    & \gmtkcell{2.11}{}
    & \cellcolor{skala_highlight}%
\gmtkcell{0.37}{\goldmedal{}%
}
    & \gmtkcell{1.72}{}
    & \gmtkcell{0.63}{\silvermedal{}%
}
    & \gmtkcell{0.98}{}
    & \gmtkcell{0.77}{\bronzemedal{}%
}
    \\
    \multicolumn{5}{@{}l}{\textbf{Barrier heights}\rule[-0.3ex]{0pt}{0.3ex}} & \cellcolor{skala_highlight} & \multicolumn{4}{l}{} \\
    BH76 & {\scriptsize Barrier heights of various reaction types}
    & \gmtkcell{7.84}{}
    & \gmtkcell{7.01}{}
    & \gmtkcell{4.15}{}
    & \cellcolor{skala_highlight}%
\gmtkcell{1.75}{}
    & \gmtkcell{5.10}{}
    & \gmtkcell{1.22}{\goldmedal{}%
}
    & \gmtkcell{1.69}{\bronzemedal{}%
}
    & \gmtkcell{1.33}{\silvermedal{}%
}
    \\
    BHDIV10 & {\scriptsize Diverse reaction barrier heights}
    & \gmtkcell{7.81}{}
    & \gmtkcell{6.12}{}
    & \gmtkcell{2.90}{}
    & \cellcolor{skala_highlight}%
\gmtkcell{0.53}{\goldmedal{}%
}
    & \gmtkcell{3.21}{}
    & \gmtkcell{1.02}{\bronzemedal{}%
}
    & \gmtkcell{0.87}{\silvermedal{}%
}
    & \gmtkcell{1.28}{}
    \\
    BHPERI & {\scriptsize Barrier heights of pericyclic reactions}
    & \gmtkcell{6.27}{}
    & \gmtkcell{4.73}{}
    & \gmtkcell{1.15}{\bronzemedal{}%
}
    & \cellcolor{skala_highlight}%
\gmtkcell{1.51}{}
    & \gmtkcell{1.15}{\silvermedal{}%
}
    & \gmtkcell{1.35}{}
    & \gmtkcell{2.09}{}
    & \gmtkcell{1.13}{\goldmedal{}%
}
    \\
    BHROT27 & {\scriptsize Barrier heights for rotation around single bonds}
    & \gmtkcell{0.37}{}
    & \gmtkcell{0.76}{}
    & \gmtkcell{0.69}{}
    & \cellcolor{skala_highlight}%
\gmtkcell{0.29}{\silvermedal{}%
}
    & \gmtkcell{0.41}{}
    & \gmtkcell{0.36}{}
    & \gmtkcell{0.31}{\bronzemedal{}%
}
    & \gmtkcell{0.22}{\goldmedal{}%
}
    \\
    INV24 & {\scriptsize Inversion/racemization barrier heights}
    & \gmtkcell{1.91}{}
    & \gmtkcell{1.13}{\bronzemedal{}%
}
    & \gmtkcell{1.22}{}
    & \cellcolor{skala_highlight}%
\gmtkcell{1.03}{\silvermedal{}%
}
    & \gmtkcell{1.02}{\goldmedal{}%
}
    & \gmtkcell{1.56}{}
    & \gmtkcell{1.66}{}
    & \gmtkcell{1.29}{}
    \\
    PX13 & {\scriptsize Proton-exchange barrier in H${}_2$O, NH${}_3$, and HF clusters}
    & \gmtkcell{8.75}{}
    & \gmtkcell{8.95}{}
    & \gmtkcell{1.09}{\silvermedal{}%
}
    & \cellcolor{skala_highlight}%
\gmtkcell{1.06}{\goldmedal{}%
}
    & \gmtkcell{4.34}{}
    & \gmtkcell{5.32}{}
    & \gmtkcell{2.54}{}
    & \gmtkcell{1.88}{\bronzemedal{}%
}
    \\
    WCPT18\rule[-0.5em]{0pt}{0.5em} & {\scriptsize Proton-transfer barriers with water or no catalyst}
    & \gmtkcell{7.25}{}
    & \gmtkcell{6.18}{}
    & \gmtkcell{1.49}{\bronzemedal{}%
}
    & \cellcolor{skala_highlight}%
\gmtkcell{0.41}{\goldmedal{}%
}
    & \gmtkcell{2.30}{}
    & \gmtkcell{1.90}{}
    & \gmtkcell{1.73}{}
    & \gmtkcell{1.41}{\silvermedal{}%
}
    \\
    \multicolumn{5}{@{}l}{\textbf{Intermolecular NCIs}\rule[-0.3ex]{0pt}{0.3ex}} & \cellcolor{skala_highlight} & \multicolumn{4}{l}{} \\
    ADIM6 & {\scriptsize Interaction energies of $n$-alkane dimers}
    & \gmtkcell{0.22}{}
    & \gmtkcell{0.14}{}
    & \gmtkcell{0.14}{}
    & \cellcolor{skala_highlight}%
\gmtkcell{0.08}{\goldmedal{}%
}
    & \gmtkcell{0.12}{\bronzemedal{}%
}
    & \gmtkcell{0.29}{}
    & \gmtkcell{0.12}{}
    & \gmtkcell{0.11}{\silvermedal{}%
}
    \\
    AHB21 & {\scriptsize Interaction energies in anion-neutral dimers}
    & \gmtkcell{1.07}{}
    & \gmtkcell{1.23}{}
    & \gmtkcell{0.47}{}
    & \cellcolor{skala_highlight}%
\gmtkcell{0.41}{}
    & \gmtkcell{0.30}{\silvermedal{}%
}
    & \gmtkcell{0.94}{}
    & \gmtkcell{0.34}{\bronzemedal{}%
}
    & \gmtkcell{0.28}{\goldmedal{}%
}
    \\
    CARBHB12 & {\scriptsize Hydrogen-bonded complexes of carbene analogs}
    & \gmtkcell{1.12}{}
    & \gmtkcell{1.21}{}
    & \gmtkcell{0.30}{\bronzemedal{}%
}
    & \cellcolor{skala_highlight}%
\gmtkcell{0.40}{}
    & \gmtkcell{0.90}{}
    & \gmtkcell{0.26}{\silvermedal{}%
}
    & \gmtkcell{0.35}{}
    & \gmtkcell{0.20}{\goldmedal{}%
}
    \\
    CHB6 & {\scriptsize Interaction energies in cation-neutral dimers}
    & \gmtkcell{0.87}{}
    & \gmtkcell{0.40}{\goldmedal{}%
}
    & \gmtkcell{0.80}{\silvermedal{}%
}
    & \cellcolor{skala_highlight}%
\gmtkcell{1.29}{}
    & \gmtkcell{1.42}{}
    & \gmtkcell{1.45}{}
    & \gmtkcell{0.86}{\bronzemedal{}%
}
    & \gmtkcell{0.96}{}
    \\
    HAL59 & {\scriptsize Binding energies in halogenated dimers (incl. halogen bonds)}
    & \gmtkcell{0.73}{}
    & \gmtkcell{0.83}{}
    & \gmtkcell{0.47}{}
    & \cellcolor{skala_highlight}%
\gmtkcell{0.23}{\goldmedal{}%
}
    & \gmtkcell{0.58}{}
    & \gmtkcell{0.35}{}
    & \gmtkcell{0.30}{\bronzemedal{}%
}
    & \gmtkcell{0.28}{\silvermedal{}%
}
    \\
    HEAVY28 & {\scriptsize Noncovalent interaction energies of heavy element hydrides}
    & \gmtkcell{0.29}{}
    & \gmtkcell{0.25}{}
    & \gmtkcell{0.23}{\bronzemedal{}%
}
    & \cellcolor{skala_highlight}%
\gmtkcell{0.36}{}
    & \gmtkcell{0.34}{}
    & \gmtkcell{0.33}{}
    & \gmtkcell{0.18}{\goldmedal{}%
}
    & \gmtkcell{0.18}{\silvermedal{}%
}
    \\
    IL16 & {\scriptsize Interaction energies in anion-cation dimers}
    & \gmtkcell{0.87}{}
    & \gmtkcell{0.57}{\bronzemedal{}%
}
    & \gmtkcell{0.52}{\goldmedal{}%
}
    & \cellcolor{skala_highlight}%
\gmtkcell{0.57}{}
    & \gmtkcell{0.84}{}
    & \gmtkcell{0.53}{\silvermedal{}%
}
    & \gmtkcell{1.05}{}
    & \gmtkcell{0.94}{}
    \\
    PNICO23 & {\scriptsize Interaction energies in pnictogen-containing dimers}
    & \gmtkcell{0.88}{}
    & \gmtkcell{0.86}{}
    & \gmtkcell{0.22}{\bronzemedal{}%
}
    & \cellcolor{skala_highlight}%
\gmtkcell{0.19}{\silvermedal{}%
}
    & \gmtkcell{0.48}{}
    & \gmtkcell{0.29}{}
    & \gmtkcell{0.18}{\goldmedal{}%
}
    & \gmtkcell{0.26}{}
    \\
    RG18 & {\scriptsize Interaction energies in rare-gas complexes}
    & \gmtkcell{0.09}{\bronzemedal{}%
}
    & \gmtkcell{0.16}{}
    & \gmtkcell{0.07}{\goldmedal{}%
}
    & \cellcolor{skala_highlight}%
\gmtkcell{0.10}{}
    & \gmtkcell{0.14}{}
    & \gmtkcell{0.24}{}
    & \gmtkcell{0.10}{}
    & \gmtkcell{0.08}{\silvermedal{}%
}
    \\
    S22 & {\scriptsize Binding energies of noncovalently bound dimers}
    & \gmtkcell{0.43}{}
    & \gmtkcell{0.30}{}
    & \gmtkcell{0.23}{\bronzemedal{}%
}
    & \cellcolor{skala_highlight}%
\gmtkcell{0.14}{\goldmedal{}%
}
    & \gmtkcell{0.32}{}
    & \gmtkcell{0.35}{}
    & \gmtkcell{0.21}{\silvermedal{}%
}
    & \gmtkcell{0.24}{}
    \\
    S66 & {\scriptsize Binding energies of noncovalently bound dimers}
    & \gmtkcell{0.28}{}
    & \gmtkcell{0.31}{}
    & \gmtkcell{0.13}{\bronzemedal{}%
}
    & \cellcolor{skala_highlight}%
\gmtkcell{0.07}{\goldmedal{}%
}
    & \gmtkcell{0.27}{}
    & \gmtkcell{0.24}{}
    & \gmtkcell{0.12}{\silvermedal{}%
}
    & \gmtkcell{0.14}{}
    \\
    WATER27\rule[-0.5em]{0pt}{0.5em} & {\scriptsize Binding energies in (H${}_2$O)${}_n$, H${}^+$(H${}_2$O)${}_n$ and OH${}^-$(H${}_2$O)${}_n$}
    & \gmtkcell{3.37}{}
    & \gmtkcell{7.27}{}
    & \gmtkcell{0.77}{\goldmedal{}%
}
    & \cellcolor{skala_highlight}%
\gmtkcell{1.25}{\bronzemedal{}%
}
    & \gmtkcell{4.10}{}
    & \gmtkcell{3.83}{}
    & \gmtkcell{1.56}{}
    & \gmtkcell{1.10}{\silvermedal{}%
}
    \\
    \multicolumn{5}{@{}l}{\textbf{Intramolecular NCIs}\rule[-0.3ex]{0pt}{0.3ex}} & \cellcolor{skala_highlight} & \multicolumn{4}{l}{} \\
    ACONF & {\scriptsize Relative energies of alkane conformers}
    & \gmtkcell{0.10}{}
    & \gmtkcell{0.15}{}
    & \gmtkcell{0.15}{}
    & \cellcolor{skala_highlight}%
\gmtkcell{0.07}{}
    & \gmtkcell{0.06}{\silvermedal{}%
}
    & \gmtkcell{0.24}{}
    & \gmtkcell{0.02}{\goldmedal{}%
}
    & \gmtkcell{0.06}{\bronzemedal{}%
}
    \\
    AMINO20x4 & {\scriptsize Relative energies in amino acid conformers}
    & \gmtkcell{0.36}{}
    & \gmtkcell{0.20}{}
    & \gmtkcell{0.23}{}
    & \cellcolor{skala_highlight}%
\gmtkcell{0.18}{\goldmedal{}%
}
    & \gmtkcell{0.21}{}
    & \gmtkcell{0.29}{}
    & \gmtkcell{0.19}{\silvermedal{}%
}
    & \gmtkcell{0.19}{\bronzemedal{}%
}
    \\
    BUT14DIOL & {\scriptsize Relative energies in butane-1,4-diol conformers}
    & \gmtkcell{0.31}{}
    & \gmtkcell{0.25}{}
    & \gmtkcell{0.14}{}
    & \cellcolor{skala_highlight}%
\gmtkcell{0.06}{\bronzemedal{}%
}
    & \gmtkcell{0.30}{}
    & \gmtkcell{0.13}{}
    & \gmtkcell{0.04}{\goldmedal{}%
}
    & \gmtkcell{0.05}{\silvermedal{}%
}
    \\
    ICONF & {\scriptsize Relative energies in conformers of inorganic systems}
    & \gmtkcell{0.32}{}
    & \gmtkcell{0.30}{}
    & \gmtkcell{0.29}{}
    & \cellcolor{skala_highlight}%
\gmtkcell{0.18}{\silvermedal{}%
}
    & \gmtkcell{0.28}{}
    & \gmtkcell{0.31}{}
    & \gmtkcell{0.25}{\bronzemedal{}%
}
    & \gmtkcell{0.15}{\goldmedal{}%
}
    \\
    IDISP & {\scriptsize Intramolecular dispersion interactions}
    & \gmtkcell{3.03}{}
    & \gmtkcell{2.48}{}
    & \gmtkcell{3.18}{}
    & \cellcolor{skala_highlight}%
\gmtkcell{0.60}{\goldmedal{}%
}
    & \gmtkcell{3.59}{}
    & \gmtkcell{2.07}{\bronzemedal{}%
}
    & \gmtkcell{2.59}{}
    & \gmtkcell{1.63}{\silvermedal{}%
}
    \\
    MCONF & {\scriptsize Relative energies in melatonin conformers}
    & \gmtkcell{0.44}{}
    & \gmtkcell{0.46}{}
    & \gmtkcell{0.34}{}
    & \cellcolor{skala_highlight}%
\gmtkcell{0.27}{\bronzemedal{}%
}
    & \gmtkcell{0.22}{\goldmedal{}%
}
    & \gmtkcell{0.56}{}
    & \gmtkcell{0.25}{\silvermedal{}%
}
    & \gmtkcell{0.38}{}
    \\
    PCONF21 & {\scriptsize Relative energies in tri- and tetrapeptide conformers}
    & \gmtkcell{0.86}{}
    & \gmtkcell{0.46}{\bronzemedal{}%
}
    & \gmtkcell{0.82}{}
    & \cellcolor{skala_highlight}%
\gmtkcell{0.35}{\silvermedal{}%
}
    & \gmtkcell{0.52}{}
    & \gmtkcell{1.10}{}
    & \gmtkcell{0.32}{\goldmedal{}%
}
    & \gmtkcell{0.63}{}
    \\
    SCONF & {\scriptsize Relative energies of sugar conformers}
    & \gmtkcell{0.54}{}
    & \gmtkcell{0.54}{}
    & \gmtkcell{0.18}{}
    & \cellcolor{skala_highlight}%
\gmtkcell{0.09}{\goldmedal{}%
}
    & \gmtkcell{0.30}{}
    & \gmtkcell{0.27}{}
    & \gmtkcell{0.15}{\silvermedal{}%
}
    & \gmtkcell{0.17}{\bronzemedal{}%
}
    \\
    UPU23\rule[-\ras]{0pt}{\ras} & {\scriptsize Relative energies between RNA-backbone conformers}
    & \gmtkcell{0.48}{}
    & \gmtkcell{0.39}{\silvermedal{}%
}
    & \gmtkcell{0.45}{\bronzemedal{}%
}
    & \cellcolor{skala_highlight}%
\gmtkcell{0.35}{\goldmedal{}%
}
    & \gmtkcell{0.62}{}
    & \gmtkcell{0.52}{}
    & \gmtkcell{0.54}{}
    & \gmtkcell{0.46}{}
    \\
    \midrule
    \rule{0pt}{\dimexpr\rbs/2+\ht\strutbox\relax} &
    & \makebox[\gmtknpairwd][r]{0\,\goldmedal{}}
    & \makebox[\gmtknpairwd][r]{1\,\goldmedal{}}
    & \makebox[\gmtknpairwd][r]{4\,\goldmedal{}}
    & \cellcolor{skala_highlight}%
\makebox[\gmtknpairwd][r]{32\,\goldmedal{}}
    & \makebox[\gmtknpairwd][r]{3\,\goldmedal{}}
    & \makebox[\gmtknpairwd][r]{3\,\goldmedal{}}
    & \makebox[\gmtknpairwd][r]{6\,\goldmedal{}}
    & \makebox[\gmtknpairwd][r]{6\,\goldmedal{}}
    \\
    &
    & \makebox[\gmtknpairwd][r]{0\,\silvermedal{}}
    & \makebox[\gmtknpairwd][r]{2\,\silvermedal{}}
    & \makebox[\gmtknpairwd][r]{4\,\silvermedal{}}
    & \cellcolor{skala_highlight}%
\makebox[\gmtknpairwd][r]{6\,\silvermedal{}}
    & \makebox[\gmtknpairwd][r]{3\,\silvermedal{}}
    & \makebox[\gmtknpairwd][r]{10\,\silvermedal{}}
    & \makebox[\gmtknpairwd][r]{8\,\silvermedal{}}
    & \makebox[\gmtknpairwd][r]{22\,\silvermedal{}}
    \\
    \rule[-\ras]{0pt}{\ras} &
    & \makebox[\gmtknpairwd][r]{1\,\bronzemedal{}}
    & \makebox[\gmtknpairwd][r]{5\,\bronzemedal{}}
    & \makebox[\gmtknpairwd][r]{9\,\bronzemedal{}}
    & \cellcolor{skala_highlight}%
\makebox[\gmtknpairwd][r]{4\,\bronzemedal{}}
    & \makebox[\gmtknpairwd][r]{2\,\bronzemedal{}}
    & \makebox[\gmtknpairwd][r]{9\,\bronzemedal{}}
    & \makebox[\gmtknpairwd][r]{14\,\bronzemedal{}}
    & \makebox[\gmtknpairwd][r]{11\,\bronzemedal{}}
    \\\midrule
    \multicolumn{2}{@{}l}{WTMAD-2\rule{0pt}{\rbs+\ht\strutbox}\rule[-\ras]{0pt}{\ras}}
    & \gmtkcell{8.35}{}
    & \gmtkcell{7.25}{}
    & \gmtkcell{5.56}{}
    & \cellcolor{skala_highlight}%
\gmtkcell{2.80}{\goldmedal{}%
}
    & \gmtkcell{6.38}{}
    & \gmtkcell{4.83}{}
    & \gmtkcell{3.96}{\bronzemedal{}%
}
    & \gmtkcell{3.23}{\silvermedal{}%
}
    \\
    \bottomrule
\end{tabular}%
}

\end{table}

\section{Accuracy of \clippy}
\label{sec:SkalaMainResults}
An XC functional is used to predict the energy and properties of {\em new} molecules: it must therefore show {\em compositional} generalization to different compounds from those seen during training. This should not be confused with the simpler {\em configurational} generalization to unseen configurations of the same system used in training.\cite{akashi_can_2025}
For this reason, as detailed in Sec.~\ref{sec:training-data}, we have removed from the training set all systems with more than two atoms whose molecular graphs match those in our two main test sets, GMTKN55 and W4-17.
For performance across main group chemistry, we test on the GMTKN55 database,\cite{goerigk_look_2017}  which is the de facto standard benchmark for electronic structure methods, comprising 55 subsets
covering five categories: basic properties, thermochemistry, kinetics, intermolecular non-covalent interactions, and
conformational energies. The overall accuracy of an electronic structure method on this broad dataset is encoded in the weighted total mean absolute deviation (WTMAD-2).\cite{goerigk_look_2017} Additionally, because our largest training dataset consists of total atomization energies (TAEs) -- the energy required to dissociate a molecule into its constituent atoms, and a stringent test for electronic-structure methods -- we also evaluate performance on the well-established W4-17 benchmark,\cite{karton_w417_2017} comprising 200 diverse atomization energies of small molecules.

The accuracy of \clippy\ for general main-group chemistry is readily apparent in Fig.~\ref{fig:jacobs-ladder}, which compares its WTMAD-2 error on GMTKN55 to that of widely used functionals, including the best-performing ones in the first three rungs of Jacob’s ladder (up to the hybrid, or $O(N^4)$, rung).\footnotemark 
\footnotetext{We do not include comparisons with local hybrid functionals, which are more computationally demanding than global hybrids in standard implementations and have not yet seen widespread adoption. Examples include the DM21 functional,\cite{kirkpatrick_pushing_2021} a machine-learned local hybrid that reports a WTMAD-2 of 3.97~kcal/mol. More recent local hybrids\cite{wodynski_localhybrid_2026} incorporate neural networks to learn portions of the functional and achieve WTMAD-2 values down to 2.47~kcal/mol. We also omit the very recent global hybrid functional COACH,\cite{liang_reaching_2026} which is not yet widely available and reports incremental improvements over $\omega$B97M-V. In all these approaches, portions of the GMTKN55 benchmark are included in the training data, and the W4-17 dataset is used in its entirety for training.} 
\clippy achieves the lowest overall WTMAD-2 error, surpassing all hybrid functionals despite operating at semi-local cost. A detailed breakdown across the 55 individual subsets of GMTKN55 in Table~\ref{tab:gmtkn55-breakdown} shows that \clippy is the best-performing functional on 32 of the 55 subsets, more than all other functionals combined. In particular, \clippy outperforms the best hybrid functional across nearly all thermochemistry and kinetics benchmarks, ranging from small to large systems. It is widely assumed that this level of accuracy requires ascending to the hybrid (or even double-hybrid) rung of Jacob's ladder; however, our results indicate that \clippy learns the necessary non-local effects directly from data, without recourse to Hartree–Fock exchange. 

Moving to atomization energies, we note that all molecular structures in our TAE training set have single-reference electronic structure character, which can be treated accurately with the thermochemical W1-F12 and W1w protocols\cite{karton_w4_2006,karton_explicitly_2012} based on CCSD(T)/CBS. The test set W4-17, by contrast, is labeled with the higher-level CCSDTQ5-based W4 protocol,\cite{karton_w4_2006,karton_w411_2011} which can also handle the multi-reference cases present in this test set.
Training on the large TAE set enables \clippy to reach chemical accuracy (MAE 0.92 kcal/mol) on the single-reference subset of W4-17 (183 reactions out of 200), with an overall MAE of 1.23 kcal/mol when we include also the multi-reference cases for which we lack training data (see Extended Data Fig.~\ref{fig:w4-17-breakdown}). In Extended Data Fig.~\ref{fig:holdout-generalization}, we evaluate \clippy\ on holdout splits of the training data spanning diverse chemical properties, demonstrating strong generalization across a broad chemical space.

In Sec.~\ref{sec:suppeval} of the Supplementary Information, we further assess aspects of practical usability, such as SCF-cycle convergence as reported in Table~\ref{tab:conv-analysis-summary} and grid-size convergence shown in Fig.~\ref{fig:grid-convergence}. As expected for a deep-learned functional, \clippy exhibits slightly less smooth behavior than traditional functionals, but all variations remain within acceptable ranges.

\begin{figure}[t!]
    \begin{subfigure}{0.55\linewidth}
        \centering
        \hspace{-.9em}%
        \includegraphics[width=1.0\linewidth]{figures/ablation_combined_skalav11.pdf}      
        \vspace{-0.2cm}
        \caption{Ablations on Diet GMTKN55}
        \label{fig:model-data-ablation}
    \end{subfigure}
    \begin{subfigure}{0.43\linewidth}
        \centering
        \hspace{-.9em}%
        \includegraphics[width=1.0\linewidth]{figures/tc_at_gamma1_final.pdf}
        \vspace{-0.2cm}
        \caption{Analysis of $T_c$ positivity constraint}
        \label{fig:tc_data_ablation_main}    
    \end{subfigure}
    \caption{\textbf{Model and data ablations}. (a) Left panel: Accuracy of \clippy's non-local architecture compared with its local branch only, trained on the full dataset in Extended Data Table~\ref{tab:training-data}. Right panel: Data composition ablation from Extended Data Table~\ref{tab:training-data}: set A contains thermochemistry, basic properties, distorted geometries and transition metals. Set B contains barrier heights and reactions. Set C contains NCI and conformations. Details on the specific datasets that belong to A, B, C are shown in Extended Data Table~\ref{tab:training-data}.
    In both ablations, for each setting we trained five models using different random seeds (affecting model initialization and data ordering); each dot represents one model instance and bars show the average. Note that for the local models, two seeds yield nearly identical results, causing their dots to overlap. SCF fine-tuning was done for 20,000 steps, and evaluation was performed on the smaller Diet GMTKN55.\cite{gould_diet_2018} For details on SCF convergence in the ablation studies, see Sec.~\ref{sec:scf-retry}. (b) The kinetic correlation component $T_\cor$ of $\Exc$ for several atoms, when training on MSR-ACC/TAE25\cite{ehlert_accurate_2025} only, and when training on the full set. Reference data\cite{vuckovic_density_2019,vuckovic_interpolated_2017} obtained by reverse-engineering CCSD densities in a pure Kohn-Sham framework for H, He, Ne and Ar are also shown for comparison.}
\end{figure}

\subsection{The importance of learning nonlocal interactions}
\label{sec:local-ablation}
The non-local branch of our architecture is remarkably lightweight --- \clippy comprises just $385,217$ parameters in total, with $331,265$ allocated to the local branch. This compact design is crucial for maintaining scalability. We therefore quantify the performance gains enabled by learned non-locality. In the left panel of Fig.~\ref{fig:model-data-ablation}, we compare models that are purely local to models sharing the same non-local architecture as \clippy, both on the full training set of Extended Data Table~\ref{tab:training-data}. For each setting, we trained multiple instances with different random seeds, and evaluated them on Diet GMTKN55,\cite{gould_diet_2018} a representative subset designed to approximate the WTMAD-2 metric on the full GMTKN55 dataset. 
This ablation shows the significant improvement enabled by the non-local branch, which reduces the average WTMAD-2 error across 5 seeds by about 50\%.

The clear performance gap between the local and non-local models confirms that \clippy captures essential non-local effects directly from data, without relying on Hartree–Fock exchange.
This finding is further supported by the molecular ``thermometer'' recently introduced by \citet{grimme_molecular_2026} for quantifying effective non-local exchange. 
As shown in Extended Data Fig.~\ref{fig:non-local-exchange}, \clippyOnePointOne yields an effective non-local exchange of 53.7\%, close to the 60\% reference value, outperforming several hybrid functionals.\cite{grimme_molecular_2026} As another prototypical application where Hartree–Fock exchange is widely considered essential, we report in Extended Data Fig.~\ref{fig:w1-sn2-bh} barrier heights for S$_{\rm N}$2 reactions:\cite{karton_w1sn2bh_2026} once again, \clippyOnePointOne achieves accuracy comparable to the best hybrid functionals, with errors roughly 2.5 times smaller than those of the best meta‑GGAs.

\subsection{\clippy's accuracy improves systematically with training data}
\label{sec:data-ablation}
The right panel in figure~\ref{fig:model-data-ablation} reports an ablation study on training data composition, which shows systematic improvement of \clippy on Diet GMTKN55 as we add more diverse chemistry in training. We find that training on thermochemistry, basic properties, distortions and transition metals (set A) on average yields GGA-level accuracy on Diet GMTKN55. Adding reactions and barrier heights (set B) brings a modest improvement, and the inclusion of non-covalent interactions and conformers (set C) further reduces the WTMAD-2 error (see Extended Data Table~\ref{tab:training-data} for the detailed composition of each set). 
At each stage, increasing the size and diversity of the training data both lowers the average WTMAD-2 error and reduces the inter-seed variance, reflecting better coverage of the relevant chemical space.

\subsection{The emergence of learned physical constraints with training data}
Exact constraints of the XC functional have been pivotal in guiding the approximations that made DFT practical for thousands of applications in chemistry and materials science.\cite{kaplan_predictive_2022} Many are ingeniously built in by design,\cite{perdew_generalized_1996,sun_strongly_2015} lending robustness to the functionals that include them. In \clippy, we imposed only minimal constraints to maximize model flexibility, making it interesting to explore whether physical constraints can emerge from data.

In Fig.~\ref{fig:tc_data_ablation_main}, we assess whether the model learns to satisfy the positivity of $T_\cor$,\cite{levy_hellmannfeynman_1985} the kinetic correlation component of $\Exc$. This constraint reflects the physical principle that electron correlation increases the kinetic energy, as electrons move faster to avoid one another due to their mutual Coulomb repulsion.
We evaluate $T_\cor^\theta$  for several atoms and perform a data ablation analysis. When \clippy\ is trained only on the MSR-ACC/TAE25 dataset,\cite{ehlert_accurate_2025}
the constraint is largely violated. In contrast, training on the full dataset leads to consistent satisfaction of the constraint, yielding values of $T_\cor^\theta$  not far from values obtained by reverse-engineering CCSD densities.\cite{vuckovic_interpolated_2017,vuckovic_density_2019}

A subtle point is that our functional is evaluated within the generalized Kohn-Sham (GKS) framework, where the kinetic correlation energy is not constrained to match its definition in the pure Kohn-Sham scheme, and may therefore differ from the corresponding exact $T_\cor$. A meta-GGA within the generalized Kohn-Sham has much more freedom to move energy contributions between the kinetic energy and the XC functional. Notably, with our large training set \clippy yields $T_\cor$ values close to pure Kohn-Sham ones, despite the additional freedom of the GKS framework.
For more detailed results and discussion, the reader is referred to Sec.~\ref{sec:exact-constraints} of the Supplementary Information.

\section{Beyond energies: Densities and equilibrium geometries}
After establishing the accuracy of \clippy on energies, we assess the quality of its predicted electron densities via dipole moments and evaluate its predicted equilibrium geometries.
\begin{figure}[t!]
    \begin{minipage}[b]{0.5\linewidth}
    \centering
    \subfloat[Effect of fine-tuning with SCF on dipole and reaction error]{
        \centering
        \includegraphics[width=\linewidth]{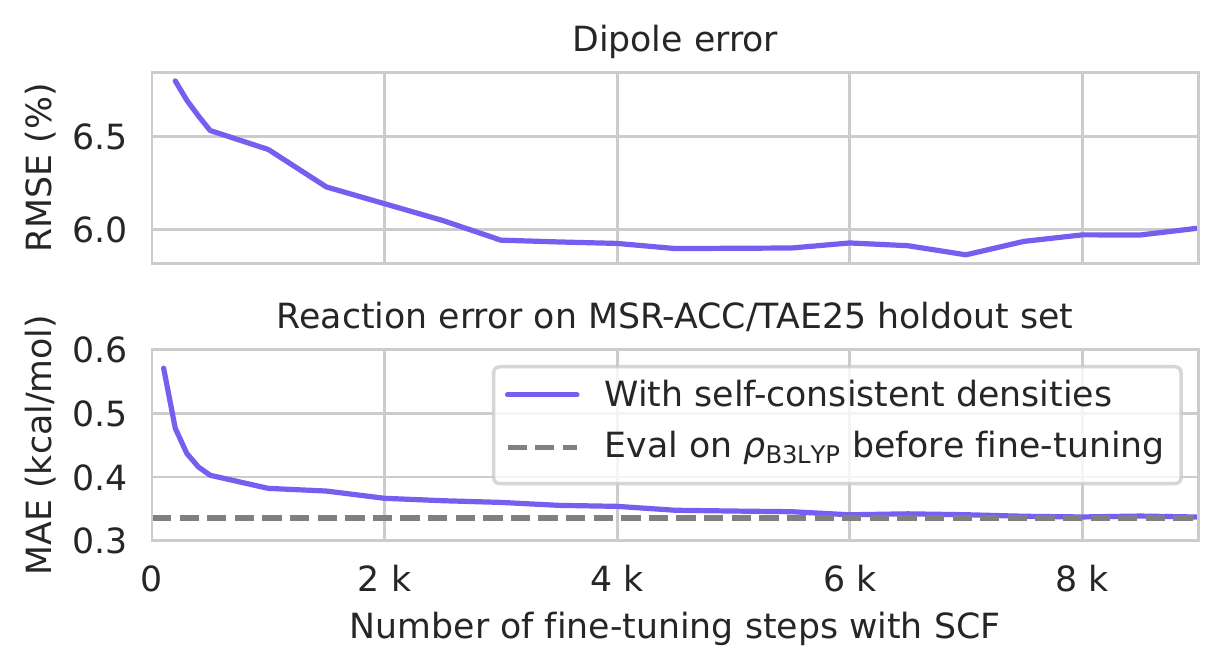}
        \vspace{-0.2cm}
        \label{fig:finetune-gap-dipole}
    }
    \end{minipage}
    \begin{minipage}[b]{0.5\linewidth}
    \centering
    \subfloat[Dipole errors\cite{hait_how_2018a} of various functionals]{
        \addtolength{\tabcolsep}{-3pt}
        \tablefontsize 
\begin{tabular}{lllllll}
\toprule
 & \multicolumn{3}{l}{RMSE} & Mean & Max & Std \\\cmidrule(lr){2-4}
 & All & NSP & SP &  &  &  \\
\midrule
revPBE & 11.79 & 9.87 & 14.78 & 8.45 & 42.96 & 8.26 \\
r\textsuperscript2SCAN & 8.95 & 8.16 & 10.28 & 6.07 & 32.27 & 6.60 \\
B97M-V & 11.54 & 10.19 & 13.74 & 7.21 & 67.70 & 9.03 \\
B3LYP & 7.09 & 6.60 & 7.95 & 4.13 & 45.94 & 5.79 \\
M06-2X & 7.73 & 7.69 & 7.79 & 4.19 & 61.56 & 6.52 \\
\(\omega\)B97X-V & 5.18 & 4.64 & 6.08 & 3.58 & 18.81 & 3.76 \\
\(\omega\)B97M-V & 5.84 & 5.44 & 6.54 & 3.81 & 32.31 & 4.44 \\
\textbf{\clippyOnePointOne} & 4.43 & 4.23 & 4.80 & 3.04 & 18.75 & 3.24 \\
\bottomrule
\end{tabular}

        \vspace{3.1em}
        \vspace{-0.2cm}
        \label{tab:dipole}
    }
    \end{minipage}\\
    \vspace{0.5em}\\
    \begin{minipage}[b]{\linewidth}
        \centering
        \include{tables/geometry}
        \subcaption{Geometry optimization errors of various functionals}
        \label{tab:geometry}
    \end{minipage}
    \caption{\textbf{Dipoles and equilibrium geometries.}
        (a): 
            The dipole error\cite{hait_how_2018a} (top) and reaction error on the holdout set of total atomization energies (bottom) during the fine-tuning of \clippy with self-consistent densities instead of B3LYP densities.
            For dipoles, we show the root-mean-square of the regularized error\cite{hait_how_2018a} (RMSE, top) and, for energies, the mean absolute error (MAE, bottom). 
        (b):
            Comparison of \clippy's dipole errors on the benchmark of \citet{hait_how_2018a} against reference functionals.
        (c): 
            Geometry optimization results. The geometries in the benchmark datasets LMGB35, HMGB11\cite{grimme_consistent_2015}, CCse21\cite{piccardo_semiexperimental_2015} and W4-11-GEOM\cite{spackman_basis_2016} were optimized for various functionals and \clippy. Average absolute errors are reported for bond lengths (in \r{A}ngstrom) and bond angles (in degrees) compared to the ground truth values from these datasets. Box plots show the quartiles of the error distribution. 
        }
        \vspace{-0.3cm}
\end{figure}

\subsection{Densities} \label{sec:densities}
Recall that the energy error from a KS DFT calculation with a given XC functional can be decomposed into two components: a {\em functional error}, which is the error the functional would make if evaluated on the exact density, and a {\em density-driven error}, which is the error the exact functional would make when evaluated on the self-consistent density of the approximate functional.\cite{kim_understanding_2013,mezei_electron_2017,gubler_accuracy_2025} 
These two errors can compensate each other,\cite{kanungo_unconventional_2024,kaplan_how_2024} yielding XC approximations that improve energies by worsening their SCF densities, ``straying from the path toward the exact functional'', quoting \citet{medvedev_density_2017} 

We pretrain our functional to reproduce accurate wavefunction energy differences when evaluated on the approximate densities $\rho_{\rm B3LYP}$. We then fine-tune it for a small number of steps using on-the-fly SCF densities, thereby closing the gap between the accuracy learned on $\rho_{\rm B3LYP}$ and that obtained on the self-consistent densities $\rho_{\rm \clippy}$ produced by \clippy, as detailed in Sec.~\ref{subsupp:SCF-finetuning} of the Supplementary Information. To ensure that this SCF fine-tuning does not rely on error compensation, we monitor the quality of the resulting densities by comparing their dipole moments against a highly accurate dataset of 151 structures.\cite{hait_how_2018a}
Fig.~\ref{fig:finetune-gap-dipole} illustrates this process. The lower panel shows, as a function of the number of fine-tuning steps, the gap between the accuracy learned on B3LYP densities and the actual SCF evaluation of \clippy\ on the TAE holdout set; the dashed line indicates the accuracy on $\rho_{\rm B3LYP}$ at the end of the pretraining phase. The upper panel reports the errors of the SCF \clippy\ density during fine-tuning. In an initial phase ($\lesssim 2000$ fine-tuning steps), SCF convergence is erratic and errors are reported only for the molecules that converge. After $\sim 2000$ steps, all systems converge and the error reflects the full dataset. From this point on, both energies and densities improve simultaneously -- an emerging exact property consistent with the behaviour of the exact functional.
The final \clippy\ error on the dipole dataset falls well below that of B3LYP and outperforms the best range-separated hybrid functionals, as shown in Table~\ref{tab:dipole}.

\subsection{Equilibrium geometries}
One of the primary uses of DFT is the prediction of equilibrium molecular structures by relaxing the nuclear coordinates to their lowest-energy configuration. We find that when \clippy\ is trained only on equilibrium structures and reaction paths, its predicted equilibrium geometries reach, at best, GGA-level accuracy.
Motivated by the idea that \clippy's accuracy can be systematically improved by augmenting the training data to target specific regions of chemical space and properties, we generated approximately 20,000 slightly distorted structures from a subset of molecules in the MSR-ACC/TAE dataset. For each distorted structure, we computed the energy difference with respect to equilibrium using the accurate W1w protocol,\cite{karton_w4_2006,karton_explicitly_2012}  thereby constructing the MSR-ACC/Distortion dataset. Because the MSR-ACC/TAE equilibrium structures were originally obtained with an approximate functional\cite{ehlert_accurate_2025} (B3LYP), the resulting W1w labels for these small distortions effectively encode information about minima at CCSD(T)/CBS accuracy. 

We benchmark geometries optimized with \clippy\ against (semi-)experimental datasets comprising light main-group bond lengths (LMGB35),\cite{grimme_consistent_2015} heavy main-group bond lengths (HMGB11),\cite{grimme_consistent_2015} as well as bond lengths and angles for small molecules in the CCse21\cite{piccardo_semiexperimental_2015} and the CCSD(T)/CBS W4-11-GEOM datasets.\cite{spackman_basis_2016} The results, shown in Table~\ref{tab:geometry}, compare \clippy not only with functionals across different rungs of Jacob’s ladder, but also with the semi-empirical GFN2-xTB method.\cite{bannwarth_gfn2xtban_2019} \clippy achieves an accuracy on par with, or exceeding, that of the best hybrid functionals.

In practical applications, \clippy would typically be used to compute both equilibrium structures and reaction energies. In Sec.~\ref{sec:SkalaMainResults}, however, we follow standard practice by evaluating all functionals, including \clippy, on fixed accurate geometries. This raises the question of the additional error introduced when a functional is used to determine both quantities. This error can be quantified using the GEO\cite{vuckovic_quantifying_2020} metric, defined as the energy difference obtained when a given functional is evaluated on its own optimized geometry versus a reference geometry. For \clippy, we find a very small average GEO of $\sim 0.02$~kcal/mol on the W4-11-GEOM dataset.
Details of the evaluation protocol are provided in Sec.~\ref{sec:suppgeom} of the Supplementary Information, alongside GEO results for other functionals and an ablation study that shows the impact of the MSR-ACC/Distortion training set on geometry accuracy (Sec.~\ref{sec:geometry-ablation}).  Besides improving structure prediction, the MSR-ACC/Distortion dataset allows \clippy to generalize to the prediction of highly strained conformations of larger molecules such as adenosine, benzylpenicillin, and efavirenz,\cite{brew_wiggle150_2025} as shown in Extended Data Figs.~\ref{fig:wiggle150-breakdown} and \ref{fig:wiggle150-vib-ablation}.

\section{Computational Efficiency and Integration into Production Codes}
\label{sec:cost}

\begin{figure}[tb]
    \begin{minipage}[b]{0.55\linewidth}
    \vspace{-0.5cm}%
    \includegraphics[width=\linewidth]{figures/scaling_exc_only.pdf}
    \subcaption{}
    \end{minipage}%
    \hspace{0.05\linewidth}%
    \begin{minipage}[b]{0.4\linewidth}
    \includegraphics[width=\linewidth]{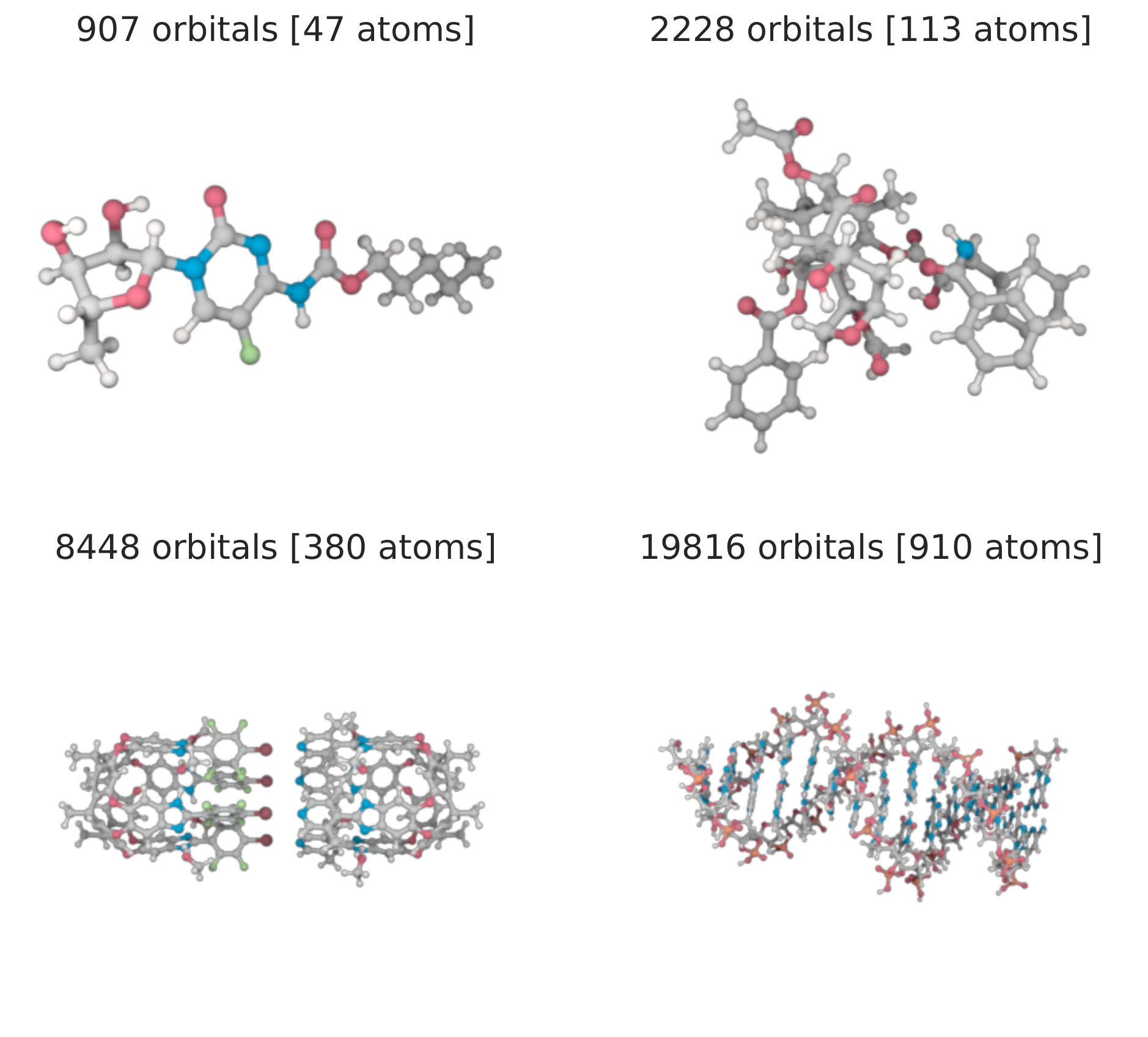}
    \vspace{-1.cm}
    \subcaption{}
    \end{minipage}%
    \caption{\textbf{Computational cost.}
    (a) Wall-clock time for exchange-correlation integration per SCF step as a function of molecular size, measured using GauXC.
    Only the XC numerical integration is timed: eval\_exc\_vxc for pure functionals (SVWN5, revPBE, r\textsuperscript{2}SCAN) and eval\_exc\_vxc plus sn-LinK exact exchange for hybrids (B3LYP, M06-2X).
    Coulomb matrix construction, diagonalization, and SCF convergence are not included.
    All computations use the def2-TZVP basis set, GauXC's GM3 integration grid (Mura–Knowles radial quadrature, Robust pruning), and unrestricted Kohn-Sham (UKS).
    The left panel shows GPU timings on an NVIDIA A100; the right panel shows CPU timings using 16 OpenMP threads on an AMD EPYC 7V13.
    We were unable to run hybrid functionals for molecules with over 2000 basis functions with GauXC's sn-LinK.
    Lines show fitted power laws $a N_{\rm orbitals}^n$; the fitted exponent $n$ is reported in each legend entry.
    Each data point is the minimum of five timed iterations following three warmup iterations.
    (b) A sample of the molecules used for evaluating timings of \clippy with GauXC.
    See Sec.~\ref{sec:eval-sets} for more information on all molecules.}
    \label{fig:timing}
\end{figure}
The computational cost of quantum chemistry methods is commonly expressed through its asymptotic scaling with system size, $O(f(N))$, a convention we have followed so far in this paper. In practice, the prefactors of that scaling can differ by orders of magnitude between methods, the cost can be dominated by other terms for smaller to medium-sized molecules, and the bottleneck for scaling may be memory rather than compute.\cite{scuseria_linear_1999} Moreover, hardware-specific optimizations and algorithmic advances continue to lower the cost in practice.  

For all these reasons, although our architecture design ensures that \clippy has the same asymptotic scaling as meta-GGA semi-local DFT, we have to empirically verify its prefactor and actual cost as system size increases. A relevant analogy to clarify why this is crucially important is the following. 
At the hybrid rung of Jacob's ladder we find both global hybrids and local hybrids. For global hybrids, the XC functional contains a fixed fraction of exact exchange evaluated on the basis set, which can be made computationally efficient but lacks universality, as different systems often require different optimal fractions. In contrast, the more flexible local hybrids allow the fraction of exact exchange to be position-dependent, requiring the exchange to be evaluated on the grid.
Although both have the same asymptotic scaling, the latter has a much larger prefactor, with basic implementations being even more expensive than the double hybrids of the next rung. Despite impressive progress over the last decade,\cite{bahmann_efficient_2015,laqua_efficient_2018,kussmann_highly_2021} less costly implementations of local hybrids are still rare,\cite{kaupp_next_2024} which has prevented their widespread use so far. 

For five popular XC functionals (SVWN5, revPBE, r\textsuperscript{2}SCAN, B3LYP, M06-2X) and \clippy, spanning semilocal to hybrid functionals, we have measured the cost of the XC numerical integration together with sn-LinK\cite{neese_efficient_2009,laqua_efficient_2018,laqua_accelerating_2021,williams-young_distributed_2023} exact exchange for hybrid functionals (B3LYP and M06-2X) (\cref{fig:timing}). This isolates the cost of XC evaluation from the other components of a KS calculation, which can vary substantially across different codes and implementations. On GPU, \clippy\ exhibits a computational cost within 30\% of the meta-GGA functional r\textsuperscript{2}SCAN. On CPU, we observe an approximately twofold prefactor relative to r\textsuperscript{2}SCAN for smaller systems, which decreases toward unity as the cost becomes dominated by atomic orbital evaluation for larger systems. Hybrid functionals remain roughly a factor of 3–6 more expensive than \clippy.

Having confirmed empirically that the \clippy\ architecture is scalable while retaining the computational cost of semilocal DFT, a second barrier to its widespread adoption is its availability in production codes. Many modern electronic structure packages implement standard XC functionals via libraries such as libXC\cite{lehtola_recent_2018}, which process grid points independently. The grid-level non-locality introduced by \clippy\ makes it incompatible with this paradigm. However, our results suggest that achieving substantially higher accuracy in DFT without increasing computational cost will require enabling deep-learning functionals that learn non-locality, such as \clippy, to be readily implemented and optimized within production codes.
For quantum chemistry codes based on Gaussian basis sets, the GauXC library\cite{williams-young_efficient_2020,williams-young_distributed_2023} provides a practical pathway to support this new class of functionals. GauXC is a modern, modular C\texttt{++} library for the evaluation of exchange-correlation (XC) energies and their derivatives in Gaussian basis sets: given a density matrix, it performs all grid operations internally and returns the required integrated quantities in the chosen basis, using efficient, scalable distributed-memory XC and $K$ integrators for both CPU and accelerator-based (GPU) architectures. Codes such as Psi4 \cite{smith_si4_2020} and CP2K \cite{kuhne_cp2k_2020} are already integrating \clippy\ through GauXC. Another possible pathway, which works for any basis set, is direct integration from the released model checkpoint, and this is used in the integration to PySCF \cite{sun_pyscf_2018} and is available to other DFT codes through Python, C/C\texttt{++}, and Fortran interfaces.

\section{Conclusions}
We have presented a deep learning-based exchange-correlation functional, \clippy, that marks a significant step forward in the long-standing quest for a general-purpose, chemically accurate and computationally efficient density functional. 
By assembling a large-scale training set, designing a training protocol that scales efficiently, and developing a non-local architecture with low inference cost, we demonstrate that it is possible to learn non-local quantum mechanical effects from simple semi-local inputs, without sacrificing the favourable scaling of semi-local DFT. \clippy outperforms state-of-the-art range-separated hybrid functionals across a broad range of main-group chemistry at the cost of semi-local DFT, demonstrating that deep learning can break the traditional accuracy-cost trade-off that has governed XC functional development for decades.

As we continue to expand the training dataset to encompass a broader range of chemical phenomena, we expect \clippy to systematically improve in both accuracy and generality.
As a first step toward extending beyond main‑group chemistry, we have incorporated minimal atomic information for 3d and 4d transition metals into the training set. As a preview, we probe \clippy on selected transition‑metal benchmarks (see Extended Data Table~\ref{fig:eval-tm}), where, even with this minimal training signal, it already performs on par with established functionals, indicating strong transferability; further gains, however, will require substantially expanding the scale and diversity of high‑accuracy data for transition‑metal chemistry.
A key challenge in this endeavor will be extending coverage to multi‑reference and strongly correlated systems, where generating accurate reference data at scale remains an obstacle that will require new scientific and computational advances to overcome.

One of the key advantages of learning an XC functional --- rather than a force field --- from high-accuracy wavefunction data is that the KS framework inherently captures the dominant energy contributions needed to generalize across unseen elements and larger systems. The XC functional represents a smaller correction term, and by embedding a minimal set of physical constraints into its design, \clippy remains robust: it generalizes with high accuracy to most thermochemistry benchmark sets and, in the worst-case scenario, defaults to the performance of standard semi-local DFT. This makes learning an accurate XC functional a compelling strategy for transferring the accuracy of wavefunction methods from small systems to the medium-large ones accessible to DFT. In turn, the learned functional can be used to generate high-quality data for larger systems, enabling the training of force fields and other models that typically require large-scale datasets generated with DFT. This creates a cascade of accuracy transfer across scales, with the potential to transform the predictive power of computational chemistry.

\printbibliography
\numgdef\resetnum{\csuse{blx@labelnumber@\therefsection}+1}
\end{refsection}

\clearpage

\begin{refsection}
\section{Methods}
Here we expand on details of the model architecture and the training data. The Supplementary Information contains further detailed information on the model (\cref{sec:supp-model}), training details (\cref{supp:training-details}), training data (\cref{sec:training-data-details}), evaluation protocols (\cref{sec:suppeval}) and additional results (\cref{sec:suppadditional}). 

\subsection{\clippy: A model for scalable non-local representation learning}
\label{sec:architecture}
\clippy's enhancement factor in Eq.~\eqref{eq:Exc-basic} is a non-local functional modeled with a deep neural network that takes as input a set of semi-local, density-dependent features $\mathbf{x}[\dens]$ from the standard meta-generalized-gradient approximation (meta-GGA) rung, which has $O(N^3)$ scaling, and which are represented on the large, irregular numerical integration grid used in DFT. 
The challenge here is to design an accurate XC functional that models intricate non-local interactions across the grid in order to achieve the accuracy that is often only attainable by more expensive functionals of a higher rung, while maintaining a computational cost comparable to functionals from the meta-GGA rung. While a naive solution with all-to-all communication across the grid would enable non-local representation learning, it is not a scalable design, since the cost of doing so on grids of the order of $10^4\sim10^6$ points quickly grows out of control. Instead, \clippy introduces a second coarse grid with far fewer points,\cite{gao_learning_2024} which acts as an intermediary layer through which the points on the finer grid can communicate.  

Extended Data Fig.~\ref{fig:architecture-diagram-nonlocal} shows the overall schematic of the neural network architecture. Starting from the input meta-GGA features, the seven semi-local inputs are log-transformed, followed by a small multilayer perceptron (MLP) that acts strictly locally on each grid point.
The MLP is applied twice, once to each spin-ordering of the transformed features, followed by an averaging operation. This yields a spin-symmetrized semi-local hidden representation that serves as input for the rest of the model. By making the hidden layer spin symmetric before feeding it through any non-local computation across the grid, we avoid having to run the more expensive part of the non-local neural network twice, saving computational cost. 

Before the spin-symmetrized features are passed into the non-local interaction model, they are projected to a lower-dimensional hidden vector. Subsequently, the coarse points collect non-local information from the fine grid, analogous to the accumulation of multipole moments. More specifically, for each coarse point, the local hidden features on the atom's integration grid are projected onto a product of radial basis functions and spherical harmonics that depend on the distance vector between the coarse and fine points, followed by an integration over the atom's grid points. After downsampling, we apply an equivariant linear mixing and a symmetric contraction\cite{batatia_mace_2022} on the coarse features to capture higher-body-order correlations. Using the same product basis of radial and spherical components, we then construct functions that when evaluated on the finer grid yield non-local hidden features on each fine grid point, which are invariant with respect to the Euclidean symmetry. This downsampling-processing-upsampling cycle is repeated through a stack of $L$ non-local layers, with each layer maintaining the hidden dimensionality through a skip connection. The use of atom-centered grids, where each grid point is associated with a single atomic center, ensures that the model scales linearly in the number of atoms.

Finally, the features from the last non-local layer are processed through a purely local MLP and projected down to a scalar value per grid point. The scalar value is passed through a scaled sigmoid activation function with a range between 0 and 2,\cite{kirkpatrick_pushing_2021} yielding a bounded enhancement factor that enforces the Lieb-Oxford lower bound.\cite{lieb_lower_1979} This is substituted into the discretized equivalent of Eq.~\eqref{eq:Exc-basic} to yield the predicted $\Exc^{\theta}[\dens]$.   

In Sec.~\ref{sub:supplmodeltheory} we show that the hidden features on the coarse grid can be interpreted as multipole moments, and that the non-local module has the expressivity to model multi-body interactions through the symmetric contraction.\cite{dusson_atomic_2021, batatia_mace_2022} While in principle the non-local module has the ability to approximate non-local interactions independently of where the coarse points are placed, we take advantage of the structure of integration grids typically used in DFT --- centered around the atomic centers --- and place the coarse-grid points on such atomic centers. For more details on the neural network architecture, see Sec.~\ref{sec:supp-model} in the Supplementary Information.

\subsection{Training data}
\label{sec:training-data}
Our training data comprise $\sim$400k reaction energies (Extended Data Tab.~\ref{tab:training-data}) computed at the CCSD(T)/CBS level of theory or higher, as detailed in Sec.~\ref{sec:training-data-details} in the Supplementary Information.
The largest subset of our training data is composed of $\sim$120k diverse total atomization energies for general molecules with up to nine non-hydrogen atoms (MSR-ACC/TAE).
A subset of structures up to five non-hydrogen atoms from this dataset consisting of a single stable molecular fragment are released as the MSR-ACC/TAE25 dataset, described in \citet{ehlert_accurate_2025} As part of MSR-ACC we also include data that goes beyond atomization energies. In particular, we include proton affinities (MSR-ACC/PA), ionization potentials (MSR-ACC/IP), and electron affinities (MSR-ACC/EA) for the molecules in the MSR-ACC/TAE dataset. 
For conformational energies, the MSR-ACC/Conf dataset includes all conformers within a 10 kcal/mol energy window of the molecules in MSR-ACC/TAE.
For intermolecular non-covalent interaction energies, the MSR-ACC/NCI includes equilibrium structures from small van-der-Waals clusters with up to six monomers.
Additionally, we include the full potential energy surface of the water dimer in the MSR-ACC/Water subset.
The MSR-ACC/Distortion subset includes the structures from MSR-ACC/TAE distorted along their normal modes.
To cover kinetics, the MSR-ACC/Reactions dataset comprises elementary steps of reactions of small organic molecules with up to 13 non-hydrogen atoms, including both transition states and points along the reaction pathways.
In addition to the molecular sets, we include atomic datasets of electron affinities (EAs) and ionization potentials (IPs) --- including double and triple IPs --- for elements up to argon as well as total atomic energies to gauge total energies.

We extend this in-house training data with a few publicly available datasets. In particular, we add the 14 linear and cyclic carbon clusters from the W4-CC dataset \cite{karton_atomization_2009}. We draw on the relatively abundant publicly available data and select five datasets from the NCIAtlas collection (D442x10, SH250x10, R739x5, HB300SPXx10, IHB100x10),\cite{rezac_noncovalent_2020,rezac_noncovalent_2020a,kriz_noncovalent_2021,kriz_noncovalent_2022,rezac_noncovalent_2022} and we also include the DES370K non-covalent interaction data computed at CCSD(T)/\ensuremath{\delta}CBS(aug-cc-pVQZ) level of theory\cite{donchev_quantum_2021} and the NCIBLIND dataset of dissociation curves.\cite{taylor_blind_2016}
The reaction dataset BH9 is included with its forward and reverse barriers as well as the full reactions,\cite{prasad_bh9_2022} and we include the MB2061 dataset of decomposition energies for artificial closed-shell mindless molecules.\cite{gasevic_chemical_2025}
Furthermore, we include small transition metal systems covering dimer bond energies TMD10\cite{liang_goldstandard_2025} and DAPD,\cite{chan_dapd_2023} and ionization potentials and spin splittings of atomic systems in 3d4dIPSS.\cite{liang_goldstandard_2025}

\paragraph{Test set subtraction:} We removed the overlap with the test sets GMTKN55\cite{goerigk_look_2017} and W4-17\cite{karton_w417_2017} based on the molecular graphs of all systems with more than two atoms, with the exceptions of W4-CC, MB2061, atomic sets, small transition metal systems, and MSR-ACC/Water.

We determine molecular graphs (with undetermined bond order) from the bond model of GFN-FF.\cite{spicher_robust_2020}
For all the subsets except the NCI ones and BH9, we removed all reactions that contain any molecule that contains any covalently connected subgraph found in any molecule in W4-17 and GMTKN55 (some molecules in W4-17 and GMTKN55 are not recognized as fully connected by GFN-FF).
For the NCI sets and the BH9 sets, we subtract GMTKN55 from the training data by removing all reactions that share the same set of molecules (defined by the GFN-FF graph) with the same stoichiometric ratios.
This prevents leakage of W4-17 into the trained model and minimizes the overlap with GMTKN55. 

After the test set subtraction, we further reserved a small subset of the data to be used for model selection.
The holdout part of MSR-ACC/TAE is released as part of MSR-ACC/TAE25\cite{ehlert_accurate_2025} (which is a subset of MSR-ACC/TAE).

\paragraph{Density features:} As explained in Sec.~\ref{sec:learning-xc-functional}, in the pre-training phase we evaluate our model at fixed densities using B3LYP\cite{becke_densityfunctional_1993,stephens_initio_1994} using the def2-QZVP basis set\cite{weigend_balanced_2005} or the ma-def2-QZVP basis set\cite{zheng_minimally_2011} if the molecule is part of a reaction that contains anions. 
Using the fixed densities, we compute the relative energy of a reaction from the B3LYP total energies by replacing the B3LYP XC energies with the XC energies predicted with our functional.
To regularize the model with respect to numerical variations on the grid, we use four distinct radial methods (Treutler, Mura–Knowles, Gauss-Chebyshev and Delley) at PySCF\cite{sun_pyscf_2018} grid level 1 during pre-training and grid level 3 during fine-tuning, and in both stages we augment training with different partitioning weights as detailed in \cref{subsupp:partition-augmentation}.

\subsection{Data availability}
The training data are summarized in Extended Data Table~\ref{tab:training-data}.
Molecular structures from the in-house generated MSR-ACC/TAE dataset consisting of a single stable molecular fragment that are not derived from the GDB9 set are released publicly as the MSR-ACC/TAE25 dataset, described in \citet{ehlert_accurate_2025}  %
The in-house generated datasets MSR-ACC/Conf, MSR-ACC/PA, MSR-ACC/IP, MSR-ACC/EA, MSR-ACC/Distortion, MSR-ACC/NCI, and MSR-ACC/Reactions are not released publicly, but detailed information on their generation protocol is given in Sec.~\ref{sec:training-data} and Sec.~\ref{sec:training-data-details}. 
Structures of MSR-ACC/Water are publicly available\cite{smith_revised_2016} and relabeled in-house.
W4-CC,\cite{karton_atomization_2009}
D442x10,\cite{rezac_noncovalent_2022}
SH250x10,\cite{kriz_noncovalent_2022}
R739x5,\cite{kriz_noncovalent_2021}
HB300SPXx10,\cite{rezac_noncovalent_2020}
IHB100x10,\cite{rezac_noncovalent_2020}
MB2061,\cite{gasevic_chemical_2025}
NCIBLIND,\cite{taylor_blind_2016}
DES370K,\cite{donchev_quantum_2021}
and BH9\cite{prasad_bh9_2022} are all publicly available.

The evaluation benchmark sets 
W4-17,\cite{karton_w417_2017} 
MOR41,\cite{dohm_comprehensive_2018}
ROST61,\cite{maurer_assessing_2021}
MOBH35,\cite{semidalas_mobh35_2022} 3dTMV,\cite{neugebauer_benchmarkquality_2023}
3d4dIPSS,\cite{liang_goldstandard_2025}
DAPD,\cite{chan_dapd_2023}
TMD10,\cite{liang_goldstandard_2025}
Wiggle150,\cite{brew_wiggle150_2025}
W1-S$_{\rm N}$2-BH,\cite{karton_w1sn2bh_2026}
and GMTKN55,\cite{goerigk_look_2017}
the dipole moment evaluation dataset,\cite{hait_how_2018}
and the geometry optimization datasets LMGB35,\cite{grimme_consistent_2015} HMGB11,\cite{grimme_consistent_2015} W4-11-GEOM,\cite{karton_w411_2011} and CCse21\cite{piccardo_semiexperimental_2015} are all publicly available.
The molecular structures for the computational cost results are described in \cref{sec:eval-sets} (Fig.~\ref{fig:si-timing-dataset}), and are collected from the following publicly available sources: Grimme,\cite{grimme_exploration_2019} S30L,\cite{sure_comprehensive_2015} HS13L,\cite{gorges_reliable_2022} and NCI16L.\cite{gorges_efficient_2023}

\subsection{Code availability}
The \clippy model and inference code are available under MIT license at \url{https://github.com/microsoft/skala}.
The repository contains the PyTorch implementation of the \clippy model and its hookups to quantum chemistry packages PySCF,\cite{sun_pyscf_2018} GPU4PySCF\cite{li_introducing_2025} and ASE.\cite{hjorthlarsen_atomic_2017}
The \clippy model is also served in Azure AI Foundry at \url{https://ai.azure.com/catalog/models/Skala}, where the SCF evaluation is implemented using Accelerated DFT\cite{ju_acceleration_2024} (inference on GPU) and GauXC.\cite{petrone_efficient_2018,williams-young_efficient_2020,williams-young_distributed_2023}

\section*{Acknowledgments}
We thank Maik Riechert, Hannes Schulz and Eray Inanc for supporting our engineering infrastructure and Kenji Takeda and the Microsoft Accelerator team for their crucial role in designing and executing our data generation campaign.
Furthermore, we are grateful for feedback and support from Marwin Segler, Frank No\'e, Jia Zhang, Bonnie Kruft, Rachel Howard, Rosa de Rosa and Bev Baker. We also thank Jan Gerit Brandenburg, Nicola Marzari and John P.\ Perdew for insightful feedback on an earlier version of this manuscript.

3D visualizations in this paper were rendered with Mitsuba 3.\cite{jakob_mitsuba_2022}

\section*{Author contributions}
Conceptualization: C.-W.H., G.L., T.V., D.P.K., S.E., K.J.H.G., J.H., R.v.d.B., P.G.-G.
Methodology: G.L., C.-W.H., T.V., D.P.K., S.E., B.M., S.-O.K., J.H., R.v.d.B., P.G.-G.
Software: T.V., C.-W.H., G.L., D.P.K., S.E., S.L., J.H., R.v.d.B., D.G., X.W., L.H., R.C.Z., Ab.K., K.J.H.G.
Validation: D.P.K., C.-W.H., T.V., G.L., S.E., J.H., R.v.d.B.
Formal analysis: C.-W.H., T.V., G.L., D.P.K., K.J.H.G., S.E., J.H., R.v.d.B.
Investigation: C.-W.H., T.V., G.L., D.P.K., S.E., J.H., R.v.d.B., P.G.-G., D.G., Y.C., D.B.W.-Y., Ab.K., K.J.H.G., M.S., W.P.B., S.B.
Resources: C.M.B., D.G.
Data generation \& curation: J.H., S.E., D.P.K., Am.K., S.L., J.G.T., K.J.H.G., C.-W.H., G.L., T.V., G.N.C.S., P.B.S., R.v.d.B., P.G.-G.
Writing -- Original Draft: G.L., C.-W.H., T.V., D.P.K., S.E., K.J.H.G., J.H., R.v.d.B., P.G.-G.
Writing -- Review \& Editing: G.L., C.-W.H., T.V., C.M.B., J.H., P.B.S., R.v.d.B., P.G.-G.
Visualization: T.V., C.-W.H., S.E., D.P.K., W.P.B., M.S., G.N.C.S., P.B.S., R.v.d.B.
Supervision: J.H., R.v.d.B., P.G.-G.
Project administration: R.S.

\section*{Competing interests}
All authors declare employment by Microsoft while engaged in the research for this manuscript.
G.L., C.-W.H., T.V., D.P.K., S.E., S.L., K.J.H.G., D.G., M.S., W.P.B., R.S., J.H., R.v.d.B., and P.G.-G. have filed a patent application for the model and training pipeline described in this article.

\printbibliography[resetnumbers=\resetnum, title=Method References]
\numgdef\resetnum{\csuse{blx@labelnumber@\therefsection}+1}
\clearpage

\renewcommand{\figurename}{Extended Data Figure}
\renewcommand{\tablename}{Extended Data Table}

\begin{figure}[ht!]
    \begin{subfigure}{\linewidth}
        \tikzset{
        roundedbox/.style={draw, rounded corners, inner sep=8pt},
        operation/.style={roundedbox, fill=gray!15},
        arrowstyle/.style={-{Stealth[length=2mm]}, rounded corners},
        arrowlabel/.style={fill=white, inner sep=1pt, font=\footnotesize\sffamily}
        }
        \hspace{-1em}
        \scalebox{0.56}{
            \input{figures/architecture/architecture_overview}
        }
    \caption{\clippy architecture overview}
    \end{subfigure}
    \begin{subfigure}{\linewidth}
        \tikzset{
        roundedbox/.style={draw, rounded corners, inner sep=8pt},
        operation/.style={roundedbox, fill=gray!15},
        arrowstyle/.style={-{Stealth[length=2mm]}, rounded corners},
        arrowlabel/.style={fill=white, inner sep=1pt, font=\footnotesize\sffamily}
        }
        \hspace{-1.7em}
        \scalebox{0.52}{
            \input{figures/architecture/non_local}
        }
        \caption{Non-local interaction model}
    \end{subfigure}
    \caption{\textbf{\clippy's architecture} (a): 
        Overview of the architecture modules. $G$ is the size of the DFT integration grid, $G'$ denotes atomic sub-grids, and $C$ is the number of coarse points (atomic centers).
        We log-transform 7 meta-GGA features, apply the same MLP to both spin-orderings, and average to generate spin-order-invariant hidden features.
        The grid is then partitioned into atomic sub-grids, after which three non-local interaction layers exchange information between grid points through coarse points at nuclear positions, as detailed in~(b).
        A final MLP produces an enhancement factor, which is reassembled via soft partitioning, multiplied by a scale function based on the local density, and integrated over the grid to obtain $\Exc^\params$.
    (b): The non-local interaction model (one layer shown; applied $\times 3$).
        For each coarse point and spherical harmonic level $\ell=0,1,2,3$, local grid features are projected onto $2\ell+1$ spherical harmonics and 16 radial basis functions via a tensor product, then integrated to form coarsened features.
        These are mixed by an equivariant linear layer and a symmetric contraction on each coarse point, capturing higher-body-order correlations.
        The mixed features become coefficients for the spherical harmonic basis and are projected back to the grid via a second tensor product, producing functions that capture non-local interactions.
        These are scaled by $\exp(-\dens)$ (total density) and combined with the input features via a skip connection, preserving the hidden dimension across layers.
        }
    \label{fig:architecture-diagram-nonlocal}
\end{figure}

\clearpage

\begin{figure}[ht!]
    \centering
    
\scriptsize
\addtolength{\tabcolsep}{-3pt}
\begin{tabularx}{\textwidth}{l *{2}{r}@{ }r@{}*{2}{r} *{2}{l}}
    \toprule
    \multirow{2}{*}{Dataset}
    & \multicolumn{4}{l}{Number of reactions}
    & Avg.\ |E|
    & \multirow{2}{*}{Elements}
    & \multirow{2}{*}{Description}
    \\

    & \multicolumn{1}{l}{\emph{Full}}
    & \multicolumn{2}{l}{\emph{Training}}
    & \multicolumn{1}{l}{\emph{Cat.}}
    & \textcolor{gray}{[kcal/mol]}
    &
    \\
    \cmidrule(lr){1-1}
    \cmidrule(lr){2-5} \cmidrule(lr){6-6} \cmidrule(lr){7-7} \cmidrule(lr){8-8}

    MSR-ACC/\\
    \hspace{1em}
    TAE
    & 121321
    & 118485
    & (97.7\%)
    & A
    & 813.1
    & H, Li--F, Na--Cl
    & Total atomization energies
    \\
    \hspace{1em}
    Conf
    & 41294
    & 40612
    & (98.3\%)
    & C
    & 1.8
    & H, Li--F, Na--Cl
    & Conformational energies
    \\
    \hspace{1em}
    PA
    & 21359
    & 20177
    & (94.5\%)
    & A
    & 222.4
    & H, Li--F, Na--Cl
    & Proton affinities
    \\
    \hspace{1em}
    IP
    & 17028
    & 16039
    & (94.2\%)
    & A
    & 162.3
    & H, Li--F, Na--Cl
    & Ionization potentials
    \\
    \hspace{1em}
    EA
    & 13666
    & 12879
    & (94.2\%)
    & A
    & 35.1
    & H, Li--F, Na--Cl
    & Electron affinities
    \\
    \hspace{1em}
    Reactions
    & 154255
    & 114349
    & (74.1\%)
    & B
    & 38.6
    & H, Li--F, Na--Cl
    & Reaction paths
    \\
    \hspace{1em}
    Water
    & 2308
    & 2193
    & (95.0\%)
    & C
    & 5.8
    & H, O
    & Water dimer structures
    \\
    \hspace{1em}
    Distortions
    & 19429
    & 18895
    & (97.3\%)
    & A
    & 6.9
    & H, Li--F, Na--Cl
    & Distorted equilibrium structures
    \\
    \hspace{1em}
    NCI
    & 3948
    & 3747
    & (94.9\%)
    & C
    & 10.2
    & H, Li--Be, C--Ar
    & Non-covalent clusters
    \\[1ex]
    Atomic/\\
    \hspace{1em}
    TOT
    & 16
    & 16
    & (100.0\%)
    & A
    & ---
    & H--He, B--Ar
    & Atomic total energies
    \\
    \hspace{1em}
    EA
    & 11
    & 11
    & (100.0\%)
    & A
    & 33.6
    & H, B--C, O--F, Na, Al--Cl
    & Atomic electron affinities
    \\
    \hspace{1em}
    IP
    & 43
    & 43
    & (100.0\%)
    & A
    & 667.2
    & He, B--Ar
    & Atomic ionization potentials
    \\
    \hspace{1em}
    TM
    & 54
    & 35
    & (64.8\%)
    & A
    & 102.3
    & H, B--F, Al--Cl, Sc--Zn, Br, Y--Cd
    & Atomic transition metals
    \\[1ex]
    NCIAtlas/\\
    \hspace{1em}
    D442x10
    & 4420
    & 4150
    & (93.9\%)
    & C
    & 1.4
    & H--He, B--Ne, P--Ar, Br--Kr, I--Xe
    & Dispersion interactions
    \\
    \hspace{1em}
    R739x5
    & 3695
    & 3263
    & (88.3\%)
    & C
    & 1.1
    & H--He, B--Ne, P--Ar, Br--Kr, I--Xe
    & Repulsive contacts
    \\
    \hspace{1em}
    HB300SPXx10
    & 3000
    & 2840
    & (94.7\%)
    & C
    & 3.2
    & H, C--F, P--Cl, Br, I
    & Hydrogen bonds
    \\
    \hspace{1em}
    IHB100x10
    & 1000
    & 896
    & (89.6\%)
    & C
    & 15.6
    & H, C--O
    & Hydrogen bonds
    \\
    \hspace{1em}
    SH250x10
    & 2500
    & 2290
    & (91.6\%)
    & C
    & 4.0
    & H, C--F, P--Cl, As--Br, I
    & Sigma-hole contacts
    \\[1ex]
    Public/\\
    \hspace{1em}
    GDB9-W1F12
    & 3366
    & 3117
    & (92.6\%)
    & A
    & 1382.4
    & H, C--F
    & Total atomization energies
    \\
    \hspace{1em}
    W4-CC
    & 14
    & 14
    & (100.0\%)
    & A
    & 745.1
    & C
    & Total atomization energies
    \\
    \hspace{1em}
    MB2061
    & 2063
    & 1960
    & (95.0\%)
    & A
    & 306.4
    & H, Li--F, Na--Cl, K--Ca, Ga--Br, Rb--Sr, In--I
    & Decomposition energies
    \\
    \hspace{1em}
    NCIBLIND
    & 80
    & 48
    & (60.0\%)
    & C
    & 2.5
    & H, C--O
    & Dissociation curves
    \\
    \hspace{1em}
    DES370k
    & 58950
    & 43074
    & (73.1\%)
    & C
    & 8.4
    & H--Li, C--Mg, S--Ca
    & Non-covalent interactions
    \\
    \hspace{1em}
    BH9
    & 1345
    & 291
    & (21.6\%)
    & B
    & 19.7
    & H, B--F, Si--Cl
    & Reaction barriers
    \\[1ex]

    \midrule
    Total
    & 475165
    & 409424
    & (86.2\%)
    &
    &
    & H--Xe
    &

    \\
    \bottomrule
\end{tabularx}

    \caption{\textbf{Training datasets.} The table shows the original number of labels and the effective count after removing overlap with the GMTKN55 and W4-17 test sets and splitting off validation sets. The category for each dataset (A, B, C) defines the subsets for the data ablation study of Sec.~\ref{sec:data-ablation}.}
    \label{tab:training-data}
\end{figure}

\clearpage
\begin{figure}[ht!]
    \begin{minipage}[b]{1.00\linewidth}
        \centering
        \input{tables/holdouts/all_holdouts_grid}
        \vspace{1.5em}
        \subcaption{Holdout error distributions}
        \label{fig:holdout-error}
    \end{minipage}
    \begin{minipage}[b]{0.4\linewidth}
        \centering
        \includegraphics[width=\linewidth]{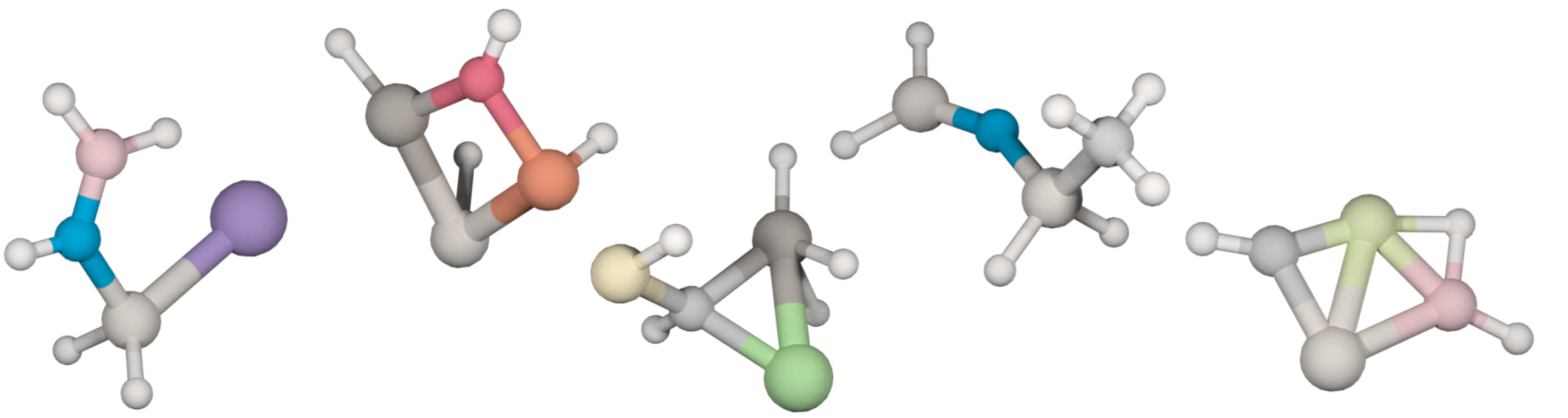}
        \vspace{1.5em}
        \subcaption{MSR-ACC/TAE25 holdout molecules}
        \label{fig:reaction-errors-tae-molecules}
    \end{minipage}
    \hspace{0.02\linewidth}
    \begin{minipage}[b]{0.55\linewidth}
        \centering
        \includegraphics[width=\linewidth]{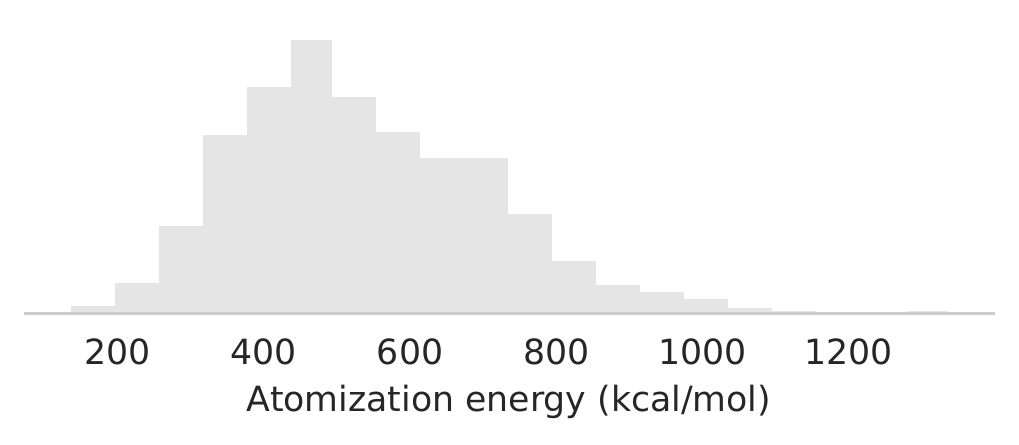}
        \subcaption{TAE distribution}
        \label{fig:reaction-errors-tae-histogram}
    \end{minipage}
    \caption{\textbf{Generalization on holdout data.}
        (a) Absolute error distributions on holdout splits of the training data, with MAE in kcal/mol shown to the left of each bar; the numbers in parentheses indicate holdout set sizes.
        Structures of all holdout sets shown in this figure are publicly available, with the caveat that Water2510 is relabeled as part of MSR-ACC/Water.
        (b) Example molecules from the MSR-ACC/TAE25 holdout set. (c) The distribution of total atomization energies in this set.
        These results confirm that \clippy generalizes well beyond its training data across diverse chemical properties.
    }
    \label{fig:holdout-generalization}
\end{figure}

\clearpage
\begin{figure}[ht!]
    \centering
    \begin{minipage}[b]{.7\linewidth}
        \centering
        \providecolor{col_even}{HTML}{EAEEF3}
\tablefontsize 
\begingroup
\ifdefined\rbs\else\newlength{\rbs}\fi\setlength{\rbs}{4pt}
\ifdefined\ras\else\newlength{\ras}\fi\setlength{\ras}{4pt}
\ifdefined\rhs\else\newlength{\rhs}\fi\setlength{\rhs}{6pt}
\setlength{\aboverulesep}{0pt}
\setlength{\belowrulesep}{0pt}
\begin{tabular}{l >{\columncolor{col_even}}r r r}
\toprule
 & Full & Single ref. & Multi ref.\rule[-0.5\rhs]{0pt}{\ht\strutbox+\rhs} \\
\midrule
revPBE & \rule{0pt}{\rbs+\ht\strutbox} 8.09 & 7.22 & 17.44 \\
r\textsuperscript2SCAN & 4.51 & 3.84 & 11.67 \\
B97M-V & 2.90 & 2.52 & 6.95 \\
B3LYP & 3.86 & 3.73 & 5.25 \\
M06-2X & 3.01 & 2.39 & 9.70 \\
\(\omega\)B97X-V & 2.59 & 2.14 & 7.39 \\
\(\omega\)B97M-V & 2.06 & 1.66 & 6.36 \\
\textbf{\clippyOnePointOne} & 1.23 & 0.92 & 4.61\rule[-\ras]{0pt}{\ras} \\
\bottomrule
\end{tabular}
\endgroup

        \vspace{0.5em}
        \subcaption{W4-17 total atomization energies}
        \label{fig:w4-17-breakdown}
    \end{minipage}
    \vspace{1.5em}
    \begin{minipage}[b]{.50\linewidth}
        \centering
        \providecolor{col_even}{HTML}{EAEEF3}
\tablefontsize 
\begingroup
\ifdefined\rbs\else\newlength{\rbs}\fi\setlength{\rbs}{4pt}
\ifdefined\ras\else\newlength{\ras}\fi\setlength{\ras}{4pt}
\ifdefined\rhs\else\newlength{\rhs}\fi\setlength{\rhs}{6pt}
\setlength{\aboverulesep}{0pt}
\setlength{\belowrulesep}{0pt}
\begin{tabular}{l >{\columncolor{col_even}}r r r r}
\toprule
 & Overall & EFA & BPN & ADO\rule[-0.5\rhs]{0pt}{\ht\strutbox+\rhs} \\
\midrule
revPBE & \rule{0pt}{\rbs+\ht\strutbox} 4.43 & 5.35 & 3.88 & 4.04 \\
r\textsuperscript2SCAN & 1.15 & 1.58 & 1.00 & 0.87 \\
B97M-V & 1.30 & 1.49 & 1.39 & 1.00 \\
B3LYP & 1.42 & 1.72 & 1.53 & 1.00 \\
M06-2X & 2.03 & 2.89 & 1.73 & 1.48 \\
\(\omega\)B97X-V & 2.44 & 3.00 & 2.54 & 1.79 \\
\(\omega\)B97M-V & 0.87 & 1.35 & 0.67 & 0.60 \\
\textbf{\clippyOnePointOne} & 0.73 & 0.75 & 0.86 & 0.59\rule[-\ras]{0pt}{\ras} \\
\bottomrule
\end{tabular}
\endgroup

        \vspace{15pt}
        \subcaption{Wiggle150 relative energies}
        \label{fig:wiggle150-breakdown}
    \end{minipage}
    \begin{minipage}[b]{.35\linewidth}
        \centering
        \includegraphics[width=0.98\linewidth]{figures/wiggle150-vib-ablation}
        \subcaption{Distortion data ablation}
        \label{fig:wiggle150-vib-ablation}
    \end{minipage}
    \hspace{0.04\linewidth}
    \caption{\textbf{Atomization and distortion energy prediction.}
        (a) W4-17 total atomization energy MAE (kcal/mol) on the full set, the 183 single-reference structures with \%TAE[(T)] $<$ 10\%, and the 17 multi-reference structures.\cite{karton_w417_2017}
        (b) Wiggle150\cite{brew_wiggle150_2025} MAE (kcal/mol) on the relative energy of the highly strained conformers of efavirenz, benzylpenicillin,  and adenosine, on the full set and per-molecule subsets .
        (c) Impact on the Wiggle150 MAE of including distorted geometries of small molecules (MSR-ACC/Distortion) in training, evaluated across five independent seeds per setting.
        Together, these results demonstrate that \clippy generalizes accurately to held-out atomization energies and that distortion training data substantially improves prediction of relative energies along vibrational modes.
    }
    \label{fig:atomization-distortion}
\end{figure}

\clearpage

\begin{figure}[ht!]
    \centering
    \begin{minipage}[t]{1.00\linewidth}
        \centering
        \begin{tikzpicture}
\pgfmathsetmacro{\pw}{8.0}
\pgfmathsetmacro{\ph}{4.0}

\draw[black!15, thin] (0, 0.400) -- (\pw, 0.400);
\node[anchor=east, font=\scriptsize, text=ibm_orange, inner sep=1pt] at (-0.10, 0.400) {GGA};
\draw[black!15, thin] (0, 2.050) -- (\pw, 2.050);
\node[anchor=east, font=\scriptsize, text=ibm_purple, inner sep=1pt] at (-0.10, 2.050) {meta-GGA};
\draw[black!15, thin] (0, 3.700) -- (\pw, 3.700);
\node[anchor=east, font=\scriptsize, text=ibm_blue, inner sep=1pt] at (-0.10, 3.700) {Hybrid};

\draw[black, thin, dashed] (6.629992141173109, 0.22) -- (6.629992141173109, \ph);
\node[anchor=east, font=\scriptsize, text=black, inner sep=1pt, xshift=-4pt] at (6.629992141173109, 1.06) {CCSD(T) reference};

\draw[black, thin] (0, 0) -- (\pw, 0);

\draw[thin, black!70] (0.914, -0.06) -- (0.914, 0);
\node[anchor=north, font=\scriptsize, inner sep=1pt] at (0.914, -0.08) {10};
\draw[thin, black!70] (2.058, -0.06) -- (2.058, 0);
\node[anchor=north, font=\scriptsize, inner sep=1pt] at (2.058, -0.08) {20};
\draw[thin, black!70] (3.201, -0.06) -- (3.201, 0);
\node[anchor=north, font=\scriptsize, inner sep=1pt] at (3.201, -0.08) {30};
\draw[thin, black!70] (4.344, -0.06) -- (4.344, 0);
\node[anchor=north, font=\scriptsize, inner sep=1pt] at (4.344, -0.08) {40};
\draw[thin, black!70] (5.487, -0.06) -- (5.487, 0);
\node[anchor=north, font=\scriptsize, inner sep=1pt] at (5.487, -0.08) {50};
\draw[thin, black!70] (6.630, -0.06) -- (6.630, 0);
\node[anchor=north, font=\scriptsize, inner sep=1pt] at (6.630, -0.08) {60};
\draw[thin, black!70] (7.773, -0.06) -- (7.773, 0);
\node[anchor=north, font=\scriptsize, inner sep=1pt] at (7.773, -0.08) {70};
\node[anchor=north, font=\scriptsize] at (4.3435, -0.28) {Measured non-local exchange (\%)};

\fill[ibm_orange] (1.075, 0.400) circle[radius=0.065];
\fill[ibm_purple] (0.892, 2.050) circle[radius=0.065];
\fill[ibm_purple] (1.040, 2.050) circle[radius=0.065];
\fill[ibm_pink] (5.915, 2.050) circle[radius=0.065];
\fill[ibm_blue] (2.355, 3.700) circle[radius=0.065];
\fill[ibm_blue] (5.041, 3.700) circle[radius=0.065];
\fill[ibm_blue] (6.950, 3.700) circle[radius=0.065];
\fill[ibm_blue] (7.053, 3.700) circle[radius=0.065];

\node[anchor=south west, rotate=60, font=\scriptsize, text=ibm_orange, inner sep=1pt] at ([xshift=0pt]1.115, 0.460) {revPBE};
\node[anchor=south west, rotate=60, font=\scriptsize, text=ibm_purple, inner sep=1pt] at ([xshift=0pt]0.932, 2.110) {r\textsuperscript{2}SCAN};
\node[anchor=south west, rotate=60, font=\scriptsize, text=ibm_purple, inner sep=1pt] at ([xshift=4pt]1.080, 2.110) {B97M-V};
\node[anchor=south west, rotate=60, font=\scriptsize, text=ibm_pink, inner sep=1pt] at ([xshift=0pt]5.955, 2.110) {\clippyOnePointOne};
\node[anchor=south west, rotate=60, font=\scriptsize, text=ibm_blue, inner sep=1pt] at ([xshift=0pt]2.395, 3.760) {B3LYP};
\node[anchor=south west, rotate=60, font=\scriptsize, text=ibm_blue, inner sep=1pt] at ([xshift=0pt]5.081, 3.760) {M06-2X};
\node[anchor=south west, rotate=60, font=\scriptsize, text=ibm_blue, inner sep=1pt] at ([xshift=-4pt]6.990, 3.760) {$\omega$B97M-V};
\node[anchor=south west, rotate=60, font=\scriptsize, text=ibm_blue, inner sep=1pt] at ([xshift=4pt]7.093, 3.760) {$\omega$B97X-V};
\end{tikzpicture}
        \vspace{0.5em}
        \subcaption{Measured non-local exchange}
        \label{fig:non-local-exchange}
    \end{minipage}
    \vspace{1.5em}\\
    \begin{minipage}[t]{1.00\linewidth}
        \centering
        
\tablefontsize
\begin{tabular}{l *{8}{r}}
    \toprule
    \multicolumn{1}{l}{Basis set}
    & revPBE
    & r$^2$SCAN
    & B97M-V
    & B3LYP
    & M06-2X
    & $\omega$B97X-V
    & $\omega$B97M-V
    & \textbf{\clippyOnePointOne}
    \\
    \cmidrule(lr){2-2}
    \cmidrule(lr){3-3}
    \cmidrule(lr){4-4}
    \cmidrule(lr){5-5}
    \cmidrule(lr){6-6}
    \cmidrule(lr){7-7}
    \cmidrule(lr){8-8}
    \cmidrule(lr){9-9}
    def2-TZVPPD
    & 12.72
    & 10.59
    & 7.25
    & 6.59
    & 2.07
    & 2.98
    & 1.90
    & 3.03
    \\
    def2-QZVPPD
    & 11.84
    & 9.78
    & 6.25
    & 5.79
    & 1.60
    & 3.31
    & 1.69
    & 2.44
    \\
    ma-def2-QZVP
    & 10.47
    & 8.77
    & 4.89
    & 4.66
    & 1.59
    & 3.92
    & 2.09
    & 1.99
    \\
    \bottomrule
\end{tabular}

        \vspace{0.5em}
        \subcaption{S$_{\rm N}$2 reaction barriers}
        \label{fig:w1-sn2-bh}
    \end{minipage}%
    \caption{\textbf{Non-local exchange and S$_{\rm N}$2 reaction barriers.}
        (a) Effective non-local exchange as quantified by the molecular ``thermometer'' of \citet{grimme_molecular_2026} for functionals across different rungs of Jacob's ladder (up to the hybrids). The vertical dashed line indicates the CCSD(T) reference value.
        (b) MAE (kcal/mol) on the W1-S$_{\rm N}$2-BH dataset\cite{karton_w1sn2bh_2026} of reaction barriers for S$_{\rm N}$2 reactions using different basis sets. These results together show that the non-local effects learned by \clippy are enough to accurately describe cases that typically need the more costly Hartree-Fock exchange.
    }
\end{figure}

\clearpage

\begin{figure}[ht!]
    \centering
    \begin{minipage}[t]{1.00\linewidth}
        \centering
        
\tablefontsize
\begin{tabular}{l *{8}{r}}
    \toprule
    & \multicolumn{3}{c}{Seen} & \multicolumn{5}{c}{Unseen} \\
    \cmidrule(lr){2-4} \cmidrule(lr){5-9}
    & 3d4dIPSS
    & TMD10
    & DAPD
    & MOR41
    & MOBH35
    & TMB11
    & 3dTMV
    & ROST59
    \\
    \midrule
        revPBE
        & 9.75%
        & 6.92%
        & 9.48%
        & 3.79%
        & 3.37%
        & 3.72%
        & 9.09%
        & 4.28%
        \\
        r$^2$SCAN
        & 12.52%
        & 8.48%
        & 13.03%
        & 3.22%
        & 2.62%
        & 2.59%
        & 9.12%
        & 3.20%
        \\
        B97M-V
        & 11.81%
        & 4.31%
        & 6.67%
        & 3.86%
        & 1.82%
        & 2.32%
        & 8.00%
        & 3.50%
        \\
        B3LYP
        & 7.69%
        & 6.36%
        & 3.54%
        & 4.84%
        & 2.54%
        & 1.93%
        & 7.34%
        & 3.21%
        \\
        M06-2X
        & 8.04%
        & 8.65%
        & 15.16%
        & 6.71%
        & 2.64%
        & 2.25%
        & 14.75%
        & 6.41%
        \\
        $\omega$B97X-V
        & 7.31%
        & 5.18%
        & 7.30%
        & 2.18%
        & 2.29%
        & 1.96%
        & 14.23%
        & 2.67%
        \\
        $\omega$B97M-V
        & 6.37%
        & 7.50%
        & 7.85%
        & 2.13%
        & 1.84%
        & 1.28%
        & 13.93%
        & 2.72%
        \\
        TPSSh
        & 8.85%
        & 6.43%
        & 5.95%
        & 2.83%
        & 2.01%
        & 2.43%
        & 8.44%
        & 2.68%
        \\
        \midrule
        Skala-1.1
        & 7.19%
        & 9.34%
        & 1.45%
        & 3.29%
        & 2.08%
        & 1.85%
        & 8.01%
        & 3.38%
        \\
    \bottomrule
\end{tabular}

    \end{minipage}
    \caption{\textbf{Transition metal reactions.}
        MAE (kcal/mol) for the benchmark sets 3d4dIPSS,\cite{liang_goldstandard_2025},
        TMD10,\cite{liang_goldstandard_2025},
        DAPD,\cite{chan_dapd_2023},
        MOR41,\cite{dohm_comprehensive_2018},
        MOBH35,\cite{semidalas_mobh35_2022} 
        3dTMV,\cite{neugebauer_benchmarkquality_2023}
and ROST59 (a subset of ROST61\cite{maurer_assessing_2021} with reaction numbers 24 and 52 dropped due to unconverged calculations). The sets (partially) used in training are clearly marked as ``Seen''. These results show that, with very minimal coverage of 3d and 4d systems in training, \clippy already performs on par with some of the best hybrid and meta-GGA functionals.
}
    \label{fig:eval-tm}
\end{figure}

\end{refsection}

\clearpage
\appendix
\renewcommand{\figurename}{Figure}
\renewcommand{\tablename}{Table}

\renewcommand\thepart{}
\renewcommand\partname{}
\part{\texorpdfstring{
    \begin{center} \Large
        Supplementary information: \\
        Accurate and scalable exchange-correlation with deep learning
    \end{center}
}{Supplementary information}}
\etocsettocstyle{\section*{Table of contents}\vspace{.5em}}{}
\setlength{\cftsecindent}{0pt}
\setlength{\cftsubsecindent}{0pt}
\setlength{\cftsecnumwidth}{2em}
\setlength{\cftsubsecnumwidth}{3.5em}
\setlength{\cftbeforesubsecskip}{2pt}
\localtableofcontents

\clearpage

\begin{refsection}
    
\section{Modeling the exchange-correlation functional}
\label{sec:supp-model}

In this section we summarize how we model the exchange-correlation functional. We describe the input features and the architecture.

\subsection{In theory}
The foundation of DFT is built on the insight that the ground-state energy of a many-electron system can, in principle, be expressed as a functional of the electron density alone. As explained in Sec.~\ref{sec:learning-xc-functional} in the main body of the paper, such a functional of the electron density has an unknown component, the exchange-correlation part.
Following the Kohn-Sham formalism,\cite{kohn_selfconsistent_1965} the ground-state energy of a many-electron system in a static potential $v$ can be written as \begin{equation}
E = \min_{\dens} E_{\text{tot}}[\dens], \quad E_{\text{tot}}[\dens] = \int v(r)\dens(r)\,dr + \frac{1}{2} \int \int\frac{\dens(r) \dens(r')}{\abs{r-r'}} dr dr' + T_{s}[\dens] + \Exc[\dens],\label{eq:TotEnergyDFT}
\end{equation} 
where $v(r)$ in the first term is the external potential due to the nuclei, the second term is the Hartree electrostatic energy, $T_{s}[\dens]$ is the kinetic energy of a system of non-interacting electrons with density $\dens$ and $\Exc[\dens]$ is the exchange-correlation energy of an interacting system with density $\dens$.~\cite{kohn_selfconsistent_1965}
The expression of $\Exc$ is unknown and the central challenge then becomes to find an accurate description for it. $\Exc$ is a functional of the electron density, meaning that we can define it as $\Exc: L^1(\R^3) \rightarrow \R$. $\Exc$ takes a density $\dens$ as input and it outputs a scalar $\Exc[\dens]$. There exist many different ways to produce an approximation. Many ML and traditional functionals represent the $\Exc$ as an integral of an energy density, as follows:
\begin{equation}\label{eq:exc-fun}
 \Exc^\theta[\dens] = - \frac{3}{4} \left(\frac{6}{\pi}\right)^{\frac{1}{3}}\int \left(\dens^{(\uparrow)}(r)^{4/3} + \dens^{(\downarrow)}(r)^{4/3}\right) f_\params[\mathbf{x}[\dens]](r) \,dr,   
\end{equation}
where $f_\params$ is a learnable function of a set of features $\mathbf{x}[\dens]$ called the \textit{enhancement factor}. When $f_\params=1$, the remaining terms reduce to
the Local Density Approximation (LDA) exchange functional.\cite{dirac_note_1930} 
In \clippy, the enhancement factor $f_\params$ is parameterized by a deep neural network, whose architecture is explained in detail in Sec.~\ref{sec:neural-network-architecture}.
This particular form of $\Exc$, written as a learnable enhancement factor times the LDA exchange energy density, is mainly designed to
make it easier to enforce properties that the exact $\Exc$ functional is known to satisfy, such as the high-density uniform coordinate scaling, size consistency, and the Lieb-Oxford lower bound.\cite{lieb_lower_1979}

\subsection{In practice} Theoretically, the density $\dens$ is a function $\R^3 \rightarrow \R$. However, in practice, we work with a discretized version. Similar to Ref.~\citenum{dick_machine_2020,kasim_learning_2021,cuierrier_constructing_2021,kirkpatrick_pushing_2021,kanungo_learning_2024}, 
we choose to discretize the density features by evaluating them on a set of points $\{ r_i \in \R^3,\,i=1, \dots, G\}$ that are defined by a classical integration grid. An integration grid is a set of points $r_i \in \R^3$ (effectively, a point cloud in $\R^3)$ and associated weights $w_i \in \R$ used to numerically approximate spatial integrals, such as  the exchange-correlation energy and its potential. 
We refer to Ref.~\citenum{lebedev_quadrature_1999, becke_multicenter_1988, treutler_efficient_1995} for details on such integration grids. This representation of the density evaluated on a point cloud in $\R^3$ has the advantage of being independent of the basis set.

Following the discretization, $\Exc$ is therefore approximated as \begin{equation}\label{eq:discretized-exc-fun}
    \Exc[\dens] \approx - \frac{3}{4} \left(\frac{6}{\pi}\right)^{\frac{1}{3}}\sum_{i=1}^G \left(\dens^{(\uparrow)}(r_i)^{4/3} + \dens^{(\downarrow)}(r_i)^{4/3}\right)f_\params [\mathbf{x}[\dens]](r_i) w_i.
\end{equation}
In our setting, we pick the following set of features $\mathbf{x}[\dens]$: \begin{equation}\label{eq:semilocal-features}
    \mathbf{x}[\dens](r_i) =    
    \left[\dens^{(\uparrow)}(r_i), \dens^{(\downarrow)}(r_i), \norm*{\nabla \dens^{(\uparrow)}(r_i)}^2, \norm*{\nabla \dens^{(\downarrow)}(r_i)}^2, \tau^{(\uparrow)}(r_i), \tau^{(\downarrow)}(r_i), \norm*{\nabla \dens^{(\uparrow)}(r_i) + \nabla \dens^{(\downarrow)}(r_i)}^2\right],
\end{equation}
where $\tau$ denotes the Kohn-Sham kinetic energy density and the $\uparrow$ and $\downarrow$ denote the two spin channels. Such features are standard semi-local features used in meta-GGA functionals.\cite{perdew_jacobs_2001}
Effectively, the input $\mathbf{x}[\dens]$ of the neural network is then a tensor in $\R^{G\times 7}$, where $G$ typically depends on system size. 
These input features are called semi-local, because they only collect information at each given grid point. It is known that the exact functional cannot be captured with an enhancement factor that is just a function of semi-local features. It must also have a non-local dependence on the density. 

Analyzing the expression of the enhancement factor $f_\params[\mathbf{x}[\dens]]$, it is clear that there are two separate strategies to incorporate such non-local information on the density $\dens$: \begin{itemize}
    \item including extra hand-designed features in the set $\mathbf{x}[\dens]$ that capture in each grid point information from density features at other distant points, as it is done by adding exchange-like features\cite{kirkpatrick_pushing_2021} or convolved features;\cite{nagai_completing_2020,bystrom_cider_2022} this is in the spirit of climbing Jacob's ladder.
    \item keeping the set of features as in Eq.~\eqref{eq:semilocal-features} and allowing the model $f_\params$ to \textit{learn} longer range dependencies and mix information across different points. 
\end{itemize}
We take on the second approach, which is a step away from traditional DFT approaches built on Jacob's ladder or other hand-designed features and gears toward inferring non-locality through data and the model.
Among non-local effects, dispersion is very long-range and traditional functionals do not model it directly.
Instead, post-correction is typically applied. 
As we train on B3LYP densities, we train with its D3 correction\cite{grimme_consistent_2010,grimme_effect_2011} as part of the total energy, and focus on learning the other shorter-range non-local effects.

\subsection{Neural network architecture}\label{sec:neural-network-architecture}
In this section, we give an overview of the functional architecture used to parameterize the enhancement factor $f_\params$ in Eq.~\eqref{eq:exc-fun}. 
We process the integration grid on a per-atom basis: each atom $j$ has its own set of grid points $\{r_k^{(j)}\}$ with associated weights $\{w_k^{(j)}\}$, and the model produces an atom-specific energy density that is summed to obtain the total exchange-correlation energy.
We use two indexing conventions depending on context. When describing the model's internal operations (parts one through three below), which act identically and independently on each atom's grid, we use a single index $i$ to denote a grid point, leaving the atom index implicit (e.g., $h_i$, $x_i$). When the atom identity matters---for instance, in the non-local layer (\cref{sec:non-local}) where grid points communicate with their associated atom's coarse point---we write the atom index $j$ explicitly as a superscript, e.g., $\dens_i^{(j)}$ for the density at grid point $i$ of atom $j$.
The 7 semi-local features introduced in Eq.~\eqref{eq:semilocal-features} are first processed as follows:
\begin{align}
    x^{(\uparrow,\,\downarrow)}_i = 
    \log\left(
    \left[\dens^{(\uparrow)}_i, \dens^{(\downarrow)}_i, \norm*{\nabla \dens^{(\uparrow)}_i}^2, \norm*{\nabla \dens^{(\downarrow)}_i}^2, \tau^{(\uparrow)}_i, \tau^{(\downarrow)}_i, \norm*{\nabla \dens^{(\uparrow)}_i + \nabla \dens^{(\downarrow)}_i}^2\right]
    + \epsilon
    \right)
\end{align}
where the arrows refer to the respective spin channels, $(\uparrow,\,\downarrow)$ in $x_i$, emphasizes the ordering of the spin in the input vector and $\epsilon$ is a small constant to ensure numerical stability that in our setting is equal to $10^{-5}$. 

Our architecture mainly comprises three parts. 
\paragraph{Part one.} The first part is an input representation extractor
\begin{align}
    f_{\text{repr}}(x) = \sigma\left(W_2 \,\sigma\left(W_1 x + b_1\right) + b_2\right),
\end{align}
where $W_1$ and $b_1$ represent the first linear layer that projects the input vector $x$ onto a higher-dimensional space $\R^{D_\text{hid}}$, followed by a Swish activation function \cite{ramachandran_searching_2017}  $\sigma$ and another fully-connected layer with parameters $W_2$ and $b_2$. 
We evaluate the representation model twice on each grid point, changing the ordering of the spin channels in the input features, and then average the two, in order to obtain a feature vector that is invariant to the ordering of spin channels:
\begin{align}
    h_i = \frac{f_\text{repr}\left(x^{(\uparrow\,,\downarrow)}_i\right)+f_\text{repr}\left(x^{(\downarrow\,,\uparrow)}_i\right)}{2}.
\end{align}

\paragraph{Part two.} The second part processes these features through a stack of $L$ non-local layers. Each non-local layer updates the feature vector by incorporating non-local information via equivariant message passing through the atom's coarse point:
\begin{align}\label{eq:non-local-layer}
    h_i^{(l+1)} = f_\text{nonl}^{(l)}(h_i^{(l)}, \{R_j\}), \qquad l=0,\ldots,L-1,
\end{align}
where $h_i^{(0)} = h_i$, $\{R_j\}$ are the atom coordinates used as coarse points, and each $f_\text{nonl}^{(l)}$ has independent learnable parameters. The output of each non-local layer remains $D_\text{hid}$-dimensional, enabling the layers to be composed.
We elaborate on the structure of $f_\text{nonl}$ in detail in Sec.~\ref{sec:non-local}.

\paragraph{Part three.} The third component of the architecture is an output model that maps the feature vector $h_i^{(L)}$ from the final non-local layer through an MLP to produce the enhancement factor
\begin{align}
    h_{\text{enh}, i} = \sigma_{\text{out}}\left(W_4 \,\sigma\left(W_3 \, h_i^{(L)} + b_3\right) + b_4\right).
\end{align}
Here $W_3$ is $D_\text{hid}\times D_\text{hid}$ and $W_4$ is $1\times D_\text{hid}$, mapping to a scalar. 
The biases are all conformable to their corresponding weights. The last activation function of the output model is a scaled sigmoid function centered around one: $\sigma_\text{out}(x)= \frac{2}{1 + \exp(-x/2)}$.
This ensures that if the logit is close to zero, the final integration is close to the LDA exchange, as we multiply the enhancement factor with the LDA exchange energy density to obtain the final exchange-correlation energy
\begin{align}
    \Exc[\dens] =  - \frac{3}{4} \left(\frac{6}{\pi}\right)^{\frac{1}{3}} \sum_{i=1}^{G} h_{\text{enh},i} \left(\dens_i^{(\uparrow)\, 4/3} + \dens_i^{(\downarrow)\, 4/3}\right)w_i 
\end{align}

where $w_i$ are the merged integration grid weights and the sum runs over all grid points, as in Eq.~\eqref{eq:discretized-exc-fun}. The $f_\params [\mathbf{x}[\dens]](r_i)$ in Eq.~\eqref{eq:discretized-exc-fun} is then $h_{\text{enh},i}$ and the parameters $\params$ denote all the learnable parameters in the architecture described above. 

In all our experiments, we set $D_\text{hid}=256$, $D_\text{nonl}=16$, and $L=3$.

\subsection{Non-local interaction through coarse points}
\label{sec:non-local}
In this section, we expand on the structure of a single non-local layer $f_\text{nonl}^{(l)}$, as sketched in Eq.~\eqref{eq:non-local-layer}. The layer takes $h_i^{(l)}$ as input and produces $h_i^{(l+1)}$ as output; for brevity we drop the layer superscript on all intermediate quantities.
Due to the large number of integration grid points, it is computationally prohibitive to let all the grid points communicate with each other. In this section, we introduce a mechanism to pass on non-local information with the aid of some chosen helper nodes referred to as coarse points.  The information coarsening step is akin to accumulating multipole moments. 
We will discuss this relation formally in \cref{sub:supplmodeltheory}.

Since the integration grid is partitioned into atom-centered subgrids, each grid point $r_i$ is associated with a particular atom. The coarse points $\{R_j\}$ are placed at the atomic centers. In each non-local layer, each grid point communicates with the coarse point of its associated atom: downsampling aggregates information from an atom's grid points onto its coarse point, and upsampling sends the processed coarse features back to the same atom's grid points. Cross-atom information is carried by the density features themselves, which depend on the full molecular density. We use the term ``downsampling'' to refer to the operations that send messages from the integration grid points to the coarse points, which reduces the dimensionality.
We call it ``upsampling'' when sending messages back from the coarse points to the integration grid points, which increases the dimensionality. 

\paragraph{Pre-downsampling transform:} As the non-local interaction is the computational bottleneck, we first project the local features onto a lower-dimensional vector which is cheaper to manipulate. We define 
\begin{align}
    h_{\text{pre-down},i} = \sigma(W_\text{pre-down} h_i^{(l)} + b_\text{pre-down}),
\end{align}
where the weight $W_\text{pre-down}$ is $D_\text{nonl} \times D_\text{hid}$ (the hidden feature dimension; see \cref{sec:neural-network-architecture}).
We maintain the dimensionality $D_\text{nonl}=16$ throughout the non-local component. 

\paragraph{Downsampling:} Let $R_j$ be the coordinates of a coarse point. We define the coarsened feature of order $\ell$ in channel $c$ as a $2\ell+1$-dimensional vector
\begin{align}
    H_{j\ell c} = \sum_{k \in \mathcal{G}_j} \phi_c(\norm{r_{kj}}) Y_\ell(\widehat{r_{kj}}) \sum_{c'} W_{\text{down}, \ell cc'} h_{\text{pre-down},kc'}\, \tilde{w}_k
    \label{eq:downsampling}
\end{align}
where $r_{kj}=r_k - R_j$, $\widehat{r} = r/\norm{r}$, $Y_\ell$ are the spherical harmonics, $\phi_c$ is a radial basis function described in Eq.~\eqref{eq:radial-basis-function}, $\mathcal{G}_j$ denotes the set of grid points associated with atom $j$, and $\tilde{w}_k$ are the per-atom integration weights, as opposed to the merged weights $w_i$ used in the final energy integration. 
This expansion first projects the local scalar feature $h_{\text{pre-down},kc'}$ onto the product basis $\phi_c Y_\ell$ (with a learnable weight per tensor order $\ell$, which mixes the channels), and is followed by an aggregation step that sums up the messages sent by all integration grid points $k$ to the coarse point $j$. The resulting feature $H_{j\ell c}$ is a feature on the coarse point $j$ which transforms equivariantly according to the rotation of the input grid points and coarse points. 
The spherical tensor order $\ell$ ranges from $0$ to $\ell_\text{max}$.
We use $\ell_\text{max}=3$ in the paper.

\paragraph{Coarse-point mixing and symmetric contraction:}
After the downsampling step, we apply an equivariant linear transformation on the coarse features to mix channels within each tensor order $\ell$:
\begin{align}
    H'_{j\ell c} = \sum_{c'} W_{\text{mix}, \ell cc'} H_{j\ell c'},
\end{align}
where $W_{\text{mix},\ell}$ is a $D_\text{nonl} \times D_\text{nonl}$ learnable weight matrix for each $\ell$, initialized to the identity.
We then apply a symmetric contraction \cite{batatia_mace_2022} of correlation order $\nu$ to the coarse features, which generates higher-body-order interactions by taking $\nu$-fold symmetric tensor products of $H'_{j}$. Unlike the original formulation in MACE, which uses element-dependent weights to distinguish atom types in force field applications, our symmetric contraction uses shared weights across all coarse points, since the density functional depends only on the electron density and not on the nuclear charges.
We use $\nu=3$ in our experiments.

\paragraph{Upsampling:} Next, we send the message back from the coarse point to its associated grid points via
\begin{align}
    h'^{(j)}_{ic} = \sum_\ell \phi_c(\norm{r_{ij}}) \left\langle Y_\ell(\widehat{r_{ij}}) , \sum_{c'} W_{\text{up}, \ell cc'} H_{j\ell c'} \right\rangle
    \label{eq:upsampling}
\end{align}
where $r_{ij}=r_i^{(j)} - R_{j}$, $r_i^{(j)}$ is a grid point of atom $j$, and the upsampling uses the same radial basis functions $\phi_c$ as the downsampling step.

\paragraph{Post-upsampling transform and skip connection:}

After the message is sent back, we postprocess it and combine with the input features via a skip connection:
\begin{align}\label{eq:post-up}
    h_{\text{post-up},i} = \sigma(W_\text{post-up} h'_i + b_\text{post-up})
\end{align}
\begin{align}\label{eq:skip-connection}
    h_i^{(l+1)} = \sigma\left(W_\text{cat}\left[h_i^{(l)},\, \exp(-\dens_i)\, h_{\text{post-up},i}\right] + b_\text{cat}\right)
\end{align}
where $W_\text{post-up}$ is $D_\text{nonl}\times D_\text{nonl}$, $W_\text{cat}$ is $D_\text{hid}\times(D_\text{hid}+D_\text{nonl})$, and $\dens_i = \dens_i^{(\uparrow)} + \dens_i^{(\downarrow)}$ is the total density. The factor $\exp(-\dens_i)$ suppresses non-local effects on grid points with higher density values, such as regions near the nuclei, as an inductive bias. The concatenation and projection back to $D_\text{hid}$ enable the non-local layers to be composed iteratively.

\paragraph{Radial basis function:}
We use the following second-moment Gaussian radial basis function in the message passing involved in down- and up-sampling
\begin{align}\label{eq:radial-basis-function}
    \phi_c(r) = \frac{2}{\dim\cdot(2\pi s_c^2)^{\frac{\dim}{2}}}\frac{r^2}{2s_c^2} \exp\left(-\frac{r^2}{2s_c^2}\right)
\end{align}
where $\dim=3$ and $s_c$ are 16 different scale coefficients evenly spaced between \qty{0.3023}{\bohr} and \qty{2.192}{\bohr}, chosen such that two standard deviations of the Gaussians would reach the smallest and largest covalent radius estimates from Pyykko and Atsumi.\cite{pyykko_molecular_2009}
The squared term $r^2$ is chosen to suppress the influence of the near-core features, so that the learned coarsened features focus more on the bonding area.
The Gaussian decay naturally suppresses long-range contributions. Since each atom's grid points are processed independently, the model scales linearly in the number of atoms.

\subsection{Atomic partition of the energy density}\label{sec:atomic-partition}

In this section, we explain how processing the integration grid on a per-atom basis gives rise to a natural atomic partition of the exchange-correlation energy density.

In the downsampling step (Eq.~\eqref{eq:downsampling}), we sum over the grid points $\mathcal{G}_j$ of atom $j$ using the per-atom weights $\tilde{w}_k$:
\begin{align}
    H(R_j) \leftarrow \sum_{k \in \mathcal{G}_j} f(r_k^{(j)} - R_j)\, h(\mathbf{x}(r_k^{(j)}))\, \tilde{w}_k \approx \int f(r - R_j)\, h(\mathbf{x}(r))\, dr,
\end{align}
where we suppress the channel and spherical harmonic indices for clarity.

Similarly, the upsampling step (Eq.~\eqref{eq:upsampling}) sends the processed coarse features back to the grid points of the same atom:
\begin{align}
    h'(r_i^{(j)}; R_j) = f(r_i^{(j)} - R_j)\, H(R_j),
\end{align}
where each grid point communicates only with its associated atom's coarse point, and we write $R_j$ explicitly to emphasize that the upsampled features depend only on atom $j$ and not on any other atoms.

Since the upsampled features $h'$ depend only on a single atom $j$, the model effectively learns an atom-dependent energy density. These atom-specific contributions are combined in the final integration through the merged grid weights, which are given by
\begin{align}\label{eq:merged-weights}
    w_i^{(j)} = \tilde{w}_i^{(j)}\, \pi(r_i^{(j)}, R_j; \{R_{j'}\}),
\end{align}
where $\pi$ is a partition-of-unity function satisfying $\sum_j \pi(r, R_j; \{R_{j'}\}) = 1$ for all $r$. 
Substituting into the exchange-correlation energy (Eq.~\eqref{eq:discretized-exc-fun}), we obtain
\begin{align}
    \Exc[\dens] &= \sum_{j} \sum_{i} g(h'(r_i^{(j)}; R_j),\, h(\mathbf{x}(r_i^{(j)})))\, \tilde{w}_i^{(j)}\, \pi(r_i^{(j)}, R_j; \{R_{j'}\}) \nonumber\\
    &\approx \int \sum_{j} g(h'(r; R_j),\, h(\mathbf{x}(r)))\, \pi(r, R_j; \{R_{j'}\})\, dr,
\end{align}
where $g$ denotes the composition of the output model with the LDA prefactor.
Each atom contributes a learned energy density $g$, weighted by the partition function $\pi$.
In practice, we train with partition-scheme augmentation to ensure the model does not depend on a particular choice of $\pi$; this is detailed in \cref{subsupp:partition-augmentation}.

\subsection{Intuition on the structure of the non-local layer}\label{sub:supplmodeltheory}

This section is meant as a theoretical motivation for the structure of the non-local component of our functional. The main role of the non-local component is to capture interactions between features at different grid points, through helper nodes called coarse points.  
We provide an intuition about:
\begin{enumerate}
    \item how a single downsampling--upsampling pass through a coarse point can approximate any two-body interaction between features on the grid, independent of the coarse point chosen for computational convenience;
    \item how the symmetric contraction (\cref{sec:non-local}) extends this to multi-body interactions, and how the expressivity of our functional could be systematically increased to approximate arbitrary target functionals to any desired accuracy.
\end{enumerate}
We consider $(\{\phi_c,\, Y_{\ell}^m\})$ a basis set of $L^2(\R^3)$. We take two copies of it, indexed by $(c, \ell, m)$ and $(c', \ell', m')$ and consider a product basis of $L^2(\R^3 \times \R^3)$.
A given function $\kappa\in L^2(\R^3 \times \R^3)$ that is globally rotationally invariant, i.e. $\kappa(Q r_i, Q r_k) = \kappa(r_i, r_k)$ for any $Q\in \text{SO(3)}$, can be expressed (this is shown in Theorem \ref{thm:two-body-expansion} later in this section) as
\begin{align}\label{eq-two-body-kernel}
\kappa(r_i, r_k) %
    &= \sum_{c}\sum_{\ell} \phi_{c}(\norm{r_i}) \left\langle Y_{\ell}(\widehat{r_i}),\, 
   \sum_{c'} C_{c, \ell,  c'} \phi_{c'}(\norm{r_k})Y_{\ell}(\widehat{r_k}) \right\rangle,
\end{align}
where $C_{c, \ell, c'}$ is a set of coefficients that depend on the function $\kappa$ and the basis chosen and $Y_\ell$ is a vector that contains $Y_{\ell}^m$ for $-\ell \leq m \leq \ell$.
Using this expression, a convolved feature $h$ of the form 
\begin{equation*}
    h(r_i) := \int \tilde{h}(r_k) \kappa(r_i, r_k)\, dr_k
\end{equation*}
can be written as
 \begin{equation}\label{eq:two-body-rep}
    h(r_i) = \sum_{c}\sum_{\ell} \phi_{c}(\norm{r_i}) \left\langle Y_{\ell}(\widehat{r_i}),\, 
   \sum_{c'} C_{c, \ell,  c'} \int \tilde{h}(r_k)\phi_{c'}(\norm{r_k})Y_{\ell}(\widehat{r_k})  \,dr_k\right\rangle.
\end{equation}

We can now draw a parallel between this expression and the structure of the non-local layer, 
which shows that a single downsampling--upsampling pass can approximate two-body interactions between features on the grid. Combining the downsampling, upsampling and post-upsampling parts described in Eqs.~\eqref{eq:downsampling}, \eqref{eq:upsampling}, and~\eqref{eq:post-up}, we obtain the following expression with some rearrangement
\begin{align}
      h_{\text{post-up},i} &= \sigma(W_\text{post-up} h'_i + b_\text{post-up}) \\  h'^{(j)}_{ic} &= \sum_\ell \phi_c(\norm{r_{ij}}) \left\langle Y_\ell(\widehat{r_{ij}}) , \sum_{c'} W_{\text{up}, \ell cc'}  \sum_{k \in \mathcal{G}_{j}} \phi_{c'}(\norm{r_{kj}}) Y_\ell(\widehat{r_{kj}})\,\, \tilde{h}_{k\ell c'} \tilde{w}_k\right\rangle, \label{eq:our-layer}
\end{align}
where $j$ is the atom associated with grid point $i$ and we redefined part of the downsampled $H_{j \ell c'}$ in Eq.~\eqref{eq:downsampling} as $\tilde{h}_{k\ell c'}$ to make the dependence on $\phi_{c'}$ and $Y_\ell$ more explicit. The $\sum_{k \in \mathcal{G}_{j}}$ approximates the integral in Eq.~\eqref{eq:two-body-rep}.

Comparing the expressions of $h_{\text{post-up}, i}$ in Eq.~\eqref{eq:our-layer} and $h(r_i)$ in Eq.~\eqref{eq:two-body-rep}, we note a similar pattern. In our case, the coefficients $C_{c, \ell, c'}$ are part of the learnable parameters, the function $\tilde{h}$ also has a learnable component. Moreover, rather than expanding around the origin, we expand around the atom's center $R_{j}$. In theory, one can expand around any position, but the fidelity of the approximation with a finite basis may depend on the choice. 
As atomic density features are expected to be symmetric around the atomic centers, atomic centers are a natural choice for the expansion. 
The final aggregation on $c$ in Eq.~\eqref{eq:two-body-rep} in our case is generalized by the linear mixing with weights $W_{\text{post-up}}$.

The above analysis considers a single downsampling--upsampling pass, which captures two-body interactions. The symmetric contraction applied between downsampling and upsampling (see \cref{sec:non-local}) extends this to higher-body-order interactions: by taking symmetric tensor products of the coarsened features $H_j$ with itself up to correlation order $\nu$, the layer can represent $(\nu+1)$-body interactions.\cite{batatia_mace_2022,dusson_atomic_2021} In our experiments, we use $\nu=3$, enabling up to four-body interactions within each non-local layer.

We now provide the result that was taken for granted at the beginning of the discussion:
\begin{theorem}\label{thm:two-body-expansion}
Let $\kappa$ be a function in $L^2(\R^{3} \times \R^3)$ capturing the 2-body interaction between 3D coordinates $r_1$ and $r_2$. 
We assume $\kappa$ is globally rotationally invariant, that is for any $Q\in SO(3)$, we have
\begin{align}
    \kappa(Qr_1, Qr_2) = \kappa(r_1, r_2).
\end{align}
Assume $\{\phi_c, Y_{\ell}^m\}$ is a basis set of $L^2(\R^3)$ and we consider 2 copies of it, indexed by $(c_1, \ell_1, m_1)$ and $(c_2, \ell_2, m_2)$ to form a product basis of $L^2(\R^{3} \times \R^3)$.
Then we can expand $\kappa$ around the origin as follows:
\begin{align}
     \kappa(r_1,  r_2) %
    &= \sum_{c_1, \ell} \phi_{c_1}(\norm{r_1}) \left\langle Y_{\ell}(\widehat{r_1}),\, 
   \sum_{c_2} C_{c_1, \ell,  c_2} \phi_{c_2}(\norm{r_2})Y_{\ell}(\widehat{r_2}) \right\rangle
\end{align}
where  $C_{c_1, \ell,  c_2}$ is a set of coefficients that depend on $\kappa$ and the basis set chosen and $Y_\ell$ is the vector form of  $Y_\ell^m$ for $-\ell \leq m \leq \ell$.
\end{theorem}

\begin{proof}
By the Hilbert space structure of $L^2(\R^{3}\times \R^3)$, we can decompose $\kappa$ as
\begin{align*}
    \kappa(r_1, r_2)
    = \sum_{\substack{c_1, \ell_1, m_1 \\ c_2, \ell_2, m_2}}  C_{c_1,\ell_1,m_1, c_2,\ell_2, m_2} \prod_{i=1}^2 \phi_{c_i}(\norm{r_i}) Y_{\ell_i}^{m_i}(\widehat{r_i}),
\end{align*}
for a certain set of coefficients $C_{c_1,\ell_1,m_1, c_2,\ell_2, m_2}$, where $-\ell_1 \leq m_1 \leq \ell_1$ and $-\ell_2 \leq m_2 \leq \ell_2$. 
Using the rotational invariance property, we have
\begin{align*}
    \kappa(r_1, r_2) = \int \kappa(r_1, r_2) \,dQ
     =  \sum_{\substack{c_1, \ell_1, m_1 \\ c_2, \ell_2, m_2}} 
    C_{c_1,\ell_1,m_1,c_2,\ell_2,m_2} 
    \prod_{i=1}^2 \phi_{c_i}(\norm{r_i}) 
    \int \prod_{j=1}^2 Y_{\ell_j}^{m_j}(Q\widehat{r_j}) \,dQ
\end{align*}
where $dQ$ denotes the uniform measure over $SO(3)$. 
Using the equivariance property of spherical harmonics (Lemma 11 in \citet{dusson_atomic_2021}), we obtain
\begin{align*}
    \kappa(r_1, r_2) 
    &= \sum_{\substack{c_1, \ell_1, m_1 \\ c_2, \ell_2, m_2}}  C_{c_1,\ell_1,m_1, c_2,\ell_2, m_2} \prod_{i=1}^2 \phi_{c_i}(\norm{r_i}) \int \prod_{j=1}^2 \sum_{k_j=-\ell_j}^{\ell_j}Y_{\ell_j}^{k_j}(\widehat{r_j}) D_{k_j m_j}^{\ell_j}(Q) \,dQ \\
    &= \sum_{\substack{c_1, \ell_1, m_1 \\ c_2, \ell_2, m_2}} C_{c_1,\ell_1,m_1, c_2,\ell_2, m_2} \prod_{i=1}^2 \phi_{c_i}(\norm{r_i}) 
    \sum_{k_1, k_2}
    \prod_{j=1}^2 Y_{\ell_j}^{k_j}(\widehat{r_j})
    \int D_{k_2 m_2}^{\ell_2}(Q)D_{k_1 m_1}^{\ell_1}(Q) \,dQ,
\end{align*}
where $D_{k_i, m_i}^{\ell_i}$ are the Wigner D-matrices.
Using Eq.~(2) in Lemma 4.11.1 in \citet{varshalovich_quantum_1988}, we evaluate the integral of the product of Wigner matrices to be $\delta_{\ell_1, \ell_2} \delta_{m_1, -m_2}\delta_{k_1, -k_2}$, up to some constants that will be absorbed into the coefficients $C$ without renaming. 
Hence, we can reduce the indices to just $c_1, c_2, \ell, m$, and $k$ and obtain
\begin{equation*}
 \kappa(r_1, r_2)  = \sum_{c_1, c_2, \ell, m} C_{c_1,c_2, \ell, m} \prod_{i=1}^2 \phi_{c_i}(\norm{r_i}) 
    \sum_{k}
     Y_{\ell}^{k}(\widehat{r_2})  Y_{\ell}^{-k}(\widehat{r_1}).
\end{equation*}
Note that the summation over $m$ is independent of the basis functions $\phi_{c_i}$ and $Y^{k}_{\ell}$.
This means we can further simplify the expansion by absorbing the summation over $m$ into a newly defined $C_{c_1, c_2, \ell}$
\begin{align*}
   \kappa(r_1, r_2)  = \sum_{c_1, c_2, \ell} C_{c_1,c_2, \ell} \prod_{i=1}^2 \phi_{c_i}(\norm{r_i}) \sum_k
     Y_{\ell}^{k}(\widehat{r_2})  Y_{\ell}^{-k}(\widehat{r_1}).
\end{align*}
which we rearrange to be
\begin{align*}
   \kappa(r_1, r_2)  
    &= \sum_{c_1, c_2, \ell}C_{c_1, \ell,  c_2} \phi_{c_1}(\norm{r_1}) \left\langle Y_{\ell}(\widehat{r_1}),\, 
    \phi_{c_2}(\norm{r_2})Y_{\ell}(\widehat{r_2}) \right\rangle \\
    &= \sum_{c_1,  \ell} \phi_{c_1}(\norm{r_1}) \left\langle Y_{\ell}(\widehat{r_1}),\, 
   \sum_{c_2} C_{c_1, \ell,  c_2} \phi_{c_2}(\norm{r_2})Y_{\ell}(\widehat{r_2}) \right\rangle, 
\end{align*}
where $Y_\ell$ is the vector containing $Y^k_\ell$ for $-\ell \leq k \leq \ell$, obtaining the desired expression.
\end{proof}

\paragraph{Extension to $N$-body interaction.}  Theoretically, one could generalize the theorem above to $N$-body interactions and adapt the non-local layer accordingly.\cite{dusson_atomic_2021, batatia_mace_2022}

\subsection{Related work}

The theoretical foundation developed in this section is firmly based on the framework of the \emph{atomic cluster expansion} (ACE),\cite{drautz_atomic_,dusson_atomic_2021} which enables the systematic construction of a complete descriptor of the atomic environment. Specifically, the expressivity of message passing can be enhanced by increasing the basis set size parameters $\ell$ and $c$, as well as the correlation order of the interaction ($\nu=3$ in our case).\cite{batatia_mace_2022}
In the ACE framework, the atomic center directly receives the message. In contrast, our approach accumulates spherical tensor features at coarse points, serving both as a computational aid and as a means to enrich modeling capacity. These aggregated features are then transmitted back to the integration grid points. This mechanism is reminiscent of numerical techniques such as the fast multipole method (FMM),\cite{greengard_fast_1987} where, for example, the Laurent expansion $\phi_c$ can be used to approximate the long-range two-body Coulomb kernel \( f(r_1, r_2) = 1/\|r_1 - r_2\| \).
A similar strategy is employed in \citet{gao_learning_2024}, where nuclei-centric descriptors of the total density are constructed. Their work further processes the coarsened features using message-passing layers within a neural force field framework,\cite{batzner_e3equivariant_2022} enabling direct graph-level readout. However, in our preliminary experiments, we observed that this approach can lead to significant overfitting when evaluated on the Diet GMTKN55 benchmark.

Other studies have also leveraged coarsened features in various contexts, such as modeling protein binding sites\cite{sestak_vnegnn_2024} and compressing input signals in fluid dynamics simulations.\cite{alkin_universal_2024}

From a theoretical standpoint, learning the enhancement factor corresponds to learning a mapping between functional spaces. In recent machine-learning literature, neural operators have emerged as powerful tools for this purpose.\cite{kovachki_neural_2023} These models learn operators that map input functions to output functions and can be evaluated with arbitrary discretization of the domain.
In particular, our non-local layer shares similarities with low-rank kernel neural operators,\cite{kovachki_neural_2023} which retain rich functional expressivity while reducing computational cost.
We tailored the non-local layer design to better suit the structure of the atomic grids, as the irregularity limits the applicability of most standard techniques from the neural operator literature.

\section{Training details}
\label{supp:training-details}
\subsection{Training objective}
\label{sec:training-objective}
The objective of training is to minimize a regression loss on reaction energies, which linearly combines molecular total energies. 
The total energy of a molecule $M$ is computed as
\begin{align}
    E_\text{tot}^\theta[M] = E_{\text{tot}-\text{xc}}[M] + \Exc^\theta[\dens]
\end{align}
where $E_{\text{tot}-\text{xc}}$ is the total energy minus the exchange-correlation component, calculated using the B3LYP functional.\cite{becke_densityfunctional_1993,stephens_initio_1994}
This is precomputed. 
Notably, this quantity contains the D3 dispersion energy. 
We calculate the reaction energy as a weighted sum of the stoichiometric coefficients $\{c_M\}$ for a reaction
\begin{align}
    \Delta E = \sum_{M} c_M E_\text{tot}[M].
\end{align}

The final loss is then defined as a weighted mean squared error
\begin{align}
    \E\left[
    \frac{\left|\Delta E-\Delta E^\text{ref}\right|^2}{\qty{e-4}{\hartree} + |\Delta E^\text{ref}|}
    \right]
    \label{eq:weighted_MSE_loss}
\end{align}
where $\Delta E^\text{ref}$ are the reference reaction energies and the expectation is taken over reactions sampled via the hierarchical procedure described in Sec.~\ref{subsupp:dataset-sampling}.

\subsection{Dataset sampling and model selection}
\label{subsupp:dataset-sampling}
The training data comprise multiple datasets spanning diverse chemical properties (see Sec.~\ref{sec:training-data}).
Because the datasets vary widely in size---from tens of reactions (e.g., transition-metal benchmarks) to over a hundred thousand (e.g., MSR-ACC/TAE)---uniform sampling would cause training to be dominated by the largest datasets and underrepresent smaller but chemically important ones.

To address this, we adopt a two-level sampling procedure.
At each training step, we first sample a dataset $i$ with probability $p_i$, and then sample a reaction uniformly from within that dataset.

\paragraph{Initial probabilities}
The initial sampling probabilities are set as $p_i \propto \alpha_{c(i)} \cdot |D_i|$, where $|D_i|$ is the number of reactions in dataset $i$, $c(i)$ denotes the category of dataset $i$, and $\alpha_c$ is a per-category weight.
We group the datasets into nine categories---total atomization energies (including non-GDB9-derived MSR-ACC/TAE25 structures and w4-cc), GDB9 structures (including GDB9-W1-F12 and our derived version), barrier heights (MSR-ACC/Reactions and BH9), thermochemistry (MSR-ACC/Conf, IP, PA, EA and MB2061), non-covalent interactions, non-covalent interaction potential energy surfaces (MSR-ACC/Water and DES370K), distortion data, main-group atomic sets, and transition metals---and compute the category weights $\{\alpha_c\}$ by solving for the values that yield prescribed target sampling proportions across categories.

\paragraph{Relative excess loss.}
To monitor training progress across all datasets,
we introduce the \emph{relative excess loss} (REL).
Most training datasets are split into a training portion and a held-out validation portion; datasets that are too small to split are used in their entirety for both training and validation.
Every 25{,}000 training steps, we evaluate the current model on each holdout set and compute the mean absolute error (MAE) $\ell_i$ in kcal/mol.
The per-dataset REL is defined as
\begin{align}
    \mathcal{R}_i = \frac{\max(\ell_i - t_i,\; 0)}{n_i},
    \label{eq:rel}
\end{align}
where $t_i$ and $n_i$ are precomputed constants specific to each holdout set $i$.
$t_i$ is the minimum MAE among a set of baseline density functionals evaluated on the same holdout set, and $n_i$ is the difference between the median and minimum baseline MAEs.
Intuitively, $\mathcal{R}_i = 0$ means \clippy{} is at least as accurate as the best baseline on dataset $i$, and $\mathcal{R}_i = 1$ means it performs at the median baseline level.
The total REL, $\mathcal{R} = \sum_i \mathcal{R}_i$, counts (in a soft sense) the number of datasets on which the model has not yet reached the best-baseline accuracy, providing a diagnostic of whether the model has sufficiently fitted the training data across all domains.

\paragraph{Adaptive sampling.}
After each validation, we update the dataset sampling probabilities using a multiplicative exponential update rule:
\begin{align}
    p_i \leftarrow p_i \cdot \exp\!\left(\eta\, \mathcal{R}_i\right),
    \label{eq:sampling-update}
\end{align}
where $\eta > 0$ is a step size, followed by renormalization $p_i \leftarrow p_i / \sum_j p_j$.
This steers training toward the model's weakest categories: datasets where $\mathcal{R}_i > 0$ receive increased sampling probability, while datasets where $\mathcal{R}_i = 0$, meaning the model already surpasses the best baseline, are automatically downweighted, redirecting computational resources to categories that still need improvement.
This mechanism also provides a natural way of handling uncertainty in the reference energies.
Since the training labels are computed with high-level wavefunction methods that carry their own finite accuracy, continuing to minimize the loss on a dataset where the model already outperforms the best baseline functional risks overfitting to noise in the reference data.
By clamping the per-dataset loss contribution at the best-baseline level, the REL-based update effectively treats the best-baseline accuracy as a soft target, avoiding diminishing returns on well-learned categories and focusing capacity where genuine improvements are still possible.

In practice we set $\eta=0.025$ and update the sampling weights every 25,000 steps.

\subsection{Parameter initialization}
The weights of the local parts of the network are initialized according to the Xavier initialization scheme\cite{glorot_understanding_} using uniform distribution and the biases are set to 0. 
The parameters of the non-local layer are initialized with a modified Xavier scheme which also takes into account that tensor features of the same order share weights.\cite{geiger_e3nn_2022}

\subsection{Optimization}
\label{subsupp:optimization}
We optimize model parameters using a combined Muon-Adam strategy.
Model parameters are split into two groups based on their dimensionality: all hidden weight matrices (parameters with $\geq 2$ dimensions, excluding the final linear layer) are optimized with the Muon optimizer, while low-dimensional parameters (biases) and the final output layer are optimized with Adam.

Muon runs standard SGD with momentum and then replaces each parameter update with the nearest orthogonal matrix.\cite{_muon_,liu_muon_2025}
Adam parameters use betas $(0.9, 0.95)$ and $\epsilon = 10^{-10}$.

Both parameter groups follow independent cosine learning rate schedules with linear warmup, defined as
\begin{align}
\text{lr}(t) = \begin{cases}
    \text{lr}_{\max} \times \dfrac{t}{t_{\text{warmup}}},& \text{if } t \leq t_{\text{warmup}},\\[6pt]
    \text{lr}_{\min} + \dfrac{\text{lr}_{\max} - \text{lr}_{\min}}{2} \left(1 + \cos\left(\dfrac{t - t_{\text{warmup}}}{T - t_{\text{warmup}}}\pi\right)\right), & \text{otherwise},
\end{cases}
\end{align}
where $T$ is the total number of training steps.
The Muon and Adam groups use separate peak learning rates.
Optimizer hyperparameters are summarized in Table~\ref{tab:hyperparameters}.

\begin{table}[t!]
    \centering
    \caption{Optimizer hyperparameters.}
    \tablefontsize
\begin{tabular}{ll}
\toprule
Hyperparameter & Value \\
\midrule
Maximum learning rate (Muon)  &  0.0007 \\
Maximum learning rate (Adam)  &  0.00015 \\
Muon momentum                 &  0.95 \\
Adam betas                    &  (0.9, 0.95) \\
Warmup time                   &  50,000 \\
Num.\ steps                   &  1,000,000 \\
Gradient clipping threshold   &  0.0001 \\
EMA decay (pre-training)      &  0.9999 \\
EMA decay (fine-tuning)       &  0.999 \\
\bottomrule
\end{tabular}

    \label{tab:hyperparameters}
\end{table}

All ablation models were run on 8 A100 GPUs, amounting to a total minibatch size of 8, in a distributed data parallel manner (one reaction per GPU).
All models were trained for 1,000,000 steps.

Furthermore, we perform exponential moving average (EMA) of model weights with a decay factor of 0.9999 to improve convergence.

\subsection{Self-consistent fine-tuning}
\label{subsupp:SCF-finetuning}
Empirically, we find the functionals trained on B3LYP features already yield decent performance when evaluated self-consistently, but sometimes there can be a gap in their predictive power compared to evaluating on B3LYP densities. 
To reduce this gap, we introduce a self-consistent fine-tuning strategy to transfer the predictive performance.

Concretely, recall that to evaluate the functional we solve the following minimization problem
\begin{align}
    E_\text{tot} = \min_C E_\text{tot}^\theta[C]
\end{align}
with $C$ being a $B\times N$ matrix ($B$ and $N$ stand for the number of basis functions and number of electrons) with orthonormal columns, representing the occupied orbitals.
By the envelope theorem (the first-order stationarity condition), we have that 
\begin{align}
    \nabla_\theta E_\text{tot} = \nabla_\theta E_\text{tot}^\theta[C^*]
\end{align}
where $C^* = \arg\min_C E_\text{tot}^\theta[C]$.
That is, to make proper model updates, we simply need to train on the ground state density matrices found self-consistently or via any minimization procedure. 
Therefore, at the fine-tuning stage, we take the model weights trained previously and fine-tune on self-consistent densities produced by the model itself.
We use the same weighted loss as in Eq.~\eqref{eq:weighted_MSE_loss}. 
We reinitialize the optimizer state, use a constant global learning rate of $10^{-5}$ for both Muon and Adam, and apply per-parameter RMS gradient clipping with threshold $10^{-4}$.
We apply EMA with a decay factor of 0.999.
We fine-tune with a minibatch size of 1 on 1 A100 GPU for 20,000 steps monitoring the density error as explained in Sec.~\ref{sec:densities} of the main paper. The adaptive sampling probabilities $p_i$ defined in Sec.~\ref{subsupp:dataset-sampling} are fixed at this stage and set equal to the sampling probabilities found at the $10^6$ step in training stage.

To produce self-consistent densities, we resort to a simpler SCF procedure:
we use the PySCF loop along with DIIS (Direct Inversion in the Iterative Subspace)\cite{pulay_convergence_1980,pulay_improved_1982} on the latest 8 updates. 
We use the Treutler-Ahlrichs radial grids and grid level 3 for integration.
Densities are represented in the def2-QZVP basis set.\cite{weigend_gaussian_2003} 
We set the density threshold to be $10^{-8}$, and discard grid points with initial density values lower than the threshold. 
Density fitting is enabled.
The procedure starts from the B3LYP initial guess and runs until energy change is lower than \qty{5e-6}{\hartree} and orbital gradient norm is smaller than \qty{1e-3}{\hartree}.
The algorithm is terminated after 40 steps if not converged. If not converged, the lowest energy found during the SCF cycles is used.

\section{Training data: computational details}
\label{sec:training-data-details}
This section specifies in detail the computational protocols that were used to produce our training data.

\subsection{MSR-ACC/TAE dataset}
In this section, we provide a summary of the MSR-ACC/TAE set.
The protocol for generating it is specified in detail in a separate publication.\cite{ehlert_accurate_2025}
The molecular structures were obtained by enumerating all plausible covalent graphs of closed-shell charge-neutral molecules with up to 4 non-hydrogen atoms up to argon excluding the rare-gas elements and by sampling from graphs with up to 5 such atoms, then determining the 3D structure by a cascade of geometry relaxation steps with increasingly accurate methods, B3LYP-D3(BJ)/def2-TZVPP being the final one.
We only keep the lowest-energy conformer for each graph.
Molecular graphs (with undetermined bond order) are obtained from the bond model of GFN-FF.\cite{spicher_robust_2020}
Molecules for which B3LYP/def2-TZVPP predicts the triplet to be lower in energy than the singlet are excluded.
The molecules are labeled with total atomization energies obtained by the W1-F12 protocol,\cite{karton_explicitly_2012} which reproduces CCSD(T)/CBS at benchmark accuracy.
Molecules for which the relative contribution of the perturbative triple excitations to the atomization energy (\%TAE[(T)]) is larger than 6\% (a sign of significant multireferential character) are excluded.
All resulting stable equilibrium structures consisting of a single covalently bound fragment (90.7\%) can be found in the released MSR-ACC/TAE25 dataset.\cite{ehlert_accurate_2025}
This released dataset is combined with the remainder of structures consisting of two or more covalent fragments to form the MSR-ACC/TAE training dataset.

In addition to the structures generated according to the graph-based protocol, MSR-ACC/TAE was extended to larger systems by using structures from GDB9,\cite{karton_highly_2025} changing the element composition by replacing second row elements with third row elements and H with \{Li, Na, F, Cl\}, followed by geometry relaxation with GFN2-xTB and B3LYP-D3(BJ)/def2-TZVPP.
The resulting $\sim$41k structures were used to form total atomization reactions and were labeled with the W1w protocol.\cite{karton_w4_2006,karton_explicitly_2012}

\subsection{MSR-ACC/IP, /EA, /PA, /Conf, /NCI, /Water, /Distortion and /Reactions datasets}

The reactions in all remaining datasets in MSR-ACC are labeled with a slightly modified W1w protocol,\cite{karton_w4_2006,karton_explicitly_2012} an earlier version of W1-F12 with a similar accuracy and somewhat higher computational cost, using the MRCC software package.\cite{mester_overview_2025}
In W1w, the Hartree-Fock (HF) component of the total energy is extrapolated to the complete basis-set limit (CBS) from the jul-cc-pV(T+d)Z and jul-cc-pV(Q+d)Z basis sets \cite{peterson_systematically_2008} using the $E(L) = E_\infty + A/L^\alpha$ two-point extrapolation formula, with $\alpha= 5$.
The valence CCSD correlation energy is extrapolated from the same basis sets with an extrapolation exponent of $\alpha=3.22$.
The valence perturbative triple excitations are extrapolated from the jul-cc-pV(D+d)Z and jul-cc-pV(T+d)Z basis with an extrapolation exponent of $\alpha = 3.22$.
The CCSD(T) inner-shell contribution is calculated with the cc-pwCVTZ basis set (modification of the original W1w which used a custom basis set).
Like in MSR-ACC/TAE, we exclude reactions where any structure has \%TAE[(T)] exceeding 6\%.

The reactions in MSR-ACC/IP, MSR-ACC/EA, and /PA were obtained by removing an electron, adding an electron, and adding a uniformly sampled proton, respectively, from uniformly sampled molecules from MSR-ACC/TAE.
The reactions in MSR-ACC/Conf are conformational changes up to 10 kcal/mol in molecules from MSR-ACC/TAE generated with the CREST program,\cite{pracht_automated_2020} with structures relaxed at the B3LYP/def2-TZVPP level.
The reactions in MSR-ACC/NCI were obtained by adding up to six fragments with one non-hydrogen atom and optimizing the resulting cluster; clusters with topology changes, e.\,g. due to chemical reactions, were discarded.
The structures in MSR-ACC/Water represent the water dimer potential energy surface; they are taken from \citet{smith_revised_2016} and relabeled.
The reactions in MSR-ACC/Distortion were obtained by distorting equilibrium structures from the MSR-ACC/TAE subset based on their vibrational normal modes, the energy difference between the equilibrium structure and the distorted one was used as reaction for describing the vibrational modes.
The reactions in MSR-ACC/Reactions are barrier heights of elementary reactions of organic molecules with up to eight atoms.
The initial structures for this dataset were sampled from the public datasets Transition1x,\cite{schreiner_transition1x_2022} RFD1-CNHO,\cite{zhao_comprehensive_2023} and 3DReact,\cite{vangerwen_3dreact_2024} and were then recombined through an in-house reaction exploration engine similar in function to Chemoton\cite{unsleber_chemoton_2022} or AutodE.\cite{young_autode_2021}
The resulting reaction paths were refined by intrinsic reaction coordinate (IRC) optimization with TPSSh\cite{staroverov_comparative_2003}/def2-SVP.

\subsection{Atomic datasets}

Total energies, ionization potentials up to triple ionization, and electron affinities of atoms up to argon excluding Li and Be were calculated at CCSD(T)/CBS by extrapolating the HF component with the two-point formula of \citet{karton_comment_2006} and the CCSD(T) correlation component with the cubic power law formula from the aug-cc-pCVQZ and aug-cc-pCV5Z basis sets.

\subsection{Density features on the grid}\label{subsec:density-features-grid}
To regularize the trained model with respect to numerical variations on the grid, we employ a data augmentation strategy. First, the density is computed on four distinct integration grids, using level 1 as defined in PySCF\cite{sun_pyscf_2018} with four different radial integration schemes: Treutler, Mura–Knowles, Gauss-Chebyshev and Delley.

\subsection{Augmentation of the space-partitioning schemes during training}\label{subsupp:partition-augmentation}
As explained in \cref{sec:atomic-partition}, we decompose the XC energy into atomic contributions using a space-partitioning scheme.
The molecular $\Exc$ should in principle be invariant the choice of the space-partitioning scheme; therefore, to prevent the model from overfitting to any single scheme we randomize it during training: at each training step, a
partitioning scheme is sampled uniformly at random from three options: (i) the standard Becke partition with Treutler-Ahlrichs radii
adjustment;\cite{becke_multicenter_1988,treutler_efficient_1995} (ii) the
Stratmann-Scuseria-Frisch (SSF) scheme;\cite{stratmann_achieving_1996} and (iii)~the Laqua-Kussmann-Ochsenfeld (LKO) scheme,\cite{laqua_improved_2018} which applies the Becke polynomial but caps
interatomic distances at a cutoff radius, enabling linear-scaling cost for large molecules. The sampled scheme is applied consistently to all molecules in a reaction. This augmentation exposes the model to different spatial weight distributions during training, reducing overfitting to artifacts of a single
partitioning scheme.

We use the same basis set and integration grid for molecules in the same reaction, to account for error cancellation.
During training we evaluate our functional at fixed densities using B3LYP\cite{becke_densityfunctional_1993,stephens_initio_1994} in the def2-QZVP basis set\cite{weigend_gaussian_2003} or the ma-def2-QZVP basis set\cite{zheng_minimally_2011} for all reactions containing anions. Based on the fixed density we compute the relative energy of a reaction from the B3LYP total energy without the XC energy.

\section{Evaluation protocols}
\label{sec:suppeval}

In this section, we describe the evaluation settings for the four benchmark categories conducted in this paper: (1) reaction energies, (2) dipole moments, (3) geometry optimization, and (4) computational cost.

\subsection{Evaluation sets}
\label{sec:eval-sets}
\paragraph{Reaction energies} 
Our primary evaluation set is the GMTKN55\cite{goerigk_look_2017} datasets, the de facto standard benchmark for electronic structure methods covering five categories: basic properties, thermochemistry, kinetics, intermolecular non-covalent interactions, and conformational energies (or intramolecular non-covalent interactions). Additionally, we test \clippy on the atomization energy test set W4-17,\cite{karton_w417_2017} on the Wiggle150 dataset of strained conformations of adenosine, benzylpenicillin, and efavirenz,\cite{brew_wiggle150_2025} on the W1-SN2-BH dataset of reaction barriers for S$_{\rm N}$2 reactions,\cite{karton_w1sn2bh_2026} and on selected datasets containing transition metals: MOR41,\cite{dohm_comprehensive_2018}
ROST61,\cite{maurer_assessing_2021}
MOBH35,\cite{semidalas_mobh35_2022} 3dTMV,\cite{neugebauer_benchmarkquality_2023}
3d4dIPSS,\cite{liang_goldstandard_2025}
DAPD,\cite{chan_dapd_2023}
TMD10.\cite{liang_goldstandard_2025}
For the ablation studies on the model architecture and dataset composition in Sec.~\ref{sec:local-ablation} and Sec.~\ref{sec:data-ablation}, we instead evaluate on the representative subset Diet GMTKN55,\cite{gould_diet_2018} containing 100 reactions, to reduce the computational cost of these ablations.
Additionally, we evaluate on the MSR-ACC/TAE25 holdout set,\cite{ehlert_accurate_2025} GDB9(W1-F12),\cite{karton_highly_2025} Mindless-24,\cite{gasevic_chemical_2025} DES370K,\cite{donchev_quantum_2021} Water2510,\cite{smith_revised_2016} IHB100x10\cite{rezac_noncovalent_2020}, D442x10,\cite{rezac_noncovalent_2022} SH250x10,\cite{kriz_noncovalent_2022} and R739x5\cite{kriz_noncovalent_2021} datasets.

\paragraph{Dipole moments}
To confirm the quality of \clippy SCF electron densities we compare its dipole moments against a well-established benchmark set.\cite{hait_how_2018} We use the recommended computational protocol\cite{hait_how_2018} with the aug-pc-4 basis. We exclude $\text{CH}_2$ because it was the only molecule for which we could not reproduce the reported dipole errors of the baseline functionals.

\paragraph{Geometry optimization benchmark sets}
Furthermore, we evaluate our functional on (semi-)experimental geometry datasets, including light main group bond lengths (LMGB35),\cite{grimme_consistent_2015} 
heavy main group bond lengths (HMGB11),\cite{grimme_consistent_2015} 
and the bond lengths and bond angles of 21 small molecules (CCse21).\cite{piccardo_semiexperimental_2015} We also evaluate our functional on the W4-11-geom dataset.\cite{karton_w411_2011,spackman_basis_2016}

\paragraph{Computational cost benchmark set}
The computational cost results were produced using a dataset of 38 molecules shown in Fig.~\ref{fig:si-timing-dataset}. The systems are collected from Grimme,\cite{grimme_exploration_2019} S30L,\cite{sure_comprehensive_2015} HS13L,\cite{gorges_reliable_2022} and NCI16L.\cite{gorges_efficient_2023}

\begin{figure}[t] %
    \centering
    \includegraphics[width=0.75\linewidth]{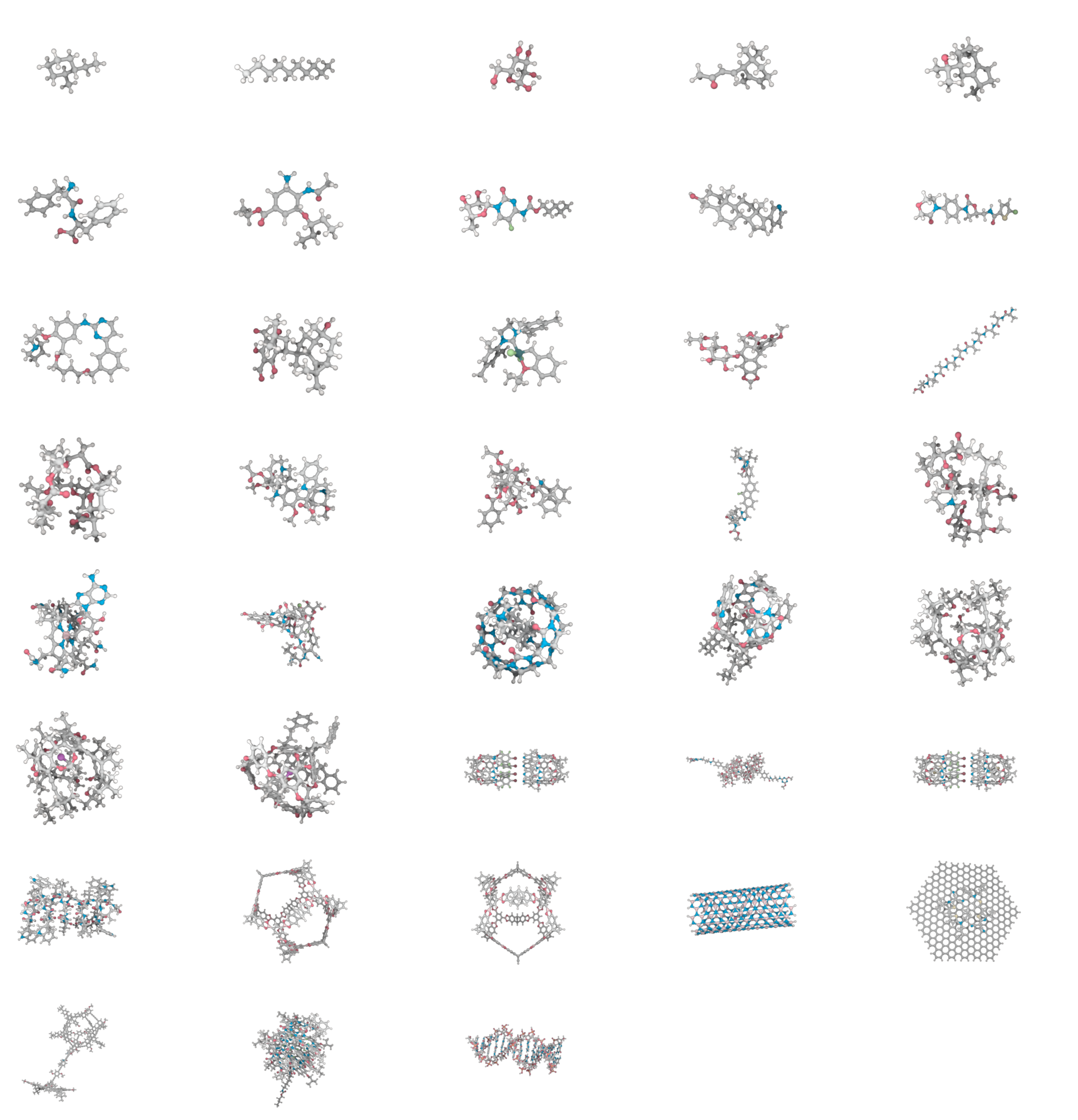}
    \caption{Systems used for evaluating the computational cost of the \clippy functional}
    \label{fig:si-timing-dataset}
\end{figure}

\subsection{Baselines and dispersion correction settings}
We compare our functional with several traditional baseline methods, including (meta-)GGA functionals, revPBE,\cite{zhang_comment_1998} B97M-V,\cite{mardirossian_mapping_2015} and r\textsuperscript{2}SCAN,\cite{furness_accurate_2020,ehlert_r2scand4_2021} and (range-separated) hybrid functionals, B3LYP,\cite{becke_densityfunctional_1993,stephens_initio_1994} M06-2X,\cite{zhao_m06_2008} \(\omega\)B97X-V,\cite{mardirossian_ob97xv_2014} and \(\omega\)B97M-V.\cite{mardirossian_ob97mv_2016} revPBE, r\textsuperscript{2}SCAN, B3LYP were evaluated with a D3(BJ) correction, while M06-2X was evaluated with a D3(0) correction. \clippy was evaluated using a D3(BJ) correction with B3LYP parameters.
B97M-V, \(\omega\)B97X-V, and  \(\omega\)B97M-V have a VV10\cite{vydrov_nonlocal_2010} correction, as indicated with ``-V''. 
D3 dispersion corrections\cite{grimme_consistent_2010,grimme_effect_2011} are included using the simple-dftd3 package.\cite{ehlert_simple_2024}

\subsection{Basis sets}
For reaction energy benchmarks, we use the def2-QZVP\cite{weigend_gaussian_2003, weigend_balanced_2005} basis set or the ma-def2-QZVP\cite{zheng_minimally_2011} basis set for reactions containing anions.
For the dipole moment benchmark, we use the aug-pc-4 basis.\cite{ambroise_probing_2019}
For geometry optimization, we use the smaller def2-TZVP basis set.\cite{weigend_balanced_2005}
We use density fitting for all evaluations in combination with the def2-universal-jkfit auxiliary basis set.\cite{weigend_hartree_2008}

\subsection{SCF retry protocol}
\label{sec:scf-retry}
As expected for an ML functional, Skala's SCF convergence is slightly more challenging than for traditional functionals, at least when using PySCF.\cite{sun_pyscf_2018} For this reason, we find it important to give all the specifics of our SCF settings and report statistics for \clippy compared to baseline functionals. 

We evaluate our functional using PySCF\cite{sun_pyscf_2018,sun_recent_2020,li_introducing_2025} with appropriate Coulomb fitting basis set.\cite{weigend_hartree_2008}
In the evaluation process, we start SCF computations from the Minimal Atomic Orbital basis (MINAO)\cite{vanlenthe_starting_2006,almlof_principles_1982} initial guess. We use a Treutler\cite{treutler_efficient_1995} grid with level 3.
We note that for dipole moment prediction, this deviates from the grid setting of Ref.~\citenum{hait_how_2018}.
For reaction energy benchmarks and dipole computations, the SCF is run with a convergence tolerance on total energy change equal to \qty{5e-6}{\hartree} and a convergence tolerance on the orbital gradient norm equal to \qty{e-3}{\hartree}. The SCF is run for a maximum number of 60 steps. If convergence criteria are not met within 60 steps, we have a retry logic in place that works as follows. Each of the following attempts is tried, in case of failure of the previous one:
\begin{enumerate}
    \item DIIS\cite{pulay_convergence_1980, pulay_improved_1982} with a window size of 8 is used with no damping or level shift;
    \item Damping is set to 0.5 with DIIS\cite{pulay_convergence_1980, pulay_improved_1982} start cycle set to 7;
    \item The level shift parameter is incrementally set to 0.1, 0.3, and 0.5 if there is an iteration of the previous attempt with HOMO-LUMO gap (both in the spin-up and spin-down channel) below \qty{0.1}{\hartree}. This means that if an attempt fails to converge, but it has HOMO-LUMO gap (both in the spin-up and spin-down channel) bigger than \qty{0.1}{\hartree} at every step, no further attempts of increasing the level shift are made;
    \item As the last attempt, we enable the second-order solver from PySCF.\cite{sun_general_2017, seeger_selfconsistent_1977}
\end{enumerate}
If none of the attempts mentioned above manages to reach convergence within 60 steps or converges to a non-Aufbau higher energy state, it is marked as unconverged and manually converged either using a gradient descent algorithm, detailed in Sec.~\ref{sec:grad-descent}, or by manually tuning the initial guess after inspecting the convergence behavior and the Kohn-Sham orbitals.
Results on how often each attempt described in the SCF retry protocol was utilized when evaluating on W4-17, GMTKN55 and dipoles are reported in Table~\ref{tab:conv-analysis-summary}.

\begin{table}[t]
    \centering
    \caption{Analysis of the retry protocol for SCF convergence described in Sec.~\ref{sec:scf-retry} for our functional \clippy and baseline functionals. The percentages are over the evaluations on W4-17 (211 structures), GMTKN55 (2372 structures) and dipoles (151 structures), for a total of 2734 structures. In the last row we show the convergence statistics for \clippy in combination with the def2-TZVP basis set.}
    \tablefontsize 
\begin{tabular}{llllllll}
\toprule
Functional & First SCF & Damp+DIIS & Lev. shift 0.1 & Lev. shift 0.3 & Lev. shift 0.5 & Newton & Manual \\
\midrule
revPBE & 99.93\% & 0.04\% & 0.00\% & 0.00\% & 0.00\% & 0.00\% & 0.03\% \\
r\textsuperscript2SCAN & 99.93\% & 0.00\% & 0.04\% & 0.00\% & 0.00\% & 0.00\% & 0.03\% \\
B97M-V & 99.89\% & 0.07\% & 0.00\% & 0.00\% & 0.00\% & 0.00\% & 0.04\% \\
B3LYP & 100.00\% & 0.00\% & 0.00\% & 0.00\% & 0.00\% & 0.00\% & 0.00\% \\
M06-2X & 99.89\% & 0.00\% & 0.00\% & 0.00\% & 0.00\% & 0.00\% & 0.11\% \\
\(\omega\)B97X-V & 100.00\% & 0.00\% & 0.00\% & 0.00\% & 0.00\% & 0.00\% & 0.00\% \\
\(\omega\)B97M-V & 99.89\% & 0.00\% & 0.00\% & 0.00\% & 0.00\% & 0.00\% & 0.11\% \\
\clippyOnePointOne & 98.87\% & 0.18\% & 0.33\% & 0.22\% & 0.00\% & 0.22\% & 0.18\% \\
\midrule
\clippyOnePointOne (TZVP) & 99.38\% & 0.15\% & 0.29\% & 0.00\% & 0.00\% & 0.07\% & 0.11\% \\
\bottomrule
\end{tabular}

    \label{tab:conv-analysis-summary}
\end{table}

For the data and model ablation in \cref{fig:model-data-ablation}, we followed the same convergence protocol. For one model trained on data group A, SCF convergence was not reached for $^2\text{H}_6 \text{N}_2^+$ in Diet GMTKN55; all other models converged without issues. For this model only, we evaluated the functional on the B3LYP density for that compound and the other compounds in the corresponding reaction $^2\text{H}_6\text{N}_2^+ \rightarrow \text{H}_3\text{N} + {^2\text{H}_3\text{N}^+}$, and used the resulting energy in the WTMAD-2 metric. The impact is minor: this yields a WTMAD-2 of 11.10~kcal/mol, compared to 11.19~kcal/mol when excluding the reaction entirely (99/100 reactions).

For geometry optimization, the same logic described above is used for the initial geometry, but with stricter convergence criteria (see Sec.~\ref{sec:suppgeom}). 
The settings are then kept fixed in subsequent geometry optimization steps.

We also report the average number of iterations required for convergence in the first SCF attempt in Table~\ref{tab:iter-stats}.

\begin{table}[t]
    \centering
    \caption{Average number of SCF iterations required for convergence in the first SCF attempt. Only systems that were converged in the first attempt by all functionals are considered (208 structures for W4-17, 2343 structures for GMTKN55, and 146 structures for dipoles).}
    \tablefontsize 
\begin{tabular}{lrrrrrrrrr}
\toprule
 & revPBE & r\textsuperscript2SCAN & B97M-V & B3LYP & M06-2X & \(\omega\)B97X-V & \(\omega\)B97M-V & \clippyOnePointOne & \clippyOnePointOne (TZVP) \\
\midrule
W4-17 & 7.8 & 7.7 & 7.7 & 7.5 & 7.6 & 7.6 & 7.6 & 8.2 & 7.8 \\
GMTKN55 & 9.1 & 8.7 & 8.8 & 8.3 & 8.0 & 8.0 & 8.0 & 9.3 & 8.9 \\
Dipoles & 7.9 & 7.8 & 7.9 & 7.5 & 7.8 & 7.7 & 7.7 & 9.4 & 9.0 \\
Total & 8.9 & 8.6 & 8.7 & 8.2 & 8.0 & 8.0 & 8.0 & 9.2 & 8.9 \\
\bottomrule
\end{tabular}

    \label{tab:iter-stats}
\end{table}

\subsection{SCF with orbital gradient descent fallback}
\label{sec:grad-descent}
As the SCF retry protocol we describe in Sec.~\ref{sec:scf-retry} does not converge 100\% of the time, we introduce a more robust routine.
We start out with a MINAO initial-guess density matrix\cite{vanlenthe_starting_2006,almlof_principles_1982} and diagonalize the Fock matrix to obtain a permissible idempotent density matrix.
In the first cycle, we compute the total energy.
From the second cycle and beyond, we first check if diagonalization with DIIS\cite{pulay_convergence_1980, pulay_improved_1982} leads to a decrease in energy; if not, we resort to gradient descent with line search in the space of molecular orbital coefficients $C$.

In what follows, we discuss the spin-agnostic case to simplify the notation.
The matrix $C$ can be decomposed into the occupied and virtual parts $\begin{bmatrix} C_\text{occ}, C_\text{vir} \end{bmatrix}$, where only the occupied part contributes to the density matrix $\oneRDM_{ij} = \sum_{k=1}^{N}  C_{ik} C_{jk}$, for $1\leq i,j\leq B$, which means $\oneRDM = C_\text{occ}{C_\text{occ}}^\top$.
From the density matrix, we generate the features needed to compute the exchange-correlation energy on the grid and compute the other energy components including the kinetic energy, Hartree electrostatics, and external potential.
We can write the total energy as a function of $\oneRDM$, $E_\text{tot}(\oneRDM)$.
The Fock matrix is defined as the gradient with respect to the density matrix $F\isDefinedAs\frac{\ud E_\text{tot}}{\ud \oneRDM}$.
When computing the gradient with respect to the orbital matrix $C$, we need to take into account the fact that $\oneRDM$ is invariant to the rotation of the columns of $C_\text{occ}$ and that all columns of $C$ need to stay orthogonal.
This can be done by looking at the virtual-occupied blocks of the Fock matrix
\begin{align}
    \Delta_C^\text{hor} = \begin{bmatrix}
        (I-\oneRDM)F C_\text{occ}, -\oneRDM FC_\text{vir}
    \end{bmatrix} .
\end{align}
This is also known as the horizontal lift of the Riemannian gradient in the Grassmannian manifold (space of density matrices) to the tangent space of the Stiefel manifold (space of orbital coefficients).\cite{bendokat_grassmann_2024}

We use the negative gradient as the search direction, and with step size $\alpha$ we use the exponential map to produce the new orbital matrix $C'\leftarrow C \exp\left(-\alpha C^\top \Delta_C^\text{hor}\right)$.
The form of the update can be equivalently derived by taking the gradient of $C\exp(A)$ with respect to an antisymmetric matrix $A$.\cite{helgaker_molecular_2000}
We initialize the step size to be $1.0$ for each calculation, and adaptively set the step size of the subsequent iterations to be $1.1$ times the previous admissible step size.
Within each iteration, a step size is considered admissible if it satisfies the Armijo rule
\begin{align}
    E_\text{tot}(C') - E_\text{tot}(C) \leq - c_1 \alpha \norm{\Delta_C^\text{hor}}^2
\end{align}
where we set $c_1=0.1$.
We do not check the curvature condition commonly used in line search to avoid unnecessary backward computation.
If the Armijo rule is violated, we form a quadratic approximation to $E_\text{tot}$ as a function of $\alpha$, and use the minimizer as the new trial.\cite{nocedal_numerical_2006}
To prevent the line search strategy from being too aggressive or too conservative, we clamp the search space to be within $[0.1\alpha, 0.9\alpha]$.

Empirically, we find that robustness of gradient descent with line search hinges on good alignment between numerical gradient and energy estimation.
Therefore, all computations with gradient descent have been performed in double precision.
For some computation attempts using gradient descent, we first converge with a downgraded basis, and then project onto the bigger basis set.
For example, we first use def2-TZVP and then def2-QZVP.

\subsection{Geometry optimization: implementation}
\label{sec:suppgeom}
For geometry optimization, we need the gradient of the total energy with respect to the nuclear coordinates.
For standard approximate DFT functionals, this entails the evaluation of the following contributions to the nuclear gradient:
\begin{enumerate}
\item The explicit dependence on the nuclear positions in the nuclear potential (Hellmann--Feynman forces)
\item The dependence of the atomic centered basis functions (Pulay forces)\cite{pulay_initio_1969}
\item The dependence of the integration grid on the atomic positions.
\end{enumerate}
For sufficiently large grid the contribution from the grid integration tends to be small, so it can be neglected for traditional functionals.
However, since the \clippy architecture does not only use the grid for integration of the final energy density, but also for down- and up-sampling, we included it in computing the total forces. Additionally, we need to take into account
\begin{enumerate}[resume]
\item The explicit dependence of the $E_\xc$ functional on the nuclear positions due to the coarse points.
\end{enumerate}
For the initial geometry optimization step, we use the same retry protocol as described in Sec.~\ref{sec:scf-retry}, and during the course of geometry optimization we stick to the SCF method which successfully converged the electronic structure of the initial geometry. 
We use the \textsc{geomeTRIC} package\cite{wang_geometry_2016} to optimize the geometry. The convergence tolerance on the energy change was tightened to \qty{5e-7}{\hartree} and on the gradient norm to \qty{e-5}{\hartree\per\bohr}, which was important to obtain accurate enough nuclear gradients for \textsc{geomeTRIC} to always converge.

\subsection{Computational cost: GPU and CPU implementation}
\label{supp:cost-settings}

The functional was integrated with the GauXC library\cite{petrone_efficient_2018,williams-young_efficient_2020,williams-young_distributed_2023} to evaluate the \clippy functional from the TorchScript exported checkpoint.
In the initial integration, the features computed on the grid points are collected in GauXC, passed to the model for computing the exchange-correlation energy and potential and afterwards redistributed into the batches handled by GauXC.
The computation of the exchange-correlation potential is done by backpropagation using the TorchScript framework.

The GauXC library supports evaluation both on GPU and CPU.
For the computational cost analysis a standalone driver of GauXC is used starting from an input density matrix and evaluating the exchange-correlation energy and potential.
Calculations for GPU timings were performed on Azure NC24ADS V4 A100 virtual machines.
CPU timings were performed on Azure E32ADS V5 virtual machines.

For these calculations a def2-TZVP basis set was used for all functionals and additionally semi-numerical exchange integrals for all hybrid functionals, GM3 grid level for integrating the exchange-correlation energy, Treutler grid pruning and Mura–Knowles radial integration scheme were used.

\section{Additional results}
\label{sec:suppadditional}

\subsection{Spin-symmetry broken solutions for multi-reference systems}
\label{sec:spin-broken}

For systems with a high multi-reference character, Kohn-Sham DFT is in practice prone to breaking the spin symmetry.\cite{perdew_symmetry_2023} These spin broken states are lower in energy and almost always give an answer closer to the ground truth. Therefore, to assess whether the same behaviour holds true for \clippyOnePointOne on the multi-reference subset of W4-17,\cite{karton_w417_2017} we allow for spin-symmetry breaking.

\begin{table}[htbp]
    \centering
    \caption{Total energies and total spin \(\langle S^2\rangle\) for the clustering of spin-symmetry broken stationary SCF solutions for the various functionals. All quantities are in atomic units.}
    \label{tab:w4-17:mr+break_down}
    \vspace{-0.8em}
    
    \begin{minipage}[t]{0.48\textwidth}
        \subcaption{revPBE}
        \vspace{-0.5em}
        \centering
        
\tablefontsize
\begin{tabularx}{0.95\textwidth}{l p{0.45\textwidth} X}
    \toprule
    Molecule & Total energy & \(\langle S^2 \rangle\)\\
    \midrule
    \ce{^{2}ClOO} & 
    [-610.470] & 
    [0.77] \\
    \ce{^{2}ClO3} & 
    [-685.617] & 
    [0.75] \\
    \ce{O3} &
    [-225.496] & 
    [0.00] \\
    \ce{ClF3} &
    [-759.538] & 
    [0.00] \\
    \ce{^{3}B2} & 
    [-49.431,-49.426] & 
    [2.54,2.01] \\
    \ce{C2} &
    [-75.932] & 
    [0.88] \\
    \ce{F2O} &
    [-274.733] & 
    [0.00] \\
    \ce{ClF5} &
    [-959.129] & 
    [-0.00] \\
    \ce{S3} &
    [-1194.448] & 
    [0.00] \\
    \ce{^{2}FO2} & 
    [-250.149] & 
    [0.76] \\
    \ce{FOOF} &
    [-349.930] & 
    [0.00] \\
    \ce{^{2}OClO} & 
    [-610.467] & 
    [0.75] \\
    \ce{Cl2O} &
    [-995.412] & 
    [0.00] \\
    \ce{^{2}OF} & 
    [-174.920] & 
    [0.75] \\
    \ce{S4-C$_{2v}$} &
    [-1592.618] & 
    [0.00] \\
    \ce{^{1}BN} & 
    [-79.412] & 
    [0.00] \\
    \ce{ClOOCl} &
    [-1070.580] & 
    [0.00] \\
    \bottomrule
\end{tabularx}

    \end{minipage}\hfill
    \begin{minipage}[t]{0.48\textwidth}
        \subcaption{r\textsuperscript2SCAN}
        \vspace{-0.5em}
        \centering
        
\tablefontsize
\begin{tabularx}{0.95\textwidth}{l p{0.45\textwidth} X}
    \toprule
    Molecule & Total energy & \(\langle S^2 \rangle\)\\
    \midrule
    \ce{^{2}ClOO} & 
    [-610.477] & 
    [0.88] \\
    \ce{^{2}ClO3} & 
    [-685.618] & 
    [0.76] \\
    \ce{O3} &
    [-225.430,-225.429] & 
    [0.21,0.00] \\
    \ce{ClF3} &
    [-759.554] & 
    [0.00] \\
    \ce{^{3}B2} & 
    [-49.405,-49.389] & 
    [2.72,2.01] \\
    \ce{C2} &
    [-75.896,-75.865] & 
    [1.01,0.00] \\
    \ce{F2O} &
    [-274.687] & 
    [0.00] \\
    \ce{ClF5} &
    [-959.135] & 
    [0.00] \\
    \ce{S3} &
    [-1194.567] & 
    [0.00] \\
    \ce{^{2}FO2} & 
    [-250.093] & 
    [0.79] \\
    \ce{FOOF} &
    [-349.860] & 
    [0.00] \\
    \ce{^{2}OClO} & 
    [-610.479] & 
    [0.76] \\
    \ce{Cl2O} &
    [-995.489] & 
    [0.00] \\
    \ce{^{2}OF} & 
    [-174.888] & 
    [0.75] \\
    \ce{S4-C$_{2v}$} &
    [-1592.776] & 
    [0.00] \\
    \ce{^{1}BN} & 
    [-79.367] & 
    [0.00] \\
    \ce{ClOOCl} &
    [-1070.638] & 
    [-0.00] \\
    \bottomrule
\end{tabularx}

    \end{minipage}
    
    \vspace{-0.4em}
    
    \begin{minipage}[t]{0.48\textwidth}
        \subcaption{B97M-V}
        \vspace{-0.5em}
        \centering
        
\tablefontsize
\begin{tabularx}{0.95\textwidth}{l p{0.45\textwidth} X}
    \toprule
    Molecule & Total energy & \(\langle S^2 \rangle\)\\
    \midrule
    \ce{^{2}ClOO} & 
    [-610.401] & 
    [0.85] \\
    \ce{^{2}ClO3} & 
    [-685.548] & 
    [0.75] \\
    \ce{O3} &
    [-225.458] & 
    [0.00] \\
    \ce{ClF3} &
    [-759.462] & 
    [0.00] \\
    \ce{^{3}B2} & 
    [-49.455,-49.450] & 
    [2.48,2.01] \\
    \ce{C2} &
    [-75.941,-75.918] & 
    [0.96,0.00] \\
    \ce{F2O} &
    [-274.703] & 
    [0.00] \\
    \ce{ClF5} &
    [-959.038] & 
    [0.00] \\
    \ce{S3} &
    [-1194.334] & 
    [0.00] \\
    \ce{^{2}FO2} & 
    [-250.110] & 
    [0.79] \\
    \ce{FOOF} &
    [-349.882] & 
    [0.00] \\
    \ce{^{2}OClO} & 
    [-610.402] & 
    [0.75] \\
    \ce{Cl2O} &
    [-995.318] & 
    [0.00] \\
    \ce{^{2}OF} & 
    [-174.901] & 
    [0.75] \\
    \ce{S4-C$_{2v}$} &
    [-1592.462] & 
    [0.00] \\
    \ce{^{1}BN} & 
    [-79.418] & 
    [0.00] \\
    \ce{ClOOCl} &
    [-1070.476] & 
    [0.00] \\
    \bottomrule
\end{tabularx}

    \end{minipage}\hfill
    \begin{minipage}[t]{0.48\textwidth}
        \subcaption{B3LYP}
        \vspace{-0.5em}
        \centering
        
\tablefontsize
\begin{tabularx}{0.95\textwidth}{l p{0.45\textwidth} X}
    \toprule
    Molecule & Total energy & \(\langle S^2 \rangle\)\\
    \midrule
    \ce{^{2}ClOO} & 
    [-610.462] & 
    [0.98] \\
    \ce{^{2}ClO3} & 
    [-685.591] & 
    [0.76] \\
    \ce{O3} &
    [-225.438,-225.435] & 
    [0.34,0.00] \\
    \ce{ClF3} &
    [-759.548] & 
    [0.00] \\
    \ce{^{3}B2} & 
    [-49.393] & 
    [2.00] \\
    \ce{C2} &
    [-75.890,-75.871] & 
    [0.78,0.00] \\
    \ce{F2O} &
    [-274.703] & 
    [0.00] \\
    \ce{ClF5} &
    [-959.120] & 
    [0.00] \\
    \ce{S3} &
    [-1194.515] & 
    [0.00] \\
    \ce{^{2}FO2} & 
    [-250.100] & 
    [0.83] \\
    \ce{FOOF} &
    [-349.869] & 
    [0.00] \\
    \ce{^{2}OClO} & 
    [-610.459] & 
    [0.76] \\
    \ce{Cl2O} &
    [-995.461] & 
    [0.00] \\
    \ce{^{2}OF} & 
    [-174.900] & 
    [0.75] \\
    \ce{S4-C$_{2v}$} &
    [-1592.702] & 
    [0.00] \\
    \ce{^{1}BN} & 
    [-79.376] & 
    [0.00] \\
    \ce{ClOOCl} &
    [-1070.609] & 
    [0.00] \\
    \bottomrule
\end{tabularx}

    \end{minipage}
    
    \vspace{-0.4em}
    
    \begin{minipage}[t]{0.48\textwidth}
        \subcaption{M06-2X}
        \vspace{-0.5em}
        \centering
        
\tablefontsize
\begin{tabularx}{0.95\textwidth}{l p{0.45\textwidth} X}
    \toprule
    Molecule & Total energy & \(\langle S^2 \rangle\)\\
    \midrule
    \ce{^{2}ClOO} & 
    [-610.482] & 
    [1.19] \\
    \ce{^{2}ClO3} & 
    [-685.616,-685.546] & 
    [0.76,0.77] \\
    \ce{O3} &
    [-225.428,-225.420] & 
    [0.50,0.00] \\
    \ce{ClF3} &
    [-759.568] & 
    [0.00] \\
    \ce{^{3}B2} & 
    [-49.401] & 
    [2.00] \\
    \ce{C2} &
    [-75.879] & 
    [0.00] \\
    \ce{F2O} &
    [-274.694,-274.538] & 
    [0.00,1.03] \\
    \ce{ClF5} &
    [-959.134] & 
    [0.00] \\
    \ce{S3} &
    [-1194.583] & 
    [0.00] \\
    \ce{^{2}FO2} & 
    [-250.080] & 
    [1.00] \\
    \ce{FOOF} &
    [-349.844] & 
    [0.00] \\
    \ce{^{2}OClO} & 
    [-610.485] & 
    [0.76] \\
    \ce{Cl2O} &
    [-995.528] & 
    [0.00] \\
    \ce{^{2}OF} & 
    [-174.893] & 
    [0.75] \\
    \ce{S4-C$_{2v}$} &
    [-1592.786] & 
    [0.00] \\
    \ce{^{1}BN} & 
    [-79.380] & 
    [0.00] \\
    \ce{ClOOCl} &
    [-1070.674] & 
    [0.00] \\
    \bottomrule
\end{tabularx}

    \end{minipage}\hfill
    \begin{minipage}[t]{0.48\textwidth}
        \subcaption{\(\omega\)B97X-V}
        \vspace{-0.5em}
        \centering
        
\tablefontsize
\begin{tabularx}{0.95\textwidth}{l p{0.45\textwidth} X}
    \toprule
    Molecule & Total energy & \(\langle S^2 \rangle\)\\
    \midrule
    \ce{^{2}ClOO} & 
    [-610.438,-610.429,-610.429] & 
    [1.21,0.98,0.82] \\
    \ce{^{2}ClO3} & 
    [-685.579,-685.511] & 
    [0.76,0.77] \\
    \ce{O3} &
    [-225.446,-225.440] & 
    [0.47,0.00] \\
    \ce{ClF3} &
    [-759.518] & 
    [0.00] \\
    \ce{^{3}B2} & 
    [-49.324] & 
    [2.03] \\
    \ce{C2} &
    [-75.872] & 
    [0.00] \\
    \ce{F2O} &
    [-274.703] & 
    [0.00] \\
    \ce{ClF5} &
    [-959.090] & 
    [0.00] \\
    \ce{S3} &
    [-1194.437] & 
    [0.00] \\
    \ce{^{2}FO2} & 
    [-250.098,-250.024] & 
    [0.93,0.77] \\
    \ce{FOOF} &
    [-349.865] & 
    [0.00] \\
    \ce{^{2}OClO} & 
    [-610.441] & 
    [0.76] \\
    \ce{Cl2O} &
    [-995.419] & 
    [0.00] \\
    \ce{^{2}OF} & 
    [-174.901] & 
    [0.75] \\
    \ce{S4-C$_{2v}$} &
    [-1592.589] & 
    [0.00] \\
    \ce{^{1}BN} & 
    [-79.376] & 
    [0.00] \\
    \ce{ClOOCl} &
    [-1070.569] & 
    [0.00] \\
    \bottomrule
\end{tabularx}

    \end{minipage}
    
    \vspace{-0.4em}
    
    \begin{minipage}[t]{0.48\textwidth}
        \subcaption{\(\omega\)B97M-V}
        \vspace{-0.5em}
        \centering
        
\tablefontsize
\begin{tabularx}{0.95\textwidth}{l p{0.45\textwidth} X}
    \toprule
    Molecule & Total energy & \(\langle S^2 \rangle\)\\
    \midrule
    \ce{^{2}ClOO} & 
    [-610.497,-610.436] & 
    [1.13,0.77] \\
    \ce{^{2}ClO3} & 
    [-685.647,-685.580] & 
    [0.76,0.77] \\
    \ce{O3} &
    [-225.467,-225.462] & 
    [0.42,0.00] \\
    \ce{ClF3} &
    [-759.599] & 
    [0.00] \\
    \ce{^{3}B2} & 
    [-49.398,-49.392,-49.322] & 
    [2.59,2.00,2.02] \\
    \ce{C2} &
    [-75.881] & 
    [0.00] \\
    \ce{F2O} &
    [-274.732,-274.561] & 
    [0.00,1.02] \\
    \ce{ClF5} &
    [-959.194] & 
    [0.00] \\
    \ce{S3} &
    [-1194.570] & 
    [0.00] \\
    \ce{^{2}FO2} & 
    [-250.121,-250.049] & 
    [0.90,0.77] \\
    \ce{FOOF} &
    [-349.901] & 
    [0.00] \\
    \ce{^{2}OClO} & 
    [-610.502] & 
    [0.76] \\
    \ce{Cl2O} &
    [-995.521] & 
    [0.00] \\
    \ce{^{2}OF} & 
    [-174.918] & 
    [0.75] \\
    \ce{S4-C$_{2v}$} &
    [-1592.775,-1592.767] & 
    [0.68,0.00] \\
    \ce{^{1}BN} & 
    [-79.386] & 
    [0.00] \\
    \ce{ClOOCl} &
    [-1070.679] & 
    [0.00] \\
    \bottomrule
\end{tabularx}

    \end{minipage}\hfill
    \begin{minipage}[t]{0.48\textwidth}
        \subcaption{\clippyOnePointOne}
        \vspace{-0.5em}
        \centering
        
\tablefontsize
\begin{tabularx}{0.95\textwidth}{l p{0.45\textwidth} X}
    \toprule
    Molecule & Total energy & \(\langle S^2 \rangle\)\\
    \midrule
    \ce{^{2}ClOO} & 
    [-610.486] & 
    [1.01] \\
    \ce{^{2}ClO3} & 
    [-685.628] & 
    [0.75] \\
    \ce{O3} &
    [-225.423] & 
    [0.00] \\
    \ce{ClF3} &
    [-759.566] & 
    [0.00] \\
    \ce{^{3}B2} & 
    [-49.411,-49.404,-49.393] & 
    [2.64,2.01,2.74] \\
    \ce{C2} &
    [-75.910,-75.909,-75.904] & 
    [0.89,1.01,0.68] \\
    \ce{F2O} &
    [-274.683] & 
    [0.00] \\
    \ce{ClF5} &
    [-959.132] & 
    [0.00] \\
    \ce{S3} &
    [-1194.624] & 
    [0.00] \\
    \ce{^{2}FO2} & 
    [-250.075] & 
    [0.78] \\
    \ce{FOOF} &
    [-349.841] & 
    [0.00] \\
    \ce{^{2}OClO} & 
    [-610.490] & 
    [0.76] \\
    \ce{Cl2O} &
    [-995.543] & 
    [0.00] \\
    \ce{^{2}OF} & 
    [-174.883] & 
    [0.76] \\
    \ce{S4-C$_{2v}$} &
    [-1592.849] & 
    [0.00] \\
    \ce{^{1}BN} & 
    [-79.404] & 
    [1.02] \\
    \ce{ClOOCl} &
    [-1070.686] & 
    [0.00] \\
    \bottomrule
\end{tabularx}

    \end{minipage}

\end{table}

Since the initial guess in PySCF has proper spin-symmetry and the gradient with respect to spin-symmetry breaking directions is zero, we kick start the search for symmetry-breaking solutions by adding random noise scaled between \([-0.1, 0.1]\) to the elements of the spin-resolved initial one-electron reduced-density-matrix (1RDM) guess. This randomly perturbed initial guess is subsequently purified by first symmetrizing it and subsequently clipping its eigenvalues to be bounded between 0 and 1. For the baseline functionals, we did not normalize the trace to recover the correct number of electrons, but for \clippy a proper projection to the correct number of electrons was applied.\cite{cances_projected_2008} Since multiple local minima can exist for spin-symmetry breaking solutions, for each molecule of the multi-reference subset of W4-17, we ran 50 calculations with different random seeds to increase the probability of finding the lowest energy solution. We have run these calculations with tighter convergence criteria: a tolerance on the energy difference equal to \qty{e-8}{\hartree} and a convergence tolerance on gradient norm equal to \qty{e-4}{\hartree}. We subsequently clustered the converged results with an energy criterion of \qty{e-3}{\hartree} and a total spin \(\langle S^2\rangle\) criterion of \qty{e-1}{\square\planckbar}. The result for the clustering including the baseline functionals can be found in Table~\ref{tab:w4-17:mr+break_down}.

\begin{figure}[t]
  \centering
  \includegraphics[width=0.90\textwidth]{figures/tae-noise-effect-final-s0.pdf}
  \caption{The effect of using randomly perturbed initial guesses which allow for spin-symmetry breaking for \clippy on the 17 multi-referential structures of W4-17.}
  \label{fig:spin-breaking-effect-on-tae-skala}
\end{figure}

The impact of allowing \clippyOnePointOne to break spin symmetry on the TAE error is shown in Fig.~\ref{fig:spin-breaking-effect-on-tae-skala}. This mainly has an effect on the \ce{^1BN}, \ce{B2} and \ce{C2} molecules, which show the largest energy lowering upon symmetry breaking. For those three molecules we see a clear improvement in the TAE when we allow for symmetry-breaking. The remaining small deviations are not due to symmetry breaking, but the effect of the tighter convergence tolerance in conjunction with the multi SCF attempts increasing the probability of finding the global minimum.

\subsection{GMTKN55 benchmark}

Table~\ref{tab:gmtkn55-examples} highlights reactions from GMTKN55 for which the absolute difference in prediction error between \clippy and \(\omega\)B97M-V is the largest.

\begin{table}[tbp]
    \centering
    \caption{Reactions in GMTKN55 with the largest difference in energy prediction accuracy between \clippy and \(\omega\)B97M-V.}
    \label{tab:gmtkn55-examples}
    
\tablefontsize
\begin{tabularx}{\textwidth}{X r r@{\,}rr}
    \toprule
    
    & \multicolumn{1}{r}{Ref. energy}
    & \multicolumn{3}{l}{Energy error}
    \\
    & \multicolumn{1}{r}{\textcolor{gray}{[kcal/mol]}}
    & \multicolumn{1}{l}{\clippy \textcolor{gray}{($a$)}}
    & \multicolumn{1}{r}{\(\omega\)B97M-V \textcolor{gray}{($b$)}}
    & \multicolumn{1}{r}{\textcolor{gray}{$|b| - |a|$}}
    \\
    \cmidrule(lr){2-2}
    \cmidrule(lr){3-5}
    \textbf{\clippy\ performs better than \(\omega\)B97M-V} \\[1ex]
    $19\text{H}_{2}+2^{2}\text{Al}\text{B}_{3}\text{H}_{6}\text{Li}\text{Mg}_{2}\text{O}\text{Si}_{2} \rightarrow ^{3}\text{O}_{2}+6\text{B}\text{H}_{3}+2\text{H}\text{Li}+4\text{H}_{2}\text{Mg}+4\text{H}_{4}\text{Si}+2\text{Al}\text{H}_{3}$
    & 326.76
    & 4.30
    & 56.33
    & -52.04 \\
    \quad \textcolor{gray}{mb16-43} \\[1ex]
    $7\text{H}_{2}+2\text{Al}_{2}\text{B}\text{Cl}\text{H}_{7}\text{Mg}_{2}\text{Na}\text{O}_{2} \rightarrow 2^{3}\text{O}_{2}+2\text{B}\text{H}_{3}+2\text{H}\text{Na}+\text{Cl}_{2}+4\text{H}_{2}\text{Mg}+4\text{Al}\text{H}_{3}$
    & 806.23
    & 0.42
    & 45.23
    & -44.81 \\
    \quad \textcolor{gray}{mb16-43} \\[1ex]
    $13\text{H}_{2}+2\text{Al}\text{C}_{3}\text{F}\text{H}_{6}\text{Li}\text{Mg}\text{Na}\text{O}\text{S} \rightarrow \text{F}_{2}+^{3}\text{O}_{2}+^{3}\text{S}_{2}+2\text{H}\text{Na}+6\text{C}\text{H}_{4}+2\text{H}\text{Li}+2\text{H}_{2}\text{Mg}+2\text{Al}\text{H}_{3}$
    & 518.66
    & -2.14
    & 32.32
    & -30.18 \\
    \quad \textcolor{gray}{mb16-43} \\[1ex]
    $6\text{H}_{2}+2^{2}\text{B}_{2}\text{F}\text{H}_{7}\text{Li}\text{Mg}\text{O}\text{S}_{2}\text{Si} \rightarrow \text{F}_{2}+^{3}\text{O}_{2}+2^{3}\text{S}_{2}+4\text{B}\text{H}_{3}+2\text{H}\text{Li}+2\text{H}_{2}\text{Mg}+2\text{H}_{4}\text{Si}$
    & 737.79
    & -0.16
    & 27.64
    & -27.48 \\
    \quad \textcolor{gray}{mb16-43} \\[1ex]
    $19\text{H}_{2}+2\text{Al}\text{B}_{4}\text{C}\text{F}\text{H}_{4}\text{O}_{2}\text{S}_{2}\text{Si} \rightarrow \text{F}_{2}+2^{3}\text{O}_{2}+2^{3}\text{S}_{2}+8\text{B}\text{H}_{3}+2\text{C}\text{H}_{4}+2\text{H}_{4}\text{Si}+2\text{Al}\text{H}_{3}$
    & 629.78
    & 2.30
    & 28.00
    & -25.70 \\
    \quad \textcolor{gray}{mb16-43} \\[1ex]
    $12\text{H}_{2}+2^{2}\text{Al}\text{B}\text{C}\text{F}_{5}\text{H}_{4}\text{Li}_{2}\text{S}\text{Si} \rightarrow 5\text{F}_{2}+^{3}\text{S}_{2}+2\text{B}\text{H}_{3}+2\text{C}\text{H}_{4}+4\text{H}\text{Li}+2\text{H}_{4}\text{Si}+2\text{Al}\text{H}_{3}$
    & 1290.74
    & 2.08
    & 26.09
    & -24.01 \\
    \quad \textcolor{gray}{mb16-43} \\[1ex]
    $10\text{H}_{2}+2^{2}\text{Al}\text{B}_{3}\text{C}\text{H}_{8}\text{Mg}\text{N}\text{O} \rightarrow ^{3}\text{O}_{2}+\text{N}_{2}+6\text{B}\text{H}_{3}+2\text{C}\text{H}_{4}+2\text{H}_{2}\text{Mg}+2\text{Al}\text{H}_{3}$
    & 368.49
    & 0.06
    & 22.00
    & -21.94 \\
    \quad \textcolor{gray}{mb16-43} \\[1ex]
    $\text{C}_{2} \rightarrow 2^{3}\text{C}$
    & 147.02
    & 9.60
    & 31.39
    & -21.79 \\
    \quad \textcolor{gray}{w4-11} \\[1ex]
    $17\text{H}_{2}+2\text{Al}\text{B}_{3}\text{C}\text{F}\text{H}_{7}\text{O}\text{Si}_{2} \rightarrow \text{F}_{2}+^{3}\text{O}_{2}+6\text{B}\text{H}_{3}+2\text{C}\text{H}_{4}+4\text{H}_{4}\text{Si}+2\text{Al}\text{H}_{3}$
    & 356.85
    & 4.29
    & 24.09
    & -19.79 \\
    \quad \textcolor{gray}{mb16-43} \\[1ex]
    $20\text{H}_{2}+2\text{B}_{4}\text{C}\text{Cl}\text{F}\text{H}_{6}\text{Mg}\text{Si}_{2} \rightarrow \text{F}_{2}+8\text{B}\text{H}_{3}+\text{Cl}_{2}+2\text{C}\text{H}_{4}+2\text{H}_{2}\text{Mg}+4\text{H}_{4}\text{Si}$
    & 215.39
    & 7.40
    & 25.51
    & -18.11 \\
    \quad \textcolor{gray}{mb16-43} \\[1ex]
    $9\text{H}_{2}+2\text{Al}_{2}\text{B}_{2}\text{Cl}\text{H}_{8}\text{Li}\text{O}\text{Si} \rightarrow ^{3}\text{O}_{2}+4\text{B}\text{H}_{3}+\text{Cl}_{2}+2\text{H}\text{Li}+2\text{H}_{4}\text{Si}+4\text{Al}\text{H}_{3}$
    & 384.89
    & 3.58
    & 20.21
    & -16.63 \\
    \quad \textcolor{gray}{mb16-43} \\[1ex]
    $\ldots$ \\[1ex]
    \textbf{\(\omega\)B97M-V performs better than \clippy} \\[1ex]
    $\text{C}\text{Cl}_{2}\text{H}_{3}^{-} \rightarrow \text{C}\text{Cl}_{2}\text{H}_{3}^{-}$
    & -13.50
    & -3.64
    & -0.07
    & 3.57 \\
    \quad \textcolor{gray}{bh76} \\[1ex]
    $^{2}\text{H}_{4}\text{O}_{2}^{+} \rightarrow ^{2}\text{H}_{2}\text{O}^{+}+\text{H}_{2}\text{O}$
    & 9.30
    & -15.32
    & -11.72
    & 3.60 \\
    \quad \textcolor{gray}{sie4x4} \\[1ex]
    $9\text{H}_{2}+2^{2}\text{B}_{2}\text{C}\text{F}_{3}\text{H}_{6}\text{Li}\text{N}\text{S}\text{Si} \rightarrow 3\text{F}_{2}+^{3}\text{S}_{2}+\text{N}_{2}+4\text{B}\text{H}_{3}+2\text{C}\text{H}_{4}+2\text{H}\text{Li}+2\text{H}_{4}\text{Si}$
    & 695.10
    & -4.20
    & 0.42
    & 3.78 \\
    \quad \textcolor{gray}{mb16-43} \\[1ex]
    $^{3}\text{O}_{2} \rightarrow ^{2}\text{O}_{2}^{+}$
    & 277.73
    & -7.74
    & -3.54
    & 4.20 \\
    \quad \textcolor{gray}{g21ip} \\[1ex]
    $^{2}\text{H}_{2}^{+} \rightarrow ^{2}\text{H}$
    & 58.90
    & -6.35
    & -1.84
    & 4.51 \\
    \quad \textcolor{gray}{sie4x4} \\[1ex]
    $5\text{H}_{2}+2\text{B}\text{C}_{2}\text{Cl}\text{F}_{2}\text{H}_{8}\text{Li}\text{Na} \rightarrow 2\text{F}_{2}+2\text{B}\text{H}_{3}+2\text{H}\text{Na}+\text{Cl}_{2}+4\text{C}\text{H}_{4}+2\text{H}\text{Li}$
    & 436.40
    & -8.27
    & 2.91
    & 5.35 \\
    \quad \textcolor{gray}{mb16-43} \\[1ex]
    $\text{Li}_{8} \rightarrow 4\text{Li}_{2}$
    & 86.47
    & -13.34
    & 7.14
    & 6.19 \\
    \quad \textcolor{gray}{alk8} \\[1ex]
    $^{4}\text{N} \rightarrow ^{3}\text{N}^{+}$
    & 335.30
    & 7.00
    & 0.43
    & 6.57 \\
    \quad \textcolor{gray}{g21ip} \\[1ex]
    $^{2}\text{C}_{2}\text{H} \rightarrow ^{2}\text{H}+2^{3}\text{C}$
    & 266.16
    & 17.03
    & 4.68
    & 12.35 \\
    \quad \textcolor{gray}{w4-11} \\[1ex]
    $\text{Ca}\text{O} \rightarrow ^{3}\text{O}+\text{Ca}$
    & 96.20
    & -15.01
    & -1.65
    & 13.36 \\
    \quad \textcolor{gray}{alkbde10} \\[1ex]
    $^{2}\text{H}_{6}\text{N}_{2}^{+} \rightarrow \text{H}_{3}\text{N}+^{2}\text{H}_{3}\text{N}^{+}$
    & 13.40
    & -67.61
    & -7.25
    & 60.35 \\
    \quad \textcolor{gray}{sie4x4} \\[1ex]
    \bottomrule
\end{tabularx}

\end{table}

\subsection{Convergence with respect to the grid size}
\begin{figure}[t!]
    \centering
    \includegraphics[width=0.6\linewidth]{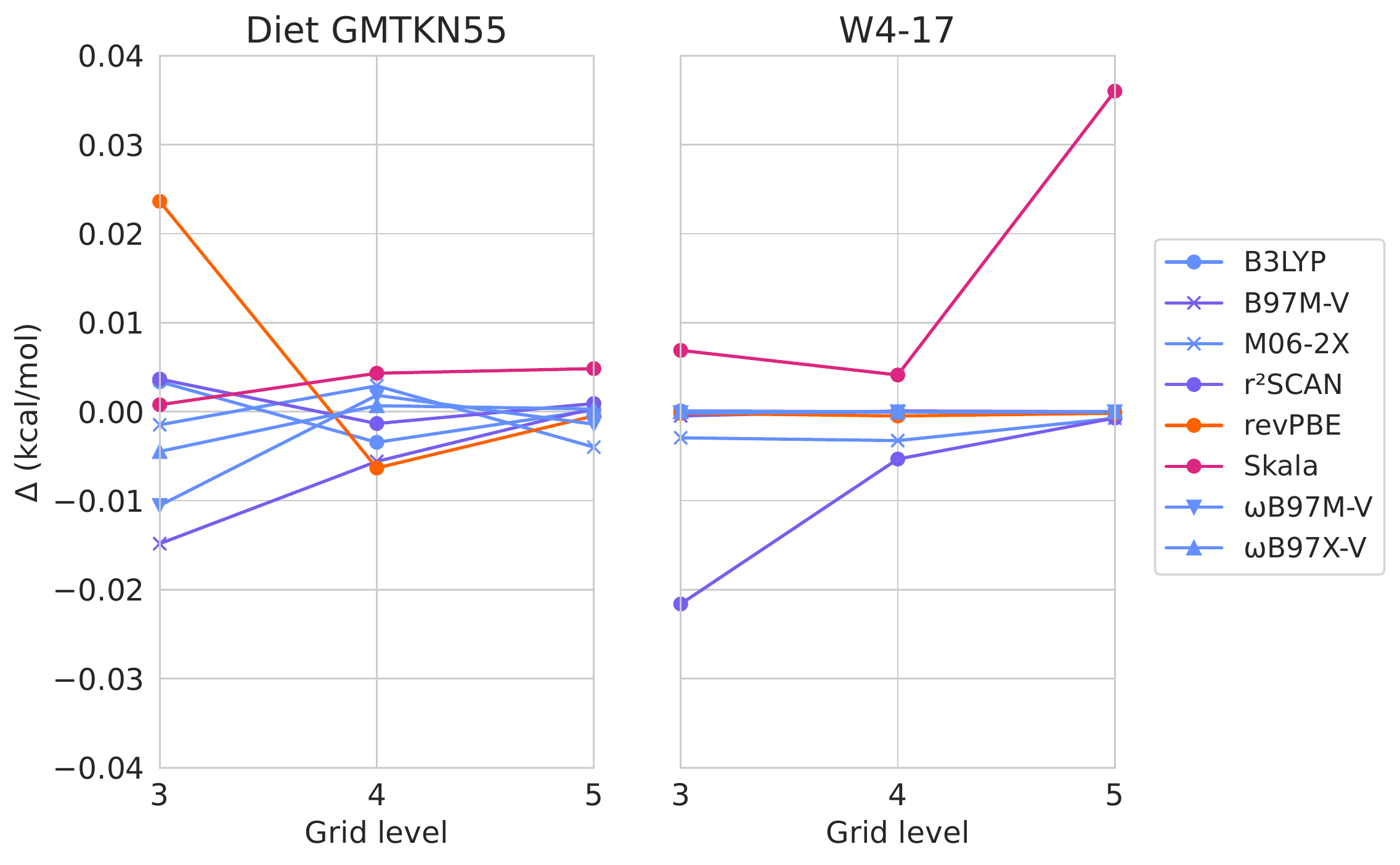}
    \caption{Evaluation on Diet GMTKN55 and W4-17 benchmarks at different pyscf grid levels (sizes). The MAE on the y-axis corresponds to the difference in reaction energies with respect to grid level 6. Reactions were included if they converged for all functionals at all grid levels with the retry logic, not including the manual step. This resulted in 199 reactions being included in W4-17 and 87 reactions being included in Diet GMTKN55.}
    \label{fig:grid-convergence}
\end{figure}

To study the effect of the DFT grid we performed benchmark calculations on W4-17 and Diet GMTKN55 using \clippy and a variety of traditional functionals, displayed in Fig.~\ref{fig:grid-convergence}. \clippyOnePointOne displays similar convergence behavior to traditional functionals on both the W4-17 and Diet GMTKN55 benchmarks.

\subsection{Emergence of exact constraints with more training data}\label{sec:exact-constraints}

\begin{figure}[t!]
    \centering
    \includegraphics[width=0.85\textwidth]{figures/tc_at_gamma1_data_ablation.pdf}
    \caption{Model $T_\cor$ (obtained from Eq.~\ref{eq:TcGKS}) of functionals trained with MSR-ACC/TAE25\cite{ehlert_accurate_2025} only and with the entire training set. 
    The area below the $x$-axis is shaded in red to indicate the violation of the positivity constraint. 
    The vertical lines represent the average across 5 independent training runs, decorated by the min and max of the derived model $T_\cor$ value across seeds. Reference data\cite{vuckovic_density_2019,vuckovic_interpolated_2017} from reverse-engineering of CCSD densities in a pure Kohn-Sham framework are also shown for comparison.
    }
    \label{fig:tc_data_ablation}
\end{figure}

Although the exact functional has an implicit form, several exact constraints that the functional satisfies can be derived.\cite{levy_hellmannfeynman_1985,kaplan_predictive_2022,lieb_density_1983,sun_strongly_2015}
Traditionally, the exact constraints are directly built into the approximate functional by restricting the form of the functional. 
With a neural network functional, it is harder to incorporate some of these constraints without limiting the model's expressivity.
But with enough data, the model could learn to satisfy the constraints as an emergent behavior.

Some of the most important constraints are derived using uniform coordinate scaling, in which a scaled density is defined as
\begin{align}
    \dens_\gamma(r) = \gamma^3\dens(\gamma r)
\end{align}
where the $\gamma^3$ factor ensures the scaled density integrates to the same number of electrons for different scaling factors $\gamma>0$. 
  The constraints related to uniform coordinate scaling for $E_{\xc}[\dens_{\gamma}]$ can be checked by evaluating a given XC model on scaled features using the following assignments
\begin{subequations}
\begin{align}
    \rho_{\gamma, i} &\leftarrow \gamma^3 \rho_i \\
    \norm{\grad\rho_{\gamma, i}} &\leftarrow \gamma^4 \norm{\grad\rho_i} \\
    \tau_{\gamma, i} &\leftarrow \gamma^5 \tau_i \label{eq:tauscaling} \\
    w_{\gamma, i} &\leftarrow \gamma^{-3} w_i \\
    r_{\gamma,i} &\leftarrow \gamma^{-1} r_i
\end{align}
\end{subequations}
where the unscaled counterparts are from a self-consistent calculation.  We recall that $r_i$ and $w_i$ are the points and the weights of the integration grids.

However, a subtle point we need to mention is that the scaling of $\tau$ of Eq.~\eqref{eq:tauscaling} is exact only in a pure KS setting, where the exact functional is partitioned as (with $U[\rho]$ the Hartree electrostatic energy defined in Eq.~\eqref{eq:TotEnergyDFT})
\begin{equation}
F[\rho]=\underbrace{\min_{\Phi\to\rho}\langle\Phi|\hat{T}|\Phi\rangle}_{T_s[\rho]}+U[\rho]+E_{\xc}[\rho].
\end{equation}
In this case, the KS orbitals have the simple scaling \(\gamma^{3/2}\phi_i(\gamma r)\) leading to the scaling of $\tau$ in Eq.~\eqref{eq:tauscaling} and the scaling of the KS kinetic energy $T_s[\rho_\gamma]=\gamma^2T_s[\rho]$. However, functionals like \clippy that contain the meta-GGA feature $\tau$ are usually evaluated within the so-called generalized Kohn-Sham scheme (GKS), in which the universal functional is implicitly partitioned as 
\begin{equation}
F[\rho]=\underbrace{\min_{\Phi\to\rho}\left(\langle\Phi|\hat{T}|\Phi\rangle+E_\xc[\rho,\tau_\Phi]\right)}_{E_{\rm Txc}[\rho]=T_s^{\rm GKS}[\rho]+E_\xc[\rho,\tau^{\rm GKS}[\rho]]}{}+U[\rho]. \label{eq:metaKS}
\end{equation}
Otherwise said, only the sum $T_s+E_\xc$ is well defined when we apply GKS to a meta-GGA. Hand-crafted meta-GGA functionals usually impose scaling constraints by assuming the functional would be evaluated within pure Kohn-Sham, making use of Eq.~\eqref{eq:tauscaling}. In the case of \clippy, where we want to check {\em a posteriori} whether constraints are satisfied, however, we need to pay attention to how the SCF $\tau$ was obtained. 
In fact, Eq.~\eqref{eq:metaKS} implies that the meta KS orbitals entering the minimizing Slater determinant $\Phi$ do not have anymore a well-defined scaling, which, in turn, means that $\tau^{\rm GKS}[\dens_\gamma]$ for the scaled density $\rho_\gamma$ is not known.  
An exception is the physical case $\gamma=1$, where the derivative with respect to the scaling parameter $\gamma$ of $E_{\rm Txc}[\dens_\gamma]$ can be obtained exactly using $\tau_\gamma$ from the usual scaling of \eqref{eq:tauscaling},
\begin{align} \label{eq:equalityDerivGamma1}
\frac{\ud }{\ud\gamma}\frac{E_{\rm Txc}[\dens_\gamma]}{\gamma}\Big|_{\gamma=1}=T_s^{\rm GKS}[\rho]+\frac{\ud }{\ud\gamma}\frac{E_\xc[\dens_\gamma,\tau_\gamma]}{\gamma}\Big|_{\gamma=1}.
\end{align}
This is due to the fact that, by definition, $E_{\rm Txc}[\dens_\gamma]\le \gamma^2T_s^{\rm GKS}[\rho]+E_\xc[\dens_\gamma,\tau_\gamma]$ for all $\gamma>0$, with equality at $\gamma=1$, implying that also their derivatives must be equal at $\gamma=1$. This will remain true if we divide both sides by $\gamma$.
We can therefore use the equation\cite{levy_hellmannfeynman_1985} that gives the physical, interacting, kinetic energy for a given scaled density $\dens_\gamma$,
\begin{align}
    T[\rho_\gamma]=\gamma^2 \frac{\ud }{\ud\gamma}\frac{E_{\rm Txc}[\dens_\gamma]}{\gamma},
\end{align}
and combine it with Eq.~\eqref{eq:equalityDerivGamma1} to obtain the GKS correlation kinetic energy at $\gamma=1$,
\begin{align}
    T_\cor^{\rm GKS}[\rho]=T[\rho]-T_s^{\rm GKS}[\rho]=\frac{\ud }{\ud\gamma}\frac{E_\xc[\dens_\gamma,\tau_\gamma]}{\gamma}\Big|_{\gamma=1}.\label{eq:TcGKS}
\end{align}
Notice that $T_\cor^{\rm GKS}[\rho]$ has the freedom to become negative: in principle the GKS system could have kinetic energy higher than the physical one, compensated by a more negative $E_\xc$, as it happens when we train on the MSR-ACC/TAE25 dataset\cite{ehlert_accurate_2025} only. Strictly speaking, this would not be wrong, but it is a less physical and convenient way to obtain accurate energies. It is remarkable, as shown in \cref{fig:tc_data_ablation_main}, that \clippy learns from data to keep $T_\cor^{\rm GKS}[\rho]$ positive and quite close to the one from pure KS. In Fig.~\ref{fig:tc_data_ablation} we show that this happens consistently across 5 different seeds. Notice that it should hold $ T_\cor^{\rm GKS}[\rho]\le T_\cor[\rho]$ if the densities were exactly the same, which is not the case here.

\subsection{Ablation of the performance on geometry optimization}\label{sec:geometry-ablation}
In this section, additional results are presented on the performance of \clippy on geometry optimization tasks.
First, the GEO metric introduced by Vuckovic and Burke \cite{vuckovic_quantifying_2020} is evaluated for the datasets and functionals displayed in Table~\ref{tab:geometry} of the main text.
The GEO metric is defined as the energy difference between evaluating a given functional on the reference and its own optimal geometry:
\begin{equation}
    \text{GEO} = E_\text{F}(R_\text{ref}) - E_\text{F}(R_\text{F}) \,
\end{equation}
where $E_F(R_\text{ref})$ and $E_F(R_F)$ denote the energy yielded by functional $F$ when evaluated on the reference or its own optimal molecular geometry, respectively.
This metric is therefore a direct measure of the additional error one makes when one uses $F$ not only for single-point calculations on reference geometries, but also for optimizing the geometries of molecules of interest.
As can be seen in Fig.~\ref{fig:geometry-ablation-geo}, the performance of \clippyOnePointOne on the W4-11-GEOM test set according to the GEO metric is on par with state-of-the-art hybrid functionals such as $\omega$B97M-V.

To elucidate the reasons for this excellent accuracy, ablations were performed on the composition of the data used to train \clippyOnePointOne.
In particular, both the right side of Fig.~\ref{fig:geometry-ablation-geo} and Table~\ref{tab:geometry-ablation} show the effect of including TAEs of molecular geometries distorted along vibrational modes (MSR-ACC/Distortion) in the training data.
Without this data, performance on bond lengths, angles and the GEO metric hovers around that of the revPBE GGA functional, while after adding this data, the performance is boosted to its current level.
This finding supports our intuition that training on distorted molecular geometries is of great importance in achieving accurate geometry optimization.

Taken together, the results in this section indicate that \clippyOnePointOne is expected to perform well in settings where reference molecular geometries are not available and one must rely on the functional for optimizing geometries.
\begin{figure}
    \centering
    \begin{minipage}[b]{0.8\linewidth}
    \centering
    \subfloat[The performance of various functionals in the GEO metric on the W4-11-GEOM dataset.]{
        \centering
        \includegraphics[width=\linewidth]{figures/si-w4-11-geom-geo.pdf}
        \label{fig:geometry-ablation-geo}
    }
    \end{minipage}
    \vspace{2em}\\
    \begin{minipage}[b]{\linewidth}
        \centering
        \include{tables/geometry_vib_ablation}
        \subcaption{Comparing geometry optimization results between checkpoints trained on data including/excluding the TAEs of vibrationally distorted geometries.}
        \label{tab:geometry-ablation}
    \end{minipage}
    \vspace{0.5em}\\
    \caption{Results on the ablation of geometry optimization performance. (a): Accuracy of molecular geometries optimized with various functionals according to the GEO metric\cite{vuckovic_quantifying_2020} on the W4-11-GEOM dataset\cite{spackman_basis_2016}.
    The right side of the plot shows the effect of including TAEs of molecular geometries distorted along vibrational modes (MSR-ACC/Distortion) in the training data of \clippy, displaying results obtained with three different seeds for each setting.
    (b): Accuracy of geometries in the benchmark datasets LMGB35, HMGB11\cite{grimme_consistent_2015}, CCse21\cite{piccardo_semiexperimental_2015} and W4-11-GEOM\cite{spackman_basis_2016} as optimized with different \clippy checkpoints, either with or without the inclusion of TAEs of vibrationally distorted molecular geometries (MSR-ACC/Distortion) in the training data, using three different seeds each.
    The table shows average absolute errors for bond lengths (in \r{A}ngstrom) and bond angles (in degrees) compared to the ground truth values from these datasets. Box plots show the quartiles of the error distribution.
    }
    \label{fig:geometry-ablation}
\end{figure}

\clearpage
\printbibliography[resetnumbers=\resetnum, title=Supplementary References]
\numgdef\resetnum{\csuse{blx@labelnumber@\therefsection}+1}
\end{refsection}

\end{document}